\documentclass[preprint,nofootinbib,amsmath,amssymb,amsfonts,aps,prd]{revtex4-1}

\usepackage{hyperref}
\usepackage{amsthm}
\usepackage{mathrsfs}

\DeclareMathOperator*{\wlim}{w-lim}
\DeclareMathOperator*{\supp}{supp}
\newcommand{\ud}{\,\mathrm{d}}
\newcommand{\hh}{\mathfrak{h}}
\newtheorem*{lemma}{Lemma}
\newtheorem{theorem}{Theorem}

\begin{document}

\title{A new framework for analyzing the effects of small scale
  inhomogeneities in cosmology}

\author{Stephen R. Green}
\email{srgreen@uchicago.edu}
\author{Robert M. Wald}
\email{rmwa@uchicago.edu}
\affiliation{Enrico Fermi Institute and Department of Physics \\
  University of Chicago \\
  5640 S. Ellis Ave., Chicago, IL 60637, U.S.A.}

\date{\today}

\begin{abstract}
  We develop a new, mathematically precise framework for treating the
  effects of nonlinear phenomena occurring on small scales in general
  relativity. Our approach is an adaptation of Burnett's formulation
  of the ``shortwave approximation'', which we generalize to analyze
  the effects of matter inhomogeneities as well as gravitational
  radiation. Our framework requires the metric to be close to a
  ``background metric'', but allows arbitrarily large stress-energy
  fluctuations on small scales. We prove that, within our framework,
  if the matter stress-energy tensor satisfies the weak energy
  condition (i.e., positivity of energy density in all frames), then
  the only effect that small scale inhomogeneities can have on the
  dynamics of the background metric is to provide an ``effective
  stress-energy tensor'' that is traceless and has positive energy
  density---corresponding to the presence of gravitational
  radiation. In particular, nonlinear effects produced by small scale
  inhomogeneities cannot mimic the effects of dark energy. We also
  develop ``perturbation theory'' off of the background metric. We
  derive an equation for the ``long-wavelength part'' of the leading
  order deviation of the metric from the background metric, which
  contains the usual terms occurring in linearized perturbation theory
  plus additional contributions from the small-scale
  inhomogeneities. Under various assumptions concerning the absence of
  gravitational radiation and the non-relativistic behavior of the
  matter, we argue that the ``short wavelength'' deviations of the
  metric from the background metric near a point $x$ should be
  accurately described by Newtonian gravity, taking into account only
  the matter lying within a ``homogeneity lengthscale'' of
  $x$. Finally, we argue that our framework should provide an accurate
  description of the actual universe.
\end{abstract}

\maketitle

\section{Introduction}

It is generally believed that our universe is very well described on
large scales by a Friedmann-Lema\^itre-Robertson-Walker (FLRW)
model. However, on small scales, extremely large departures of the
mass density from FLRW models are commonly observed, e.g., on
Earth\footnote{If we go to scales of atomic nuclei, then $\delta
  \rho/\rho \sim 10^{43}$.}, we have $\delta \rho/\rho \sim
10^{30}$. Nevertheless, common sense estimates
\cite{Holz:1997ic,Ishibashi:2005sj} suggest that (a) the deviation of
the metric (as opposed to mass density, which corresponds to second
derivatives of the metric) from a FLRW metric is globally very small
on all scales except in the immediate vicinity of strong field objects
such as black holes and neutron stars, and (b) the terms in Einstein's
equation that are nonlinear in the deviation of the metric from a FLRW
metric are negligibly small as compared with the linear terms in the
deviation from a FLRW metric except in the immediate vicinity of
strong field objects. These common sense estimates together with the
fact that the motion of matter relative to the rest frame of the
cosmic microwave background is non-relativistic strongly suggest that
(1) the large scale structure of the universe is well described by a
FLRW metric, (2) when averaged on scales sufficiently large that
$|\delta \rho/\rho| \ll 1$---i.e., scales of order $100$~Mpc
in the present universe---the deviations from a FLRW model are well
described by ordinary FLRW linear perturbation theory, and (3) on
smaller scales, the deviations from a FLRW model (or, for that matter,
from Minkowski spacetime) are well described by Newtonian
gravity---except, of course, in the immediate vicinity of strong field
objects.

The above assumptions underlie the standard cosmological model, which
has been remarkably successful in accounting for essentially all
cosmological phenomena. Thus, there is good empirical evidence that
assumptions (1)--(3) are at least essentially correct. Nevertheless,
the situation is quite unsatisfactory from the perspective of having a
mathematically consistent theory wherein the assumptions and
approximations are justified in a systematic manner. Indeed, it is not
even obvious that assumptions (1)--(3) are mathematically consistent
\cite{vanElst:1998kb,Rasanen:2010wz}. In particular, nonlinear effects
play an essential role in Newtonian dynamics, e.g., the fact that the
Earth orbits the Sun arises from Einstein's equation as a nonlinear
effect in the deviation of the metric from flatness. It is clear that
one would get an extremely poor description of small-scale structure
in the universe if one neglected the nonlinear terms in Einstein's
equation in the deviation of the metric from a FLRW model; for
example, galaxies would not be bound. But if one cannot neglect
nonlinear terms in Einstein's equation on small scales, how can one
justify neglecting them on large (i.e., $\sim 100$~Mpc or larger)
scales? In addition, since it is not clear exactly what approximations
are needed for assumptions (1)--(3) to be valid, it is far from clear
as to how one could go about systematically improving these
approximations.

Indeed, it is far from obvious, {\it a priori}, that nonlinearities
associated with small-scale inhomogeneities could not produce
important effects on the large-scale dynamics of the FLRW model
itself, as has been suggested by a number of authors
\cite{Buchert:2002ij,Coley:2005ei,Buchert:2005xf,Ellis:2005uz,Nambu:2005zn,Kolb:2005da,Rasanen:2006kp,Buchert:2007ik,Kolb:2008bn,Rasanen:2008it,Zalaletdinov:2008ts,Kolb:2009rp,Wiegand:2010uh}
as a possible way to account for the effects of ``dark energy''
without invoking a cosmological constant, a new source of matter, or a
modification of Einstein's equation. In fact, the example of
gravitational radiation of wavelength much less than the Hubble scale
illustrates that it is possible, in principle, for small-scale
inhomogeneities in the metric and curvature to affect large-scale
dynamics. The dynamics of a FLRW model whose energy content is
dominated by gravitational radiation will be very different from one
with a similar matter content but no gravitational radiation. It is
the nonlinear terms in Einstein's equation associated with the
short-wavelength gravitational radiation that are responsible for
producing this difference in the large-scale dynamics. Although
common-sense estimates indicate that similar effects on large-scale
dynamics should not be produced by nonlinear effects of small-scale
matter inhomogeneities in our universe, it would be very useful to
have a systematic and general approach that can determine exactly what
effects small-scale inhomogeneities can and cannot produce on
large-scale dynamics.

The main approach that has been taken to investigate the effects of
small-scale inhomogeneities on large-scale dynamics has been to
consider inhomogeneous models, take spatial averages to define
corresponding FLRW quantities, and derive equations of motion for
these FLRW quantities \cite{Buchert:1999er,Buchert:2001sa}. Since, in
particular, the spatial average of the square of a quantity does not
equal the square of its spatial average, the effective FLRW dynamics
of an inhomogeneous universe will differ from that of a homogeneous
universe. However, a major difficulty with this approach is that, when
the deviations of the metric from that of a FLRW background are not
very small, it is not obvious how to interpret the averaged quantities
in terms of observable quantities. For example, if the total volume of
a spatial region is found to increase with time, this certainly does
not imply that observers in this region will find that Hubble's law
appears to be satisfied. Further serious difficulties with this
approach arise from the fact that the notion of averaging is slicing
dependent and the average of tensor quantities over a region in a
non-flat spacetime is intrinsically ill defined. In addition, the
equations for averaged quantities that have been derived to date are
only a partial set of equations---they contain quantities whose
evolution is not determined---so it is difficult to analyze what
dynamical behavior of the averaged quantities is actually
possible. This difficulty is well illustrated by a recent paper of
Buchert and Obadia \cite{Buchert:2010ug}, where they suggest that
inflationary dynamics may be possible in vacuum spacetimes. However,
this conclusion is drawn by simply postulating that a particular
functional relation holds between certain averaged quantities under
dynamical evolution. In fact, Einstein's equation controls the
dynamical evolution of these quantities---so one is not free to
postulate additional relations---but the restrictions imposed by
Einstein's equation are not considered.

The main purpose of this paper is to develop a framework that allows
us to consider spacetimes where there can be significant inhomogeneity
and nonlinear dynamics on small scales, yet the framework\footnote{Our
  framework will have some significant similarities to the approach of
  \cite{Futamase:1996fk} (see also \cite{Zalaletdinov:1996aj}), but
  our assumptions will be considerably more general and our results
  will have considerably wider applicability.  Our assumptions also
  will be stated much more precisely.  In addition, we will develop
  perturbation theory within our framework.} is capable of describing
``average'' large-scale behavior in a mathematically precise
manner. We seek a framework wherein the approximations are
``controlled'' in the sense that they can be shown to hold with
arbitrarily good accuracy in some appropriate limit. The results
obtained within this framework will thereby be theorems, and the only
issue that can arise with regard to the applicability of these results
to the physical universe is how close the physical universe is to the
limiting behavior of the theorems, in which the results hold exactly.

The situation that we wish to describe via our framework is one in
which there is a ``background spacetime metric'', $g^{(0)}_{ab}$, that
is supposed to correspond to the metric ``averaged'' over small scale
inhomogeneities. In the case of interest in cosmology, $g^{(0)}_{ab}$
would be taken to be a metric with FLRW symmetry, but our framework
does not require this choice, and no restrictions will be placed upon
$g^{(0)}_{ab}$ until section IV. The difference, $h_{ab} \equiv g_{ab} -
g^{(0)}_{ab}$, between the actual metric $g_{ab}$ and the background
metric is assumed to be small everywhere. This precludes the existence
of strong field objects such as black holes and neutron stars, but
even if such objects are present, by replacing these objects with weak
field objects of the same mass, our framework should give a good
description of the universe except in the immediate vicinity of these
objects. However, even though our framework requires $h_{ab}$ to be
small, derivatives of $h_{ab}$ (say, with respect to the derivative
operator $\nabla_a$ of the background metric $g^{(0)}_{ab}$) are {\em
  not} assumed to be small. Specifically, quadratic products of
$\nabla_c h_{ab}$ are allowed to be of the same order as the curvature
of $g^{(0)}_{ab}$. Thus, {\em a priori}, such terms are allowed to
make a significant contribution to the dynamics of $g^{(0)}_{ab}$
itself. Finally, no restrictions are placed upon second derivatives of
$h_{ab}$. In particular, if matter is present, the framework allows
$\delta \rho/\rho \gg 1$.

How can one formulate a mathematically precise framework where
approximations such as the smallness of $h_{ab} = g_{ab} -
g^{(0)}_{ab}$ are ``controlled'' in the sense that limits can be taken
where they hold with arbitrarily good accuracy? The basic idea is to
consider a one-parameter family of metrics $g_{ab}(\lambda)$ that has
appropriate limiting behavior as $\lambda \rightarrow 0$. To
illustrate this idea, consider the much simpler case of ordinary
perturbation theory, wherein one wishes to describe a situation where
not only is $g_{ab} - g^{(0)}_{ab}$ small, but all of its spacetime
derivatives are correspondingly small. To describe this in a precise
way, we can consider a one-parameter family of metrics
$g_{ab}(\lambda,x)$ that is jointly smooth in the parameter $\lambda$
and the spacetime coordinates $x$. The limit as $\lambda \rightarrow
0$ of this family of metrics clearly exists and defines the background
metric $g^{(0)}_{ab}(x) = g_{ab}(0,x)$. If we assume that
$g_{ab}(\lambda)$ satisfies Einstein's equation for all $\lambda > 0$,
it follows immediately that $g^{(0)}_{ab}$ also satisfies Einstein's
equation. The first order perturbation, $\gamma_{ab}$, of
$g^{(0)}_{ab}$ is defined to be the partial derivative of
$g_{ab}(\lambda)$ with respect to $\lambda$, evaluated at $\lambda =
0$. It satisfies the linearized Einstein equation, which is derived by
taking the partial derivative with respect to $\lambda$ of Einstein's
equation for $g_{ab}(\lambda)$, evaluated at $\lambda = 0$. More
generally, the $n$th order perturbation, $\gamma^{(n)}_{ab}$, of
$g^{(0)}_{ab}$ is defined to be the $n$th partial derivative of
$g_{ab}(\lambda)$ with respect to $\lambda$ evaluated at $\lambda =
0$, and the equation it satisfies is derived by taking the $n$th
partial derivative with respect to $\lambda$ of Einstein's
equation. The perturbative equations for the metric perturbation at
each order hold rigorously and exactly. Of course, the issue remains
as to how accurately an $n$th order Taylor series approximation in
$\lambda$ describes a particular metric $g_{ab}(\lambda)$ for some
small but finite value of $\lambda$.  This issue may not be easy to
resolve in any specific case. Nevertheless, even if the accuracy of
the Taylor approximation cannot be fully resolved, it is far more
satisfactory mathematically to derive rigorous results for the
perturbative quantities than to make crude arguments about $g_{ab}$
based on the assumption that $h_{ab} = g_{ab} - g^{(0)}_{ab}$ is
``small.''

To obtain a mathematically precise framework that can be applied to
describe situations relevant for cosmology, we also wish to consider a
one-parameter family of metrics $g_{ab}(\lambda)$ that approaches a
smooth background metric $g^{(0)}_{ab}$ as $\lambda \rightarrow 0$.
However, we do not want to require that first spacetime derivatives of
\begin{equation}
h_{ab} (\lambda) \equiv g_{ab}(\lambda) - g^{(0)}_{ab}
\end{equation}
go to zero as $\lambda \rightarrow 0$. Indeed, in order to capture the
effects we are interested in, it is essential that, {\em a priori},
the framework allow quadratic products of derivatives of
$h_{ab}(\lambda)$ to be of the same order as the curvature of
$g^{(0)}_{ab}$ in the limit as $\lambda \rightarrow 0$. This suggests
that we should consider a one-parameter family wherein, as $\lambda
\rightarrow 0$, the deviations of $g_{ab}(\lambda)$ from
$g^{(0)}_{ab}$ simultaneously become of smaller amplitude and shorter
wavelength, in such a way that first spacetime derivatives of
$h_{ab}(\lambda)$ remain bounded but do not necessarily go to zero. If
$h_{ab}(\lambda) \to 0$ as $\lambda \to 0$ but spacetime derivatives
of $h_{ab}(\lambda)$ do not go to zero, then it is easy to see that
spacetime derivatives of $h_{ab}(\lambda)$ cannot converge pointwise
(i.e., at fixed spacetime points) as $\lambda \rightarrow 0$. However,
spacetime derivatives of $h_{ab}(\lambda)$ will automatically go to
zero when suitably averaged over a spacetime region; more precisely,
their ``weak limit'' exists and vanishes. Similarly, although we
cannot require that quadratic products of first spacetime derivatives
of $h_{ab}(\lambda)$ approach a limit at fixed spacetime points as
$\lambda \rightarrow 0$, it is mathematically consistent to require
that the weak limit of these quantities exists. As we shall see in the
next section, a certain combination of weak limits of quadratic
products of first spacetime derivatives of $h_{ab}(\lambda)$ acts as
an ``effective stress energy tensor'', which affects the dynamics of
the background metric $g^{(0)}_{ab}$. In this way, the possible
effects on FLRW dynamics of small-scale inhomogeneities---which are
required to be of small amplitude in the metric but may be of
unbounded amplitude in the mass density---can be studied in a
mathematically precise manner.

In fact, the issues we confront in attempting to treat the effects of
small-scale mass density fluctuations in cosmology are very similar to
the issues arising when one attempts to treat the self-gravitating
effects of short-wavelength gravitational radiation. In the latter
case, one is interested in considering a situation where the amplitude
of the gravitational radiation relative to some background metric
$g^{(0)}_{ab}$ is small, but the ``effective stress-energy tensor'' of
the gravitational radiation---i.e., products of first spacetime
derivatives of $(g_{ab} - g^{(0)}_{ab})$---is comparable to the
curvature of $g^{(0)}_{ab}$. A ``shortwave approximation'' formalism
was developed by Isaacson \cite{Isaacson:1967zz,Isaacson:1968zza} (see
also pp.~964--966 of \cite{MTW}) to treat this situation. The
shortwave approximation was put on a rigorous mathematical footing by
Burnett \cite{Burnett:1989gp}, who derived the equations satisfied by
$g^{(0)}_{ab}$ by considering a one-parameter family of metrics
$g_{ab}(\lambda)$ with suitable limiting behavior.  In this paper, we
shall generalize Burnett's formulation of the shortwave approximation
by allowing for the presence of a nonvanishing matter stress-energy
tensor\footnote{A generalization of the Isaacson formalism to include
  matter under Newtonian assumptions was previously considered by
  Noonan \cite{Noonan:1984}.} $T_{ab}$. By following Burnett's
approach, we shall derive an equation for the ``background metric'',
$g^{(0)}_{ab}$, which takes the form of Einstein's equation with an
``averaged'' matter stress-energy tensor and an additional ``effective
stress-energy'' contribution arising from the small-scale
inhomogeneities.

One of the main results of our paper, proven in the analysis of
section II, is that if the true matter stress-energy tensor $T_{ab}$
satisfies the weak energy condition (i.e., if the energy density is
positive in all frames), then the effective stress-energy tensor
appearing in the equation for $g^{(0)}_{ab}$ must be traceless and
must have positive energy density---just as in the vacuum
case\footnote{Our result on the positivity of energy density within
  this framework is new, i.e., it was not previously shown to hold in
  the vacuum case by Burnett \cite{Burnett:1989gp}.  Positivity of the
  effective energy density in the vacuum case was shown by Isaacson
  \cite{Isaacson:1968zza} only under an additional WKB ansatz.}. In
other words, no new effects on large-scale dynamics can arise from
small-scale matter inhomogeneities; the only effects that small-scale
inhomogeneities of any kind can have on large-scale dynamics
corresponds to having gravitational radiation present. Our analysis
makes no assumptions of symmetries of the background metric
$g^{(0)}_{ab}$ and makes no assumptions about the matter stress-energy
tensor $T_{ab}$ other than that it satisfies the weak energy
condition. However, if $g^{(0)}_{ab}$ is assumed to have FLRW
symmetry, then our results establish that, within our framework, the
only effect that small-scale inhomogeneities can have on FLRW dynamics
corresponds to the additional presence of an effective $P =
\frac{1}{3}\rho$ fluid with $\rho \geq 0$. In particular, within our
framework, small-scale inhomogeneities cannot provide an effective
source of ``dark energy''.

In cosmology, in addition to analyzing the dynamics of the background
FLRW spacetime, one is interested in analyzing the deviations from the
FLRW background. In section III, we undertake a general analysis of
perturbation theory within our framework. It is not straightforward to
do this because $g_{ab}(\lambda)$ is not differentiable in $\lambda$
at $\lambda = 0$, so there is no notion of a ``linearized metric
perturbation'' in our framework. However, the weak limit as $\lambda
\rightarrow 0$ of $[g_{ab}(\lambda) - g^{(0)}_{ab}]/\lambda$ may
exist, and, under the assumption that it does, this limit defines a
quantity $\gamma^{(L)}_{ab}$, which corresponds closely to what is
called the ``long wavelength part'' of the metric perturbation in
other analyses (see, e.g., \cite{Baumann:2010tm}). We also write
\begin{equation}
  h^{(S)}_{ab}(\lambda) \equiv h_{ab}(\lambda) - \lambda \gamma^{(L)}_{ab}
\end{equation}
and refer to $h^{(S)}_{ab}$ as the ``short wavelength part'' of the
deviation of the metric from a FLRW model. (Note that in our
framework, $h^{(S)}_{ab}$ and $\gamma^{(L)}_{ab}$ have precise
mathematical definitions.)  Our goal is to derive the equations
satisfied by $\gamma^{(L)}_{ab}$ as well as to determine
$h^{(S)}_{ab}(\lambda)$ to accuracy $O(\lambda)$. This will yield the
spacetime metric to accuracy $O(\lambda)$.

In section III, we will systematically derive the equations satisfied
by $\gamma^{(L)}_{ab}$ and $h^{(S)}_{ab}(\lambda)$ in a completely
general context. By taking the weak limit of $1/\lambda$ times the
difference between the exact Einstein equation for $g_{ab}(\lambda)$
and the effective Einstein equation for $g^{(0)}_{ab}$, we obtain an
equation for $\gamma^{(L)}_{ab}$ corresponding to that arising in
ordinary linearized perturbation theory. However, in addition to the
familiar terms appearing in ordinary linearized perturbation theory,
this equation contains additional ``source terms'' arising from
$h^{(S)}_{ab}$ and, if gravitational radiation is present in the
background spacetime, this equation also contains additional terms
linear in $\gamma^{(L)}_{ab}$ and quadratic in $h^{(S)}_{ab}$. We also
obtain additional relations between quantities appearing in this
equation by taking weak limits of Einstein's equation multiplied by
$h^{(S)}_{ab}/\lambda$ and by $h^{(S)}_{ab}
h^{(S)}_{cd}/\lambda$. These relations are used to simplify the
perturbation equation for $\gamma^{(L)}_{ab}$. Finally, we write down
Einstein's equation for $h^{(S)}_{ab}(\lambda)$. In the vacuum case,
we consider the simplifications that can be made to this equation if
one is interested only in determining $h^{(S)}_{ab}(\lambda)$ to
sufficient accuracy to obtain $g^{(0)}_{ab}$. We compare our approach
to that of Isaacson \cite{Isaacson:1967zz,Isaacson:1968zza} and
subsequent works (see, e.g., \cite{MTW}).

In section IV, we apply our general perturbative analysis to
cosmology. We introduce the ``generalized wave map gauge'' in
subsection IVA. In subsection IVB, we make additional assumptions
concerning initial conditions and the Newtonian nature of the
deviation, $\delta T_{ab} \equiv T_{ab}(\lambda) - T^{(0)}_{ab}$, of
the stress-energy tensor from a FLRW model. We argue that, for small
$\lambda$, $h^{(S)}_{ab}(\lambda)$ should be well approximated (in the
wave map gauge) by the Newtonian gravity solution, whereby one needs
only take into account the ``nearby'' matter. In contrast to the rest
of the paper---where all of the assumptions are stated in a
mathematically precise manner, and all of the results are
theorems---in subsection IVB we provide only a sketch of the
assumptions needed, and many of our arguments have the character of
plausibility arguments rather than proofs.  In subsection IVC, we
simplify the equation for $\gamma^{(L)}_{ab}$ derived in section III
by using our Newtonian assumptions. We show that this equation reduces
to the ordinary cosmological perturbation equation with an additional
effective source arising from $h^{(S)}_{ab}(\lambda)$, in agreement
with a result recently obtained by \cite{Baumann:2010tm}.

In summary, in this paper we introduce a new framework for
treating spacetimes whose metric is close to that of a background
metric $g^{(0)}_{ab}$ but is such that nonlinear departures from
$g^{(0)}_{ab}$ are dynamically important on small scales. We proceed
by introducing a suitable one-parameter family of metrics
$g_{ab}(\lambda)$ and deriving results in the limit as $\lambda
\rightarrow 0$. We prove that the small-scale inhomogeneities cannot
affect the dynamics of the background metric $g^{(0)}_{ab}$ except by
the addition of an effective stress-energy tensor with positive energy
density and vanishing trace, which can be interpreted as arising from
gravitational radiation. We derive an equation for the
``long-wavelength part'' of the leading order deviation of the metric
from $g^{(0)}_{ab}$, which contains the usual terms occurring in
linearized perturbation theory plus additional contributions from the
small-scale inhomogeneities. Finally, we argue that the small-scale
deviations of the metric from $g^{(0)}_{ab}$ should be accurately
described by Newtonian gravity.

Of course, the real universe is not the limit as $\lambda$ becomes
arbitrarily small of the type of a one-parameter family of metrics
$g_{ab}(\lambda)$ considered here. Thus, our results do not apply
exactly to the real universe---any more than the results of an
analysis using ordinary linearized perturbation theory would apply to
a real situation. However, in section V we will argue that since the
scales at which nonlinear dynamics are important in the present
universe (i.e., scales much less than $\sim 100$~Mpc) are much smaller
than the scale of the background curvature (i.e., the Hubble radius
$\sim 3$~Gpc), it seems reasonable to expect that the real universe
will be accurately described by a ``small $\lambda$'' approximation to
$g_{ab}(\lambda)$ within our formalism.  We believe that our analysis
thereby goes a long way toward providing a mathematically sound
framework that can be used to justify the assumptions and
approximations used in cosmology. At the end of section V, we will
discuss how these approximations can be improved.

Our notation and sign conventions follow that of
\cite{Wald:1984}. Lower case Latin indices from early in the alphabet
($a,b,c, \dots$) denote abstract spacetime indices. Greek indices
denote components of tensors. Latin indices from mid-alphabet ($i,j,k,
\dots$) denote spatial components of tensors.

\section{Dynamics of the background metric}

In this section, we will give a precise statement of the assumptions
that underlie our framework. We will then analyze the dynamics of the
background metric $g^{(0)}_{ab}$ and prove that if the matter
stress-energy tensor, $T_{ab}$, satisfies the weak energy condition,
then the ``effective stress-energy'' contributed by small-scale
inhomogeneities must have positive energy density and vanishing trace.

As explained in the Introduction, we wish to consider a situation
wherein we have a one-parameter family of metrics $g_{ab}(\lambda)$
that approaches a background metric $g^{(0)}_{ab}$, but spacetime
derivatives of $g_{ab}(\lambda)$ do not approach the corresponding
spacetime derivatives of $g^{(0)}_{ab}$. An example of the type of
behavior that we have in mind is for components of $h_{ab}(\lambda)
\equiv g_{ab}(\lambda) - g^{(0)}_{ab}$ to behave like $\lambda
\sin(x/\lambda)$. In this situation, if we let $\nabla_a$ denote the
derivative operator associated with $g^{(0)}_{ab}$, we cannot have
$\nabla_c h_{ab}(\lambda) \rightarrow 0$ pointwise as $\lambda
\rightarrow 0$. However, suitable spacetime averages of $\nabla_c
h_{ab}(\lambda)$ will go to zero. More precisely, if $f^{cab}$ is any
smooth tensor field of compact support, we have
\begin{eqnarray}
\int f^{cab} \nabla_c h_{ab}(\lambda) &=& - \int (\nabla_c f^{cab})
h_{ab}(\lambda) \nonumber \\ &\rightarrow& 0 \quad {\rm as} \quad
\lambda \rightarrow 0
\label{wdg}
\end{eqnarray}
provided only that $h_{ab}(\lambda) \rightarrow 0$ locally in $L^1$,
where the volume element in this integral is that associated with
$g^{(0)}_{ab}$. If \eqref{wdg} holds for all ``test'' (i.e.,
smooth and compact support) tensor fields, $f^{cab}$, we say that
$\nabla_c h_{ab}(\lambda) \rightarrow 0$ {\em weakly}. More generally,
if $A_{a_1 \dots a_n} (\lambda)$ is a one-parameter family of tensor
fields defined for $\lambda > 0$, we say that $A_{a_1 \dots a_n}
(\lambda)$ converges weakly to $B_{a_1 \dots a_n}$ as $\lambda
\rightarrow 0$ if for all smooth $f^{a_1 \dots a_n}$ of compact
support, we have
\begin{equation}
\lim_{\lambda \rightarrow 0} \int f^{a_1 \dots a_n} A_{a_1 \dots a_n} (\lambda) = \int f^{a_1 \dots a_n} B_{a_1 \dots a_n}\,.
\end{equation}
Roughly speaking, the weak limit performs a local spacetime average of
$A_{a_1 \dots a_n} (\lambda)$ before letting $\lambda \rightarrow 0$.

As noted above, if $g_{ab}(\lambda)$ converges to $g^{(0)}_{ab}$ in a
suitably strong sense---locally in $L^1$ suffices---then spacetime
derivatives of $h_{ab}(\lambda)$ automatically converge weakly to
zero. However, there is no reason why products of spacetime
derivatives of $h_{ab}(\lambda)$ must converge weakly at all, and, if
they do converge, one would not expect them to converge to zero. (This
latter observation is closely related to the fact that averages of
products of quantities are not normally equal to the product of their
averages.) As discussed in the Introduction, we wish to consider a
situation where first spacetime derivatives of $h_{ab}(\lambda)$
remain bounded but do not necessarily go to zero as $\lambda
\rightarrow 0$. However, in order to have well defined averaged
behavior in the limit as $\lambda \rightarrow 0$ we want the weak
limit of the nonlinear terms in Einstein's equation to exist. As we
shall see below, this will be the case if the weak limit of quadratic
products of first spacetime derivatives of $h_{ab}(\lambda)$ exists.

All of the above considerations lead us to consider a
one-parameter family of metrics $g_{ab}(\lambda)$, defined for all
$\lambda \geq 0$, satisfying the following conditions, which are
straightforward generalizations to the non-vacuum case of the
conditions considered by Burnett \cite{Burnett:1989gp} in his
formulation of the shortwave approximation. In these conditions,
$\nabla_a$ denotes an arbitrary fixed (i.e., $\lambda$-independent)
derivative operator on the spacetime manifold $M$. For convenience in
stating these conditions, we choose an arbitrary Riemannian metric
$e_{ab}$ on $M$ and for any tensor field $t_{a_1 \dots a_n}$ on $M$ we
define $|t_{a_1 \dots a_n}|^2 = e^{a_1 b_1} \dots e^{a_n b_n} t _{a_1
  \dots a_n} t_{b_1 \dots b_n}$.
\renewcommand{\labelenumi}{(\roman{enumi})}
\begin{enumerate}
\item{\label{assumption1} For all $\lambda>0$, we have 
\begin{equation}
G_{ab}(g(\lambda)) + \Lambda g_{ab}(\lambda) = 8 \pi T_{ab}(\lambda) \,\, , 
\label{Ee}
\end{equation}
where $T_{ab}(\lambda)$ satisfies the weak energy condition, i.e., for
all $\lambda >0$ we have\\$T_{ab}(\lambda) t^a(\lambda) t^b (\lambda)
\geq 0$ for all vectors $t^a(\lambda)$ that are timelike with respect
to $g_{ab}(\lambda)$.}
\item{There exists a smooth positive function $C_1(x)$ on $M$ such that 
\begin{equation}
|h_{ab}(\lambda,x)| \leq \lambda C_1(x) \,\,,
\end{equation}}
where $h_{ab} (\lambda,x) \equiv g_{ab}(\lambda,x) - g_{ab}(0,x)$.
\item{There exists a smooth positive function $C_2(x)$ on $M$ such that 
\begin{equation}
|\nabla_c h_{ab}(\lambda,x)| \leq C_2(x) \,\, .
\end{equation}}
\item{There exists a smooth tensor field $\mu_{abcdef}$ on $M$ such that 
\begin{equation}
\wlim_{\lambda \rightarrow 0}
\left[\nabla_ah_{cd}(\lambda)\nabla_bh_{ef}(\lambda) \right] =
\mu_{abcdef} \,\, ,
\label{mu}
\end{equation}
where ``$\wlim$'' denotes the weak limit.}
\end{enumerate}
\renewcommand{\labelenumi}{\theenumi} It follows immediately that
$\mu_{ab(cd)(ef)}=\mu_{abcdef}$ and $\mu_{abcdef}=\mu_{baefcd}$, and
it is not difficult to show \cite{Burnett:1989gp} that
$\mu_{(ab)cdef}=\mu_{abcdef}$.  It also is not difficult to see that
if $g_{ab}(\lambda)$ satisfies the above conditions for any choice of
$\nabla_a$ and $e_{ab}$, then it satisfies these conditions for all
choices of $\nabla_a$ and $e_{ab}$. In our calculations, it will be
convenient to choose $\nabla_a$ to be the derivative operator
associated with the background metric $g^{(0)}_{ab} \equiv g_{ab}(0)$,
and in the following, we shall make this choice. We shall also raise
and lower indices with $g^{(0)}_{ab}$.

As discussed in the Introduction, the key idea is that our actual
spacetime, with all of its inhomogeneities, is described by an element
of such a one-parameter family, at some small but finite value of
$\lambda$.  By analyzing the limiting behavior of such one-parameter
families at small $\lambda$, we hope to attain an excellent
approximate description of our universe. However, unlike ordinary
perturbative analyses, our one-parameter family $g_{ab}(\lambda)$ is
not differentiable in $\lambda$ at $\lambda =0$, so we cannot define
perturbative quantities or obtain useful equations by differentiation with
respect to $\lambda$.

Our first task is to derive an equation satisfied by the background
metric $g^{(0)}_{ab} = g_{ab}(0)$. This equation will follow directly
from Einstein's equation \eqref{Ee} for $g_{ab}(\lambda)$, using the
general relationship between the Ricci curvature of $g^{(0)}_{ab}$ and
$g_{ab}(\lambda)$, namely
\begin{equation}\label{eqn:RicciRelation}
  R_{ab}(g^{(0)})=R_{ab}(g(\lambda))+2\nabla_{[a}C^c_{\phantom{c}c]b}-2C^d_{\phantom{d}b[a}C^c_{\phantom{c}c]d}\,,
\end{equation}
where 
\begin{equation}
  C^c_{\phantom{c}ab}=\frac{1}{2}g^{cd}(\lambda)\left\{\nabla_ag_{bd}(\lambda)+\nabla_bg_{ad}(\lambda)-\nabla_dg_{ab}(\lambda)\right\}\,
\label{Ccab}
\end{equation}
and, again, we remind the reader that $\nabla_a$ denotes the
derivative operator associated with $g_{ab}^{(0)}$, so that $\nabla_c
g_{ab}^{(0)} = 0$.  It follows immediately from \eqref{eqn:RicciRelation}
that
\begin{eqnarray}
  R_{ab}(g^{(0)})&-&\frac{1}{2}g_{ab}(\lambda)g^{cd}(\lambda)R_{cd}(g^{(0)})+\Lambda g_{ab}(\lambda)\nonumber\\
  &=&G_{ab}(g(\lambda)) +\Lambda g_{ab}(\lambda)+2\nabla_{[a}C^e_{\phantom{e}e]b}-2C^f_{\phantom{f}b[a}C^e_{\phantom{e}e]f}\nonumber\\
  &&-g_{ab}(\lambda)g^{cd}(\lambda)\nabla_{[c}C^e_{\phantom{e}e]d}+g_{ab}(\lambda)g^{cd}(\lambda)C^f_{\phantom{f}d[c}C^e_{\phantom{e}e]f}\,,
\end{eqnarray}
and invoking the Einstein equation for $\lambda>0$,
\begin{eqnarray}\label{eqn:EinsteinRelation}
  R_{ab}(g^{(0)})&-&\frac{1}{2}g_{ab}(\lambda)g^{cd}(\lambda)R_{cd}(g^{(0)})+\Lambda g_{ab}(\lambda)\nonumber\\
  &=& 8\pi T_{ab}(\lambda)+2\nabla_{[a}C^e_{\phantom{e}e]b}-2C^f_{\phantom{f}b[a}C^e_{\phantom{e}e]f}\nonumber\\
  &&-g_{ab}(\lambda)g^{cd}(\lambda)\nabla_{[c}C^e_{\phantom{e}e]d}+g_{ab}(\lambda)g^{cd}(\lambda)C^f_{\phantom{f}d[c}C^e_{\phantom{e}e]f}\,.
\end{eqnarray}
We now take the weak limit of both sides of
\eqref{eqn:EinsteinRelation} as $\lambda \rightarrow 0$. It is easy to
see that the weak limit of the left side exists and is equal to
$G_{ab}(g^{(0)}) + \Lambda g^{(0)}_{ab}$. Aside from the term $8\pi
T_{ab} (\lambda)$, the terms on the right side of
(\ref{eqn:EinsteinRelation}) all contain a total of precisely two
derivatives of $h_{ab}(\lambda)$. These terms can be classified into
the following types: (a) terms linear in $h_{ab}(\lambda)$,
corresponding to the linearized Einstein operator acting on
$h_{ab}(\lambda)$; (b) terms quadratic in $h_{ab}(\lambda)$,
corresponding to the second order Einstein operator acting on
$h_{ab}(\lambda)$; and (c) terms cubic and higher order in
$h_{ab}(\lambda)$. The weak limit of terms of type (a) vanish by the
type of argument leading to (\ref{wdg}). The terms of type (b) depend
upon $h_{ab}(\lambda)$ either in the form
$\nabla_ah_{cd}(\lambda)\nabla_bh_{ef}(\lambda)$---which has weak
limit $\mu_{abcdef}$---or in the form $h_{cd}(\lambda)\nabla_a
\nabla_bh_{ef}(\lambda)$. However, since
\begin{equation}
h_{cd}(\lambda)\nabla_a \nabla_bh_{ef}(\lambda) = \nabla_a
\left[h_{cd}(\lambda)\nabla_bh_{ef}(\lambda) \right] -
\nabla_ah_{cd}(\lambda)\nabla_bh_{ef}(\lambda)
\end{equation}
and the weak limit of $\nabla_a
\left[h_{cd}(\lambda)\nabla_bh_{ef}(\lambda) \right]$ vanishes, we see
that the weak limit of $h_{cd}(\lambda)\nabla_a
\nabla_bh_{ef}(\lambda)$ also exists and is equal to
$-\mu_{abcdef}$. Finally, it is easily seen that the weak limit of all
terms of type (c) vanish.

Since the weak limit of all terms in (\ref{eqn:EinsteinRelation})
apart from $T_{ab} (\lambda)$ exist, it follows that the weak limit of
$T_{ab} (\lambda)$ itself also must exist, without the necessity to
impose any additional assumptions on our one-parameter family. We write
\begin{equation}
T^{(0)}_{ab} \equiv \wlim_{\lambda\to 0}T_{ab}(\lambda) \, ,
\end{equation}
and we may interpret $T^{(0)}_{ab}$ as representing the matter
stress-energy tensor averaged over small scale inhomogeneities.  Since
$T_{ab}(\lambda)$ satisfies the weak energy condition for all
$\lambda>0$ and since $g_{ab}(\lambda)$ converges uniformly (on
compact sets) to $g_{ab}^{(0)}$, it is not difficult to show that
$T^{(0)}_{ab}$ also satisfies the weak energy condition, i.e.,
$T^{(0)}_{ab} t^a t^b \geq 0$ for all timelike vectors $t^a$ with
respect to $g^{(0)}_{ab}$.  The weak limit of
(\ref{eqn:EinsteinRelation}) then takes the form
\begin{equation}
  G_{ab}(g^{(0)}) + \Lambda g^{(0)}_{ab} = 8\pi T^{(0)}_{ab} + 8\pi t^{(0)}_{ab}\,,
\label{Eeg0}
\end{equation}
where the ``effective gravitational stress-energy tensor''
$t^{(0)}_{ab}$ arises from the weak limit of terms of type (b) above
and can be expressed entirely in terms of $\mu_{abcdef}$. A lengthy
calculation (see \cite{Burnett:1989gp}) yields
\begin{eqnarray}\label{eqn:BackgroundEinstein-messy}
  8\pi t^{(0)}_{ab} &=&\frac{1}{8}\left\{-\mu^{c\phantom{c}de}_{\phantom{c}c\phantom{de}de}-\mu^{c\phantom{c}d\phantom{d}e}_{\phantom{c}c\phantom{d}d\phantom{e}e}+2\mu^{cd\phantom{c}e}_{\phantom{cd}c\phantom{e}de}\right\}+\frac{1}{2}\mu^{cd}_{\phantom{cd}acbd}-\frac{1}{2}\mu^{c\phantom{ca}d}_{\phantom{c}ca\phantom{d}bd}\nonumber\\
  &&+\frac{1}{4}\mu^{\phantom{ab}cd}_{ab\phantom{cd}cd}-\frac{1}{2}\mu^{c\phantom{(ab)c}d}_{\phantom{c}(ab)c\phantom{d}d}+\frac{3}{4}\mu^{c\phantom{cab}d}_{\phantom{c}cab\phantom{d}d}-\frac{1}{2}\mu^{cd}_{\phantom{cd}abcd}\,.
\end{eqnarray}
Note that $t^{(0)}_{ab}$ corresponds to the ``Isaacson average'' of
the second order Einstein tensor of $h_{ab}(\lambda)$. It can be shown
to be gauge invariant \cite{Burnett:1989gp}.

Following Burnett \cite{Burnett:1989gp}, we decompose $\mu_{abcdef}$
into two tensors $\alpha_{abcdef}\equiv\mu_{[c|[ab]|d]ef}$ and
$\beta_{abcdef}\equiv\mu_{(abcd)ef}$, so that
\begin{eqnarray}
  \mu_{abcdef}&=&-\frac{4}{3}(\alpha_{c(ab)def}+\alpha_{e(ab)fcd}-\alpha_{e(cd)fab})\nonumber\\
&&+\beta_{abcdef}+\beta_{abefcd}-\beta_{cdefab}\,.
\end{eqnarray}
Then, we have
\begin{equation}\label{eqn:BackgroundEinstein-messy2}
8\pi t^{(0)}_{ab} = \alpha^{\phantom{a}c\phantom{b}d}_{a\phantom{c}b\phantom{d}cd}+\frac{3}{2}\alpha^{cd}_{\phantom{cd}cdab}-2\alpha^{\phantom{(a}cd}_{(a\phantom{cd}|cd|b)}-\frac{1}{4}g^{(0)}_{ab}\left\{\alpha^{cd\phantom{cd}e}_{\phantom{cd}cd\phantom{e}e}-2\alpha^{cde}_{\phantom{cde}cde}\right\}\,.
\end{equation}
Note that the right hand side is independent of
$\beta_{abcdef}$. 

The remainder of this section will be devoted to establishing two key
properties of $t^{(0)}_{ab}$, namely, ${t^{(0)a}}_a = 0$ and
$t^{(0)}_{ab} t^a t^b \geq 0$ for any timelike vector $t^a$ of the
background metric $g^{(0)}_{ab}$. To prove these results, we need to
obtain further information about $\alpha_{abcdef}$ from Einstein's
equation \eqref{eqn:EinsteinRelation}. To do so, it is convenient to
re-write this equation in ``Ricci form'' as
\begin{equation}\label{eqn:RicciRelation2}
R_{ab}(g^{(0)}) - \Lambda g_{ab}(\lambda) =8\pi T_{ab}(\lambda)-4\pi g_{ab}(\lambda)g^{cd}(\lambda)T_{cd}(\lambda)+2\nabla_{[a}C^c_{\phantom{c}c]b}-2C^d_{\phantom{d}b[a}C^c_{\phantom{c}c]d}\,.
\end{equation}
We now multiply this equation by $h_{ef}(\lambda)$ and take the weak
limit as $\lambda \rightarrow 0$. The left side clearly goes to
zero. The weak limit of $h_{ef}(\lambda)
C^d_{\phantom{d}b[a}C^c_{\phantom{\ c}c]f}$ is also easily seen to vanish. 
On the other hand, we have
\begin{eqnarray}
\wlim_{\lambda\to 0} h_{ef}(\lambda) \nabla_{[a}C^c_{\phantom{c}c]b} &=&
-\wlim_{\lambda\to 0} \nabla_{[a} h_{|ef|}(\lambda) C^c_{\phantom{c}c]b} \nonumber \\
&=& - \mu_{[b|[ac]|d]ef}g^{(0)cd} \nonumber \\
&=& - \alpha_{acb\phantom{c}ef}^{\phantom{acb}c}\,.
\end{eqnarray}
Thus, we obtain
\begin{equation}
  \alpha_{a\phantom{c}bcef}^{\phantom{a}c}=4\pi\wlim_{\lambda\to 0}h_{ef}(\lambda)\left(T_{ab}(\lambda)-\frac{1}{2}g_{ab}(\lambda)g^{cd}(\lambda)T_{cd}(\lambda)\right) \,.
\label{aht}
\end{equation}
In particular, the weak limit of the right side of this equation must
exist.  By similar arguments starting with Einstein's equation in the
form \eqref{eqn:EinsteinRelation}, it also follows that the weak limit
as $\lambda \to 0$ of $h_{ef}(\lambda) T_{ab} (\lambda)$ exists.
We write
\begin{equation}
\kappa_{efab} = \wlim_{\lambda\to 0}h_{ef}(\lambda) T_{ab}(\lambda)\,.
\label{rho}
\end{equation}

We will now show that the right side of \eqref{aht} vanishes. We first
prove the following lemma.
\begin{lemma}Let $A(\lambda)$ be a one-parameter family of smooth
  tensor fields (with indices suppressed) converging uniformly on
  compact sets to $A(0)$, and let $B(\lambda)$ be a one-parameter
  family of non-negative smooth functions converging weakly to $B(0)$.
  Then $A(\lambda)B(\lambda)\to A(0)B(0)$ weakly as $\lambda \to 0$.
\end{lemma}
\begin{proof}
  Let $F$ be a test tensor field---i.e., a smooth tensor field of
  compact support---with index structure dual to that of $A$. Let $f$
  be a smooth, non-negative function of compact support with $f=1$ on
  the support of $F$. Then $F = fF$. We have
  \begin{align}\label{eqn:lemmaproof}
    &\lim_{\lambda\to 0}\int(A(\lambda)B(\lambda)- A(0)B(0))F
\nonumber\\
    &\qquad=\lim_{\lambda\to
      0}\int \left[(A(\lambda)-A(0))B(\lambda)+A(0)(B(\lambda)
-B(0)) \right]F\,,
  \end{align}
  where contraction of the indices of $A$ and $F$ is understood.  The
  second term is zero because $B(\lambda)\to B(0)$ weakly and $A(0)F$
  is a test function.  On the other hand, we have
  \begin{eqnarray}
    \left|\int(A(\lambda)-A(0))B(\lambda)F \right|
    &\le&\int|(A(\lambda)-A(0))F||B(\lambda)f|\nonumber\\
    &\le&\sup_{x\in\supp F}|(A(\lambda)-A(0))F |\int B(\lambda) f\,,
  \end{eqnarray}
  where we have used the facts that $B(\lambda) \geq 0$ and $f \geq 0$
  in the second line. Now let $\epsilon>0$. Since $A(\lambda)\to A(0)$
  uniformly as $\lambda\to 0$ on compact sets,
  there exists $\lambda_1>0$ such that $\sup_{x\in\supp
    F}|(A(\lambda)-A(0))F | < \epsilon$ for $\lambda<\lambda_1$.
  Similarly, since $B(\lambda)\to B(0)$ weakly as $\lambda\to 0$,
  there exists $\lambda_2>0$ such that
  $\int(B(\lambda)-B(0))f<\epsilon$ for all $\lambda < \lambda_2$.  
Thus, for all $\lambda <
  \min(\lambda_1,\lambda_2)$, we have
  \begin{equation}
    \left|\int(A(\lambda)-A(0))B(\lambda)F \right| < \epsilon \left(\int B(0)f + \epsilon \right) \,,
  \end{equation}
  Thus the first term in \eqref{eqn:lemmaproof} must be zero as well.
\end{proof}

The vanishing of $\kappa_{efab}$ (see \eqref{rho}) is a direct
consequence of this lemma. To see this, let $t^a$ be an arbitrary timelike
vector field in the metric $g^{(0)}_{ab}$. Then, we have
\begin{equation}
\kappa_{efab}t^a t^b = \wlim_{\lambda\to 0}h_{ef}(\lambda) (T_{ab}(\lambda) t^a t^b)
\label{rhott}
\end{equation}
Since by condition (ii) on our one-parameter families, $|g_{ab}
(\lambda) - g^{(0)}_{ab}| \leq \lambda C_1(x)$, on any fixed compact
region, we can find a $\lambda_0$ such that $t^a$ is timelike in the
metric $g_{ab} (\lambda)$ for all $\lambda \leq \lambda_0$. Since
$T_{ab}(\lambda)$ satisfies the weak energy condition, the function
$B(\lambda) \equiv T_{ab}(\lambda) t^a t^b$ is non-negative for all
$\lambda \leq \lambda_0$. We previously showed that $T_{ab}(\lambda)$
(and hence $B$) converges weakly. Assumption (ii) also directly tells
us that $A_{ab} (\lambda) \equiv h_{ab} (\lambda)$ converges uniformly
on compact sets. Thus, from the lemma, we immediately conclude that 
\begin{equation}
\kappa_{efab}t^a t^b = 0
\label{rhott2}
\end{equation}
for all timelike $t^a$ in the metric $g^{(0)}_{ab}$. However, since
$\kappa_{efab} = \kappa_{ef(ab)}$, we have $\kappa_{efab}t^a t^b = 0$
for all timelike $t^a$ if and only if $\kappa_{efab}= 0$, as we
desired to show.

A similar argument establishes that the second term on the right side
of \eqref{aht} also vanishes. Thus, we obtain
\begin{equation}
\alpha_{a\phantom{c}bcef}^{\phantom{a}c}= 0  \,.
\label{aht2}
\end{equation}
Consequently, \eqref{eqn:BackgroundEinstein-messy2} simplifies to 
\begin{equation}\label{eqn:BackgroundEinstein-simple2}
  8\pi t^{(0)}_{ab} = \alpha^{\phantom{a}c\phantom{b}d}_{a\phantom{c}b\phantom{d}cd} \,.
\end{equation}
Taking the trace of this equation and again using \eqref{aht2} we
obtain our first main result of this section:

\begin{theorem}Given a one-parameter family $g_{ab}(\lambda)$
  satisfying assumptions (i)--(iv) above, the effective stress energy
  tensor $t^{(0)}_{ab}$ appearing in equation \eqref{Eeg0} for the
  background metric $g^{(0)}_{ab}$ is traceless,
\begin{equation}
{t^{(0)a}}_a = 0 \,.
\end{equation}
\end{theorem}

\bigskip

We now show that $t^{(0)}_{ab}$ satisfies the weak energy condition.
Let $t^a$ be a timelike unit vector field with respect to
$g^{(0)}_{ab}$ .  We wish to show that $t^{(0)}_{ab}t^at^b\ge0$.  It
is convenient to choose an orthonormal basis of $g^{(0)}_{ab}$ with
$t^a$ as the timelike vector. We will use Greek letters
$\mu,\nu,\rho,\ldots$ to denote spacetime components in this basis and
Latin letters $i,j,k,\ldots$ from mid-alphabet to denote spatial
components.  Then
\begin{eqnarray}
  8\pi t^{(0)}_{ab} t^at^b&=&\alpha_{0\mu0\nu}^{\phantom{0\mu0\nu}\mu\nu}\nonumber\\
  &=&\alpha_{0j0k}^{\phantom{0j0k}jk}\nonumber\\
  &=&\alpha_{ij\phantom{i}k}^{\phantom{ij}i\phantom{k}jk}\nonumber\\
  &=&\frac{1}{4}\left\{\mu_{i\phantom{i}jk}^{\phantom{i}i\phantom{jk}jk}+\mu_{jki}^{\phantom{jki}ijk}-2\mu_{jik}^{\phantom{jik}ijk}\right\}\,,
\label{t00}
\end{eqnarray}
where in the second line we used the antisymmetry of $\alpha_{abcdef}$
in the first two and the second two indices, and in the third line we
used \eqref{aht2}.  Thus, we have expressed $t^{(0)}_{ab}t^at^b$ entirely
in terms of spatial components of $\mu_{abcdef}$, which will be useful
for taking advantage of the positive definiteness of the spatial
metric.

Aside from the tensor symmetries that arise directly from its
definition, the only restrictions on $\mu_{abcdef}$ that we have at
our disposal come from \eqref{aht2}. There is only one equation that
can be derived from \eqref{aht2} that involves only spatial components
of $\mu_{abcdef}$, namely\footnote{Equation \eqref{aht2} states that
  the weak limit of $h_{ef}$ times the linearized Ricci tensor
  vanishes, i.e., it has the character of the linearized vacuum
  Einstein equation off of flat spacetme. The linearized Hamiltonian
  constraint---i.e., the vanishing of the time-time component of the
  linearized Einstein tensor---is the only component of Einstein's
  equation that can be expressed entirely in terms of spatial
  derivatives of spatial components of the perturbed metric. Equation
  \eqref{lhc} corresponds to this equation.}
\begin{equation}
0 = \alpha_{i\mu\phantom{i\mu}kl}^{\phantom{i\mu}i\mu} - \alpha_{0\mu}{}^{0\mu}{}_{kl}\,.
\label{lhc}
\end{equation}
This yields
\begin{equation}
  \mu_{ij\phantom{ij}kl}^{\phantom{ij}ij}=\mu_{i\phantom{i}j\phantom{j}kl}^{\phantom{i}i\phantom{j}j}\,.
\end{equation}
Using this relation, we may re-write \eqref{t00} as
\begin{equation}
  8\pi t^{(0)}_{ab} t^at^b=\frac{1}{4}\left\{\mu_{i\phantom{i}jk}^{\phantom{i}i\phantom{jk}jk}-2\mu_{jik}^{\phantom{jik}ijk}+2\mu_{jki}^{\phantom{jki}ijk}-\mu_{i\phantom{i}j\phantom{j}k}^{\phantom{i}i\phantom{j}j\phantom{k}k}\right\}\,.
\end{equation}

For the remainder of our argument, we will work in a
small neighborhood of an arbitrary point $P\in M$. We will work in
Riemannian normal coordinates $x$ about $P$ adapted to our orthonormal
basis. Let $f_P^\delta$ be a one-parameter family of smooth,
non-negative functions with support contained in a $\delta$-ball centered at $P$ such that
\begin{equation}
  \int \left[f_P^\delta(x)\right]^2 \sqrt{-g} \ud^4x=1\,.
\label{fdelta}
\end{equation}
An explicit choice of $f_P^\delta$ is
\begin{equation}
  f_P^\delta(x) = \frac{1}{\delta^2\sqrt{-g}} F(x/\delta) \,,
\end{equation}
where $F$ is any smooth, non-negative function of compact support contained in a
ball of radius $1$ satisfying $\int F^2 d^4x = 1$, but there is no
need to make this particular choice.  Instead of working with
$h_{ab}(\lambda)$, we introduce the quantity
\begin{equation}
\psi_{ab}(\delta,\lambda) \equiv f_P^\delta h_{ab}(\lambda)\,.
\label{psidef}
\end{equation}
Note that for any fixed $\delta > 0$ and $\lambda > 0$, $\psi_{ab}$ is
smooth and of compact support, so, in particular, $\psi_{ab}$ and all
of its spacetime derivatives are in $L^2$. Furthermore, it follows
directly from the properties of $h_{ab}(\lambda)$ and $f_P^\delta$
that all components of $\psi_{ab}$ converge uniformly to $0$ as
$\lambda \to 0$ at fixed $\delta$. Similarly, components of $\nabla_c
\psi_{ab}$ are uniformly bounded in $\lambda$ and $x$ as $\lambda \to
0$ at fixed $\delta$. Since $\psi_{ab}$ is of fixed compact support,
it follows immediately that $\|\psi_{ab}\|_{L^2} \to 0$ and
$\|\nabla_c \psi_{ab}\|_{L^2}$ is uniformly bounded as $\lambda \to 0$,
where $\|\psi_{ab}\|_{L^2} \equiv \int |\psi_{ab}|^2$.

Since $\mu_{abcdef}$ is smooth, it is obvious from \eqref{fdelta} and
the support properties and positivity of $(f_P^\delta)^2$ that
\begin{equation}
\mu_{\mu\nu\alpha\beta\gamma\rho}(P) = \lim_{\delta \to
  0}\int\mu_{\mu\nu\alpha\beta\gamma\rho}(x)(f_P^\delta)^2 \sqrt{-g} \ud^4x\,.
\end{equation}
On the other hand, since $(f_P^\delta)^2$ is a test function, from the
definition \eqref{mu} of $\mu_{abcdef}$, at each fixed $\delta$ we have
\begin{eqnarray}
  \int\mu_{\mu\nu\alpha\beta\gamma\rho}(x)(f_P^\delta)^2\sqrt{-g}\ud^4x&=&\lim_{\lambda\to
    0}\int\partial_\mu h_{\alpha\beta}(\lambda)\partial_\nu h_{\gamma\rho}(\lambda) (f_P^\delta)^2\sqrt{-g}\ud^4x\nonumber\\ &=&\lim_{\lambda\to
    0}\int\partial_\mu\psi_{\alpha\beta}\partial_\nu\psi_{\gamma\rho}\sqrt{-g}\ud^4x\,.
\label{muavg}
\end{eqnarray}
Here, in the first line, we replaced the derivative operator
$\nabla_a$ associated with $g^{(0)}_{ab}$ with the coordinate
derivative operator $\partial_a$ associated with Riemannian normal
coordinates at $P$, making use of the fact that the definition of
$\mu_{abcdef}$ is independent of derivative operator. In the second
line, we used $\partial_\mu\psi_{\alpha\beta} =
f_P^\delta \partial_\mu h_{\alpha\beta} + h_{\alpha\beta} \partial_\mu
f_P^\delta$ and the fact that the resulting terms in \eqref{muavg}
with no derivatives on $h_{\alpha\beta}$ vanish in the limit as
$\lambda \to 0$. Taking the limit of \eqref{muavg} as $\delta \to 0$,
we obtain
\begin{equation}\label{eqn:muavgpsi}
\mu_{\mu\nu\alpha\beta\gamma\rho}(P) =
\lim_{\delta \to 0} \lim_{\lambda\to 0}\int\partial_\mu\psi_{\alpha\beta}\partial_\nu\psi_{\gamma\rho}\ud^4x\,,
\end{equation}
where we have used the fact that $\sqrt{-g} = 1$ at $P$. Note that it
is critical in this equation that the limits be taken in the order specified.

The corresponding formula for $t^{(0)}_{00}$ is
\begin{equation}
t^{(0)}_{00}(P)=\frac{1}{32\pi}\lim_{\delta\to0}\lim_{\lambda\to0}\int\ud^4x\left\{\partial_i\psi_{jk}\partial^i\psi^{jk}-2\partial_j\psi_k^{\phantom{k}i}\partial_i\psi^{jk}+2\partial_j\psi_i^{\phantom{i}i}\partial_k\psi^{jk}-\partial_i\psi_j^{\phantom{j}j}\partial^i\psi_k^{\phantom{k}k}\right\}\,.
\label{t00psi}
\end{equation}
where, in this equation, indices are raised and lowered with the flat
Euclidean spatial metric $\eta_{ij} = {\rm diag} (1,1,1)$ corresponding to the
spatial components of $g^{(0)}_{ab}$ at $P$. The major advantage of
\eqref{t00psi} is that we can apply usual Fourier transform techniques
to evaluate the integral appearing on the right side of this equation
prior to taking the limit. If we had not ``localized''
$h_{ab}(\lambda)$ by multiplying it by $f_P^\delta$, the Fourier
transform of $h_{ab}(\lambda)$ could have been ill defined and, even
if it were well defined, it would contain global information about
$h_{ab}(\lambda)$ rather than local information about the behavior of
$h_{ab}(\lambda)$ near $P$.

Our strategy now will be to prove that the quantity
\begin{equation}
\lim_{\lambda\to0}\int\ud^4x\left\{\partial_i\psi_{jk}\partial^i\psi^{jk}-2\partial_j\psi_k^{\phantom{k}i}\partial_i\psi^{jk}+2\partial_j\psi_i^{\phantom{i}i}\partial_k\psi^{jk}-\partial_i\psi_j^{\phantom{j}j}\partial^i\psi_k^{\phantom{k}k}\right\}\,
\end{equation}
can be expressed as a sum of terms which are either positive, or which
converge to zero as $\delta\to0$. Positivity of $t^{(0)}_{00}$ then
follows immediately.  We proceed by taking Fourier transforms with
respect to the spatial coordinates $\boldsymbol{x}$ only, with the
convention
\begin{equation}
  \hat{\psi}_{ij}(t,\boldsymbol{k})=
  \frac{1}{(2\pi)^{3/2}}\int\ud^3\boldsymbol{x}\,\psi_{ij}(t,\boldsymbol{x})e^{-i \boldsymbol{k}\cdot\boldsymbol{x}}\,.
\end{equation}
As previously noted, $\psi_{jk}$ and all of its derivatives are
obviously in $L^2$, since $\psi_{jk}$ is smooth and of compact
support. Since the Fourier transform is inner product preserving in
$L^2$, we have
\begin{equation}
t^{(0)}_{00} (P)=\frac{1}{32\pi}\lim_{\delta\to0}\lim_{\lambda\to0}\int \ud t\ud^3\boldsymbol{k}\, \left\{k_ik^i\hat{\psi}_{jk}\overline{\hat{\psi}^{jk}}-2k_ik_j\hat{\psi}_k^{\phantom{k}i}\overline{\hat{\psi}^{jk}}+2k_jk_k\hat{\psi}_i^{\phantom{i}i}\overline{\hat{\psi}^{jk}}-k_ik^i\hat{\psi}_j^{\phantom{j}j}\overline{\hat{\psi}_k^{\phantom{k}k}}\right\}\,.
\label{t00ft}
\end{equation}

We may decompose $\hat{\psi}_{ij}$ into its scalar, vector, and tensor parts as
\begin{equation}
  \hat{\psi}_{ij}(t,\boldsymbol{k})=\hat{\sigma}(t,\boldsymbol{k})k_ik_j- 2\hat{\varphi} q_{ij} + 2k_{(i} \hat{z}_{j)}(t,\boldsymbol{k})+ \hat{s}_{ij} (t,\boldsymbol{k}) \,,
\label{psidecom}
\end{equation}
where $k^i\hat{z}_i=0=k^i\hat{s}_{ij}$, and
$\hat{s}^i_{\phantom{i}i}=0$. Here $q_{ij}$ is the projection
orthogonal to $k^i$ of the Euclidean metric on Fourier transform
space. Since the various terms on the right side of \eqref{psidecom}
are orthogonal at each $\boldsymbol{k}$, it follows immediately that,
for example, $|\hat{\varphi}(\boldsymbol{k})|^2 \leq
\frac{1}{8}\hat{\psi}^{ij}(\boldsymbol{k})
\overline{\hat{\psi}}_{ij}(\boldsymbol{k})$. Since $\hat{\psi}_{ij}$
and all powers of $k^i$ times $\hat{\psi}_{ij}$ are in $L^2$, it
follows immediately that $\hat{\varphi}$ and all powers of $k^i$ times
$\hat{\varphi}$ are in $L^2$. Thus, we can freely take Fourier transforms
of $\hat{\varphi}$ and all powers of $k^i$ times
$\hat{\varphi}$. Furthermore, since the $L^2$ norm of $\psi_{ij}$---and,
hence, the $L^2$ norm of $\hat{\psi}_{ij}$---goes to zero as $\lambda
\to 0$, it follows immediately that the $L^2$ norm of
$\hat{\varphi}$---and hence the $L^2$ norm of $\varphi$---also goes to zero
as $\lambda \to 0$. Similarly, the $L^2$ norm of $\partial_i \varphi$
must remain uniformly bounded as $\lambda \to 0$. Similar results hold
for the other terms appearing on the right side of \eqref{psidecom}.

Substituting the decomposition \eqref{psidecom} in \eqref{t00ft} and
using the fact that $\psi_{ij}$ is real (which implies that
$\overline{\hat{\psi}}(t,\boldsymbol{k}) =
\hat{\psi}(t,-\boldsymbol{k})$), we obtain
\begin{equation}\label{eqn:T00FT}
t^{(0)}_{00}
(P)=\frac{1}{32\pi}\lim_{\delta\to0}\lim_{\lambda\to0}\int \ud
t\ud^3\boldsymbol{k}\,
\left\{k_ik^i\hat{s}_{jk}\overline{\hat{s}^{jk}}-8k_ik^i\hat{\varphi}\overline{\hat{\varphi}}\right\}\,.
\end{equation}
Thus, we see that the ``tensor part'', $\hat{s}_{ij}$, of
$\hat{\psi}_{ij}$ makes a positive contribution to the effective
gravitational energy density $t^{(0)}_{00}$. This may be interpreted
as saying that, at leading order within this framework, gravitational
radiation carries positive energy density. The scalar $\hat{\sigma}$
and the vector part $\hat{z}_i$ do not contribute at all, as might be
expected from the fact that these quantities should correspond to
``pure gauge''. Finally, the scalar $\hat{\varphi}$ makes the negative
contribution
\begin{equation}
2E_\varphi = -\frac{1}{4\pi}\int \ud t \ud^3 \boldsymbol{k} k_ik^i\hat{\varphi}\overline{\hat{\varphi}}
\label{Ephi}
\end{equation}
to the effective energy density.

In order to interpret the meaning of $\hat{\varphi}$ and $E_\varphi$, we
note that by \eqref{psidecom}, $\hat{\varphi}$ satisfies
\begin{equation}
4 k^ik_i \hat{\varphi} =  - k^i k_i \hat{\psi}^j_{\phantom{j}j} + k^i k^j \hat{\psi}_{ij}\,.
\label{pois0}
\end{equation}
In position space, \eqref{pois0} becomes
\begin{equation}
4 \partial^i \partial_i \varphi = - \partial^i \partial_i {\psi^j}_j 
+ \partial^i \partial^j \psi_{ij}\,.
\label{pois1}
\end{equation}
To put the right side of this equation in a more recognizable form, we
return to Einstein's equation \eqref{eqn:EinsteinRelation} and
consider its normal-normal component relative to a $t = {\rm const}$
surface.  This corresponds to the Hamiltonian constraint equation,
which has the property that the only terms containing second spacetime
derivatives of $h_{ab}$ involve only spatial derivatives of spatial
components of $h_{ab}$.  To obtain this equation from
\eqref{eqn:EinsteinRelation}, we raise both indices with $g^{ab}
(\lambda)$ and then take the $00$ component.  Since we are working in
Riemannian normal coordinates about $P$ we also express the background
metric as
\begin{equation}\label{rnc}
  g^{(0)}_{\alpha\beta}=\eta_{\alpha\beta}-\frac{1}{3}R_{\alpha\mu\beta\nu}x^\alpha x^\beta + O(x^3)\,.
\end{equation}
The terms in the resulting equation that are purely linear in
$h_{ab}$, contain second derivatives of $h_{ab}$, and do not depend on
the background curvature are of the form
$\frac{1}{2}\partial^i \partial_i {h^j}_j -
\frac{1}{2}\partial^i \partial^j h_{ij}$, i.e. the same combination of
derivatives of components\footnote{These terms correspond to the
  linearization of the scalar curvature of the spatial metric, as
  would be expected from the general form of the ``Hamiltonian
  constraint equation''.} as appears in \eqref{pois1}.  There are also
terms which are linear in $h_{ab}$ and contain second derivatives of
$h_{ab}$, which depend on the difference between the exact background
metric $g_{ab}^{(0)}$ and $\eta_{ab}$, and these can be expressed in
the form $U^{ijkl}\partial_i\partial_jh_{kl}$, where $U^{ijkl} =
O(x^2)$.  The terms in the equation which are nonlinear in $h_{ab}$
that contain second derivatives of $h_{ab}$ can be expressed in the
form $\partial_i W^i + Z_1$ where $W^i$ converges to zero uniformly on
compact sets as $\lambda \to 0$ and $Z_1$ is uniformly bounded on
compact sets as $\lambda \to 0$. The remaining terms in this equation
then take the form $8 \pi T^{00} + Z_2$, where $Z_2$ is uniformly
bounded on compact sets as $\lambda \to 0$. Now multiply this equation
by $f_P^\delta$. Using \eqref{pois1}, we see that the resulting
equation takes the form
\begin{equation}
\partial^i \partial_i \varphi(\lambda) = 4 \pi f^\delta_P T^{00}(\lambda) + \partial_i \omega^i(\lambda) + \zeta(\lambda) + \frac{1}{2}U^{ijkl}\partial_i\partial_j\psi_{kl}(\lambda)\,,
\label{pois2}
\end{equation}
where $\omega^i \equiv \frac{1}{2}f^\delta_P W^i$ is of fixed compact support
and converges to zero uniformly as $\lambda \to 0$, and $\zeta$ is of
fixed compact support and is uniformly bounded as $\lambda \to
0$. Thus, $\varphi$ satisfies a Poisson-like equation. Furthermore, the
position space version of \eqref{Ephi} can be written as
\begin{equation}
  2E_\varphi = -\frac{1}{4\pi} \int \ud t\ud^3\boldsymbol{x} \partial^i \varphi  \partial_i\varphi \,\, ,
\label{Ephi2}
\end{equation}
which is just the usual formula for (twice) the gravitational
potential energy in Newtonian gravity! Note that we have not made any
Newtonian approximations, nor have we made a special choice of ``time
vector'' $t^a$.

Thus, we see that the resolution of the issue of whether $t^{(0)}_{00}
\geq 0$ depends on a competition between the positive contribution
from the tensor modes and the negative contribution, $E_\varphi$, arising
from a Newtonian-like gravitational potential energy. We will now show
that if $T_{00}(\lambda) \geq 0$, then, in fact, $E_\varphi \to 0$ as
$\lambda \to 0$ and $\delta\to0$. Thus, the scalar modes make no
contribution in this limit, and the tensor modes always ``win''.

To prove this, we use \eqref{pois2} to rewrite $E_\varphi$ as
\begin{eqnarray}
E_\varphi &=& \frac{1}{8\pi}\int \ud t\ud^3\boldsymbol{x} \varphi  \partial^i \partial_i\varphi \nonumber \\
&=& \frac{1}{2} \int \ud t\ud^3\boldsymbol{x} \varphi(\lambda) \left[f^\delta_p T^{00}(\lambda) + \frac{1}{4\pi}\partial_i \omega^i(\lambda) + \frac{1}{4\pi}\zeta(\lambda)+\frac{1}{8\pi}U^{ijkl}\partial_i\partial_j\psi_{kl}(\lambda)\right] \,\, .
\label{Ephi3}
\end{eqnarray}
The second term can be written as the time integral of
\begin{equation}
\int \ud^3\boldsymbol{x} \varphi \partial_i \omega^i = - \int  \ud^3\boldsymbol{x} \partial_i \varphi \omega^i  \,\, .
\end{equation}
By the Schwartz inequality, we have
\begin{equation}
\left|\int \ud^3\boldsymbol{x} \partial_i \varphi \omega^i\right| \leq \|\partial_i\varphi\|_{L^2} \|\omega^i\|_{L^2}  \,\, .
\end{equation}
However, $\|\omega^i\|_{L^2} \to 0$ as $\lambda \to 0$, and we have
already noted that $\|\partial_i\varphi\|_{L^2}$ remains uniformly
bounded as $\lambda \to 0$. Therefore, we see that the second term
vanishes in the limit as $\lambda$ goes to zero.

To analyze the remaining terms on the right side of \eqref{Ephi3},
suppose we could show that $\varphi(\lambda)$ converges uniformly to $0$
on compact sets as $\lambda \to 0$. Then since $\zeta(\lambda)$ is of
fixed compact support and is uniformly bounded as $\lambda \to 0$, it
follows immediately that $\int \ud t\ud^3\boldsymbol{x}
\varphi(\lambda)\zeta(\lambda) \to 0$ as $\lambda \to 0$, so the third
term in \eqref{Ephi3} vanishes in the limit as $\lambda \to 0$. On the
other hand, if $\varphi(\lambda)$ converges uniformly and
$T^{00}(\lambda) \geq 0$, then the first term is exactly of the form
to which the above Lemma of this section applies, with $A=\varphi$,
$B=T^{00}$, and $f^\delta_P$ being the test function with which
$A(\lambda)B(\lambda)$ is being smeared. The Lemma then states that
the first term in \eqref{Ephi3} vanishes in the limit as $\lambda \to
0$.  Finally, the last term of \eqref{Ephi3} may be re-written as
\begin{equation}
  -\frac{1}{16\pi}\int \ud t\ud\boldsymbol{x}\left[U^{ijkl}\partial_i\varphi(\lambda)\partial_j\psi_{kl}(\lambda)+\partial_i U^{ijkl}\varphi(\lambda)\partial_j\psi_{kl}(\lambda)\right]\,.
\end{equation}
If $\varphi(\lambda)$ converges uniformly to zero, then the second of
these two terms vanishes in the limit as $\lambda\to0$ because
$U^{ijkl}$ is independent of $\lambda$ and
$\partial_j\psi_{kl}(\lambda)$ is uniformly bounded as $\lambda\to0$
and is of fixed compact support. In contrast to all the others, the
first term above does not converge to zero as $\lambda\to0$.  However,
it will still vanish once we subsequently take $\delta\to0$.  To see
this, use the fact that on the support of $f_P^\delta$, $U^{ijkl}$ is
bounded by a constant times $\delta^2$ (see \eqref{rnc}) to write
\begin{equation}\label{eqn:dpsimu}
  \left|\int \ud t\ud\boldsymbol{x}U^{ijkl}\partial_i\varphi(\lambda)\partial_j\psi_{kl}(\lambda)\right|\le C\delta^2\|\partial_i\varphi\|_{L^2}\|\partial_j\psi_{kl}\|_{L^2}\le C'\delta^2\left\|\partial_j\psi_{kl}\right\|^2_{L^2}\,.
\end{equation}
where the last inequality follows from the fact that
$\|\partial_i\varphi\|^2_{L^2}\le\frac{1}{8}\|\partial_j\psi_{kl}\|^2_{L^2}$.
However, by \eqref{eqn:muavgpsi} we have
\begin{equation}
  \lim_{\delta\to0}\lim_{\lambda\to0}\left\|\partial_j\psi_{kl}\right\|^2_{L^2}=\mu^{j\phantom{j}kl}_{\phantom{j}j\phantom{kl}kl}(P)\,.
\end{equation}
Consequently, the right hand side of \eqref{eqn:dpsimu} vanishes when
the limits as $\lambda\to0$ and $\delta\to0$ are taken.  Thus, we will
have proven that $t^{(0)}_{00} \geq 0$ provided only that we show that
$\varphi(\lambda)$ converges uniformly to $0$ on compact sets as $\lambda
\to 0$.

To prove uniform convergence to $0$ of $\varphi(\lambda)$ on compact
sets, we note that it follows immediately from \eqref{pois1} that
\begin{equation}
4\varphi = -{\psi^i}_i + \chi \,\, ,
\end{equation}
where
\begin{equation}
\partial^i \partial_i \chi = \partial^i \partial^j \psi_{ij} \,\, .
\label{poischi0}
\end{equation}
As already noted above, $\psi_{ij}(\lambda)$---and hence
${\psi^i}_i (\lambda)$---converges to $0$ uniformly as $\lambda \to 0$. Thus,
$\varphi(\lambda)$ will converge to $0$ uniformly on compact
sets if and only if $\chi(\lambda)$ converges to $0$ uniformly on compact
sets. We will now prove this by ``brute force''.

The solution to \eqref{poischi0} is
\begin{equation}
\chi(t,\boldsymbol{x}) = -\frac{1}{4\pi} \int \ud^3\boldsymbol{x}'\,\frac{\partial_i\partial_j\psi^{ij}(t,\boldsymbol{x}')}{|\boldsymbol{x}-\boldsymbol{x}'|} \,\, .
\label{poischi1}
\end{equation}
Changing integration variables to $\boldsymbol{y} = \boldsymbol{x}' -
\boldsymbol{x}$ and integrating by parts, we obtain
\begin{equation}
\chi(t,\boldsymbol{x}) = -\frac{1}{4\pi} \int \ud^3\boldsymbol{y}\,
\partial_i\psi^{ij}(t,\boldsymbol{x}+\boldsymbol{y} )\frac{y_j}{|\boldsymbol{y}|^3}\,.
\label{poischi2}
\end{equation}
For any $r_0>0$, we can break up the integral appearing on the right
side of this equation into an integral over $|\boldsymbol{y}|<r_0$ and
an integral over $|\boldsymbol{y}| \geq r_0$. We leave the first
integral alone but do another integration by parts on the second
integral. We thereby obtain
\begin{eqnarray}
\chi(t,\boldsymbol{x}) &=& -\frac{1}{4\pi} \int_{|\boldsymbol{y}|<r_0} \ud^3\boldsymbol{y}\,
\partial_i\psi^{ij}(t,\boldsymbol{x}+\boldsymbol{y})\frac{y_j}{|\boldsymbol{y}|^3} +\frac{1}{4\pi}\int_{|\boldsymbol{y}|=r_0}\ud \Omega\,r_0^2\frac{y_i}{r_0}\left(\psi^{ij}(t,\boldsymbol{x}+\boldsymbol{y})\frac{y_j}{r_0^3}\right) \nonumber \\
&& \qquad +\frac{1}{4\pi}\int_{|\boldsymbol{y}|>r_0} \ud^3\boldsymbol{y}\,\psi^{ij}(t,\boldsymbol{x}+\boldsymbol{y})\frac{\delta_{ij}|\boldsymbol{y}|^2-y_iy_j}{|\boldsymbol{y}|^5}\,.
\label{poischi3}
\end{eqnarray}
Now let $F_1(\lambda) \equiv \sup_{(t,\boldsymbol{x})} |\psi_{ij}|$
and let $F_2(\lambda) \equiv \sup_{(t,\boldsymbol{x})} |\partial_k
\psi_{ij}|$. Then $F_1 \rightarrow 0$ and $F_2$ remains bounded as
$\lambda \to 0$. It follows straightforwardly from \eqref{poischi3} that
\begin{equation}
|\chi(t,\boldsymbol{x})| \leq c_1 F_2(\lambda) r_0 + c_2 F_1(\lambda)
+ c_3 F_1(\lambda) |\ln(C/r_0)| \,\, ,
\label{poischi4}
\end{equation}
where $c_1$, $c_2$, $c_3$, and $C$ are constants (i.e., independent of
$\lambda$, $r_0$, $t$, and $\boldsymbol{x}$). This bound holds for all
$r_0$ and all $\lambda$. Therefore, as we let $\lambda \to 0$, we are
free to choose $r_0$ to vary with $\lambda$ in any way that is
convenient. Choosing
\begin{equation}
r_0 (\lambda) = \exp[-1/\sqrt{F_1(\lambda)}]\,,
\end{equation}
we obtain the bound
\begin{equation}
|\chi(t,\boldsymbol{x})| \leq C_1\exp[-1/\sqrt{F_1(\lambda)}] + C_2
F_1(\lambda) + C_3 \sqrt{F_1(\lambda)} \,\,,
\label{poischi5}
\end{equation}
from which it follows immediately that $\chi \to 0$ uniformly as
$\lambda \to 0$, as we desired to show.

We have thus proven
\begin{theorem}Given a one-parameter family $g_{ab}(\lambda)$
  satisfying assumptions (i)--(iv) above, the effective stress energy
  tensor $t^{(0)}_{ab}$ appearing in equation \eqref{Eeg0} for the
  background metric $g^{(0)}_{ab}$ satisfies the weak energy
  condition, i.e.,
\begin{equation}
t^{(0)}_{ab} t^a t^b \geq 0 
\end{equation}
for all $t^a$ that are timelike with respect to $g^{(0)}_{ab}$.
\end{theorem}

\bigskip

It should be emphasized that all of the results of this section apply
to an arbitrary one-parameter family $g_{ab}(\lambda)$ satisfying
assumptions (i)--(iv). In particular, no symmetry or other assumptions
concerning the background metric, $g^{(0)}_{ab}$, were made. However,
if FLRW symmetry is assumed for $g^{(0)}_{ab}$ as well as for the
(weak limit of) the matter stress-energy tensor, $T^{(0)}_{ab}$, then
$t^{(0)}_{ab}$ must also have this symmetry. It then follows
immediately from Theorems 1 and 2 that $t^{(0)}_{ab}$ must have the
form of a perfect fluid with $P=\frac{1}{3}\rho$ and $\rho\ge0$. In
particular, the effective stress-energy tensor arising from nonlinear
terms in Einstein's equation associated with short-wavelength
inhomogeneities cannot produce effects similar to those of dark
energy.

\section{Perturbation Theory}

In the previous section, we obtained the equation satisfied by the
background metric, $g^{(0)}_{ab}$, derived key properties of the
effective stress-energy tensor $t^{(0)}_{ab}$, and thereby proved that
small scale inhomogeneities cannot mimic the effects of dark energy on
large scale dynamics. However, in cosmology and other contexts, we
wish to know not only the dynamical behavior of $g^{(0)}_{ab}$ but
also the dynamical behavior of the deviations from $g^{(0)}_{ab}$, as
this is needed to describe the formation and growth of structures in
the universe. In particular, we would like to obtain the equations
satisfied by $h_{ab}(\lambda)$ to sufficient accuracy that
$h_{ab}(\lambda)$ can be determined to first order in $\lambda$, i.e.,
any deviations from an exact solution (over a compact spacetime
region) go to zero faster than $\lambda$ as $\lambda \rightarrow
0$. As already mentioned in the introduction, if we were in the
context of ordinary perturbation theory where $g_{ab}(\lambda,x)$ is
jointly differentiable in $\lambda$ and $x$, we would define
\begin{equation}\label{eqn:ordinarygamma}
  \gamma^{(1)}_{ab}\equiv\left.\frac{\partial g_{ab}(\lambda)}{\partial \lambda}\right|_{\lambda=0}=\lim_{\lambda\to0}\frac{g_{ab}(\lambda)-g_{ab}^{(0)}}{\lambda}\,.
\end{equation}
To derive the equation satisfied by $\gamma^{(1)}_{ab}$, we
differentiate the Einstein equation with respect to $\lambda$, at
$\lambda=0$.  The result is an equation that sets the linearized
Einstein operator acting on $\gamma^{(1)}_{ab}$ equal to the
derivative of the stress energy tensor with respect to $\lambda$,
evaluated at $\lambda=0$.  We would then take $h_{ab}(\lambda) =
\lambda \gamma^{(1)}_{ab}$. However, in the context of our framework,
$g_{ab}(\lambda,x)$ is not differentiable in $\lambda$ at $\lambda
=0$, so we cannot even define a notion of a ``metric perturbation'' by
differentiating $g_{ab}(\lambda,x)$.

Of course, the (exact) equation satisfied by $h_{ab}(\lambda)$ is
simply the equation obtained by substituting $g_{ab} (\lambda) =
g^{(0)}_{ab} + h_{ab} (\lambda)$ into Einstein's equation
\eqref{eqn:EinsteinRelation}. However, this is not any more useful in
practice than simply asserting that $g_{ab}(\lambda)$ must be a
solution of Einstein's equation for all $\lambda$; if we could solve
Einstein's equation exactly, there would be no need to develop a
perturbative formalism.  The key idea needed to obtain a more useful
version of Einstein's equation is that, although nonlinearities may be
important on small scales, there should be a simpler, linear
description on large scales. The key idea needed to implement this
description is the observation that although the ordinary (pointwise
or uniform) limit of $[g_{ab}(\lambda)-g_{ab}^{(0)}]/\lambda$ does not
exist in the context of our framework---$g_{ab}(\lambda)$ is not
differentiable---there is no reason why the {\em weak limit} of this
quantity cannot exist.

Thus, a natural generalization of the conventional linearized metric
perturbation is
\begin{equation}
  \gamma_{ab}^{(L)}\equiv\wlim_{\lambda\to0}\frac{g_{ab}(\lambda)-g_{ab}^{(0)}}{\lambda}\,.
\end{equation}
Here we have replaced the ordinary limit of \eqref{eqn:ordinarygamma}
with a weak limit, which ``averages away'' the small scale
inhomogeneities.  This quantity thereby corresponds closely to the
notion of the ``long wavelength part'' of the metric perturbation that
appears in other analyses (see, e.g., \cite{Baumann:2010tm}). We will
discuss this further in section V below.  The remainder of the
perturbation will be denoted
\begin{equation}
h^{(S)}_{ab}(\lambda)\equiv h_{ab}(\lambda)-\lambda\gamma^{(L)}_{ab}\,,
\end{equation}
and will be referred to as the ``short wavelength part'' of the
deviation\footnote{If we were to consider higher order perturbation
  theory, then we would also subtract from $h_{ab}$ higher order in
  $\lambda$ ``long wavelength'' contributions to define
  $h^{(S)}_{ab}(\lambda)$. However, we shall only be concerned with
  first order perturbation theory in this paper.} of the metric from
$g^{(0)}_{ab}$. In subsection IVB, we will argue that, under suitable
Newtonian assumptions in cosmology, to leading order in $\lambda$,
$h^{(S)}_{ab}$ depends only {\em locally} on the matter distribution
and is well approximated by a Newtonian gravity solution.  We
emphasize that, within our framework, short and long wavelength
perturbations have a very different character.  The long wavelength
part of a perturbation has a well-defined description in the
$\lambda\to 0$ limit, namely $\gamma^{(L)}_{ab}$.  On the other hand,
the short wavelength part, $h^{(S)}_{ab}(\lambda)$, is defined only
for $\lambda>0$ and has no description in terms of a limit as $\lambda
\rightarrow 0$.

We can obtain an equation for $\gamma_{ab}^{(L)}$ by taking the
difference of the exact Einstein equation \eqref{eqn:EinsteinRelation}
for $g_{ab}(\lambda)$ and the background Einstein equation
\eqref{Eeg0} for $g_{ab}^{(0)}$, dividing by $\lambda$, and taking the
weak limit as $\lambda \rightarrow 0$. However, in order to ensure
that $\gamma_{ab}^{(L)}$ is well defined and satisfies a well-defined
equation, we must append to assumptions (i)--(iv) of section II the
following additional assumptions on our one-parameter family
$g_{ab}(\lambda)$:
\renewcommand{\labelenumi}{(\roman{enumi})}
\begin{enumerate}
\setcounter{enumi}{4}
\item{There exist smooth tensor fields $\gamma^{(L)}_{ab}$,
    $\mu^{(1)}_{abcdef}$, $\nu^{(1)}_{abcde}$ and $\omega^{(1)}_{abcdefgh}$
    on $M$ such that
    \begin{enumerate}
    \item{
      \begin{equation}
        \wlim_{\lambda\to0}\frac{1}{\lambda}h_{ab}(\lambda)=\gamma^{(L)}_{ab}\,,
      \end{equation}}
    \item{
      \begin{equation}
        \wlim_{\lambda\to0}\frac{1}{\lambda}\left[\nabla_{(a}h^{(S)}_{|cd|}(\lambda)\nabla_{b)}h^{(S)}_{ef}(\lambda)-\mu_{abcdef}\right]=\mu^{(1)}_{abcdef}\,,
      \end{equation}}
    \item{
      \begin{equation}
        \wlim_{\lambda\to0}\frac{1}{\lambda}\left[h^{(S)}_{bc}(\lambda)\nabla_ah^{(S)}_{de}(\lambda)\right]=\nu^{(1)}_{abcde}\,,
      \end{equation}}
    \item{
      \begin{equation}
        \wlim_{\lambda\to0}\frac{1}{\lambda}\left[h^{(S)}_{cd}(\lambda)\nabla_ah^{(S)}_{ef}(\lambda)\nabla_bh^{(S)}_{gh}(\lambda)\right]=\omega^{(1)}_{abcdefgh}\,.
      \end{equation}}
    \end{enumerate}}
\end{enumerate}
\renewcommand{\labelenumi}{\theenumi} In the following, we shall
assume that our one-parameter family $g_{ab}(\lambda)$ satisfies
assumptions (i)--(iv) of section II together with assumption (v)
above. In this section we will make no additional assumptions about
$g_{ab}(\lambda)$, so all of the results obtained in this section
should hold, e.g., for self-gravitating gravitational radiation in a
background without any symmetries. In section IV, we shall specialize
to the case of Newtonian-like cosmological perturbations off of a
background metric with FLRW symmetry, and will make numerous
additional assumptions and simplifications.

The newly defined first order backreaction tensors,
$\mu^{(1)}_{abcdef}$, $\nu^{(1)}_{abcde}$ and
$\omega^{(1)}_{abcdefgh}$, possess certain tensor symmetries as a
direct consequence of their definitions.  Clearly, as with the zeroth
order quantity $\mu_{abcdef}$, we have
$\mu^{(1)}_{(ab)(cd)(ef)}=\mu^{(1)}_{abcdef}$ and
$\mu^{(1)}_{abcdef}=\mu^{(1)}_{baefcd}$.  However, in contrast with
$\mu_{abcdef}$, the symmetry under interchange of the first two
indices had to be built directly into the definition of
$\mu^{(1)}_{abcdef}$, rather than derived.  Indeed,
\begin{eqnarray}
  && \wlim_{\lambda\to0}\frac{1}{\lambda}\left[\nabla_ah_{cd}(\lambda)\nabla_bh_{ef}(\lambda)-\mu_{abcdef}\right]\nonumber\\
  &=&\mu^{(1)}_{abcdef}+\wlim_{\lambda\to0}\frac{1}{\lambda}\nabla_{[a}h_{|cd|}(\lambda)\nabla_{b]}h_{ef}(\lambda)\nonumber\\
  &=&\mu^{(1)}_{abcdef}+\wlim_{\lambda\to0}\frac{1}{\lambda}\left[\nabla_{[a}\left(h_{|cd|}(\lambda)\nabla_{b]}h_{ef}(\lambda)\right)-h_{cd}(\lambda)\nabla_{[a}\nabla_{b]}h_{ef}(\lambda)\right]\nonumber\\
  &=&\mu^{(1)}_{abcdef}+\nabla_{[a}\nu^{(1)}_{b]cdef}-2\wlim_{\lambda\to0}\frac{1}{\lambda}h_{cd}(\lambda)R_{ab(e}^{\phantom{ab(e}g}h_{f)g}(\lambda)\nonumber\\
  &=&\mu^{(1)}_{abcdef}+\nabla_{[a}\nu^{(1)}_{b]cdef}\,.
\end{eqnarray}
What this calculation also illustrates is that, because of the factor
of $1/\lambda$, we may no longer freely drop total derivative terms
when taking weak limits, and in general we will pick up terms of the
form $\nabla_a\nu^{(1)}_{bcdef}$.  It follows also that
$\nu^{(1)}_{a(bc)(de)}=\nu^{(1)}_{abcde}$ and
$\nu^{(1)}_{abcde}=-\nu^{(1)}_{adebc}$.  Finally,
$\omega^{(1)}_{ab(cd)(ef)(gh)}=\omega^{(1)}_{abcdefgh}$,
$\omega^{(1)}_{abcdefgh}=\omega^{(1)}_{bacdghef}$, and
\begin{equation}\label{eqn:omega-symmetry}
  \omega^{(1)}_{abcdefgh}-\omega^{(1)}_{bacdefgh}=\omega^{(1)}_{baefcdgh}-\omega^{(1)}_{abefcdgh}\,.
\end{equation}
We also note that definitions (b) and (c) would be unchanged if
$h^{(S)}_{ab}$ were replaced by $h_{ab}$, and (d) would be unchanged
if any $h^{(S)}_{ab}$ which is being differentiated were replaced by
$h_{ab}$.

We now subtract the background Einstein equation \eqref{Eeg0} from the
exact Einstein equation \eqref{eqn:EinsteinRelation}, divide by
$\lambda$, and then take the weak limit as $\lambda\to0$.  A very
lengthy calculation, performed with the help of the {\em xAct} tensor
manipulation package \cite{martin-garcia:2008,*martin-garcia:web} for
{\em Mathematica}, yields
\begin{eqnarray}\label{eqn:gammaL-1}
  &&\nabla^c\nabla_{(a}\gamma^{(L)}_{b)c}-\frac{1}{2}\nabla^c\nabla_c\gamma^{(L)}_{ab}-\frac{1}{2}\nabla_a\nabla_b\gamma^{(L)c}_{\phantom{(L)c}c}-\frac{1}{2}g^{(0)}_{ab}\left(\nabla^c\nabla^d\gamma^{(L)}_{cd}-\nabla^c\nabla_c\gamma^{(L)d}_{\phantom{(L)d}d}\right)\nonumber\\
  &&+\frac{1}{2}g^{(0)}_{ab}R^{cd}(g^{(0)})\gamma^{(L)}_{cd}-\frac{1}{2}R(g^{(0)})\gamma^{(L)}_{ab}+\Lambda\gamma^{(L)}_{ab}+\frac{1}{8}\gamma^{(L)}_{ab}\left(\mu^{c\phantom{c}d\phantom{d}e}_{\phantom{c}c\phantom{d}d\phantom{e}e}+\mu^{c\phantom{c}de}_{\phantom{c}c\phantom{de}de}-2\mu^{cd\phantom{c}e}_{\phantom{cd}c\phantom{e}de}\right)\nonumber\\
  &&+\gamma^{(L)cd}\left(\frac{1}{2}\mu_{abc\phantom{e}de}^{\phantom{abc}e}-\frac{1}{2}\mu_{c(ab)d\phantom{e}e}^{\phantom{c(ab)d}e}-\frac{1}{2}\mu_{\phantom{e}(ab)ecd}^e+\frac{3}{4}\mu_{cdab\phantom{e}e}^{\phantom{cdab}e}-\frac{1}{2}\mu_{cda\phantom{e}be}^{\phantom{cda}e}-\mu_{c\phantom{e}abde}^{\phantom{c}e}+\mu_{c\phantom{e}e(ab)d}^{\phantom{c}e}\right.\nonumber\\
  &&\left.+\frac{3}{4}\mu^e_{\phantom{e}eabcd}-\frac{1}{2}\mu^e_{\phantom{e}eacbd}+\frac{1}{8}g^{(0)}_{ab}\left\{-\mu_{cd\phantom{e}e\phantom{f}f}^{\phantom{cd}e\phantom{e}f}-\mu_{cd\phantom{ef}ef}^{\phantom{cd}ef}+4\mu_{c\phantom{e}d\phantom{f}ef}^{\phantom{c}e\phantom{d}f}-2\mu_{\phantom{e}ecd\phantom{f}f}^{e\phantom{ecd}f}+2\mu^{ef}_{\phantom{ef}cedf}\right\}\right)\nonumber\\
  &=&8\pi T^{(1)}_{ab} +\frac{1}{8}g^{(0)}_{ab}\left\{-\mu^{(1)}{}^{c\phantom{c}de}_{\phantom{c}c\phantom{de}de}-\mu^{(1)}{}^{c\phantom{c}d\phantom{d}e}_{\phantom{c}c\phantom{d}d\phantom{e}e}+2\mu^{(1)}{}^{cd\phantom{c}e}_{\phantom{cd}c\phantom{e}de}\right\}+\frac{1}{2}\mu^{(1)}{}^{cd}_{\phantom{cd}acbd}\nonumber\\
  &&-\frac{1}{2}\mu^{(1)}{}^{c\phantom{ca}d}_{\phantom{c}ca\phantom{d}bd}+\frac{1}{4}\mu^{(1)}{}^{\phantom{ab}cd}_{\phantom{}ab\phantom{cd}cd}-\frac{1}{2}\mu^{(1)}{}^{c\phantom{(ab)c}d}_{\phantom{c}(ab)c\phantom{d}d}+\frac{3}{4}\mu^{(1)}{}^{c\phantom{cab}d}_{\phantom{c}cab\phantom{d}d}-\frac{1}{2}\mu^{(1)}{}^{cd}_{\phantom{cd}abcd}\nonumber\\
  &&+\frac{1}{8}g^{(0)}_{ab}\left\{2\omega^{(1)}{}^{c\phantom{c}de\phantom{de}f}_{\phantom{c}c\phantom{de}de\phantom{f}f}+2\omega^{(1)}{}^{c\phantom{c}de\phantom{d}f}_{\phantom{c}c\phantom{de}d\phantom{f}ef}+\omega^{(1)}{}^{cd\phantom{cd}e\phantom{e}f}_{\phantom{cd}cd\phantom{e}e\phantom{f}f}+\omega^{(1)}{}^{cd\phantom{cd}ef}_{\phantom{cd}cd\phantom{ef}ef}-4\omega^{(1)}{}^{cd\phantom{c}e\phantom{d}f}_{\phantom{cd}c\phantom{e}d\phantom{f}ef}-2\omega^{(1)}{}^{cdef}_{\phantom{cdef}decf}\right\}\nonumber\\
  &&-\frac{1}{2}\omega^{(1)}{}^{\phantom{ab}cd\phantom{c}e}_{\phantom{}ab\phantom{cd}c\phantom{e}de}+\frac{1}{2}\omega^{(1)}{}^{\phantom{(a}c\phantom{|c|}d\phantom{b)d}e}_{(a\phantom{c}|c|\phantom{d}b)d\phantom{e}e}+\frac{1}{2}\omega^{(1)}{}^{\phantom{(a}cde}_{(a\phantom{cde}b)cde}-\frac{1}{8}\omega^{(1)}{}^{c\phantom{cab}d\phantom{d}e}_{\phantom{c}cab\phantom{d}d\phantom{e}e}-\frac{1}{8}\omega^{(1)}{}^{c\phantom{cab}de}_{\phantom{c}cab\phantom{de}de}-\frac{3}{4}\omega^{(1)}{}^{c\phantom{c}de}_{\phantom{c}c\phantom{de}abde}\nonumber\\
  &&+\frac{1}{2}\omega^{(1)}{}^{c\phantom{c}de}_{\phantom{c}c\phantom{de}adbe}+\frac{1}{4}\omega^{(1)}{}^{cd\phantom{abd}e}_{\phantom{cd}abd\phantom{e}ce}-\frac{3}{4}\omega^{(1)}{}^{cd\phantom{cdab}e}_{\phantom{cd}cdab\phantom{e}e}+\frac{1}{2}\omega^{(1)}{}^{cd\phantom{cda}e}_{\phantom{cd}cda\phantom{e}be}+\frac{1}{2}\omega^{(1)}{}^{cd\phantom{c}e}_{\phantom{cd}c\phantom{e}abde}-\frac{1}{2}\omega^{(1)}{}^{cd\phantom{c}e}_{\phantom{cd}c\phantom{e}adbe}\nonumber\\
  &&+\frac{1}{2}\omega^{(1)}{}^{cd\phantom{d}e}_{\phantom{cd}d\phantom{e}abce}-\frac{1}{2}\omega^{(1)}{}^{cd\phantom{d}e}_{\phantom{cd}d\phantom{e}aebc}+\frac{1}{4}g^{(0)}_{ab}\left\{2\nabla_d\nu^{(1)}{}^{c\phantom{c}de}_{\phantom{c}c\phantom{de}e}+\nabla_e\nu^{(1)}{}^{c\phantom{c}d\phantom{d}e}_{\phantom{c}c\phantom{d}d}\right\}-\frac{1}{4}\nabla_{(a}\nu^{(1)}{}^{c\phantom{b)c}d}_{\phantom{c}b)c\phantom{d}d}\nonumber\\
  &&+\frac{1}{4}\nabla_c\nu^{(1)}{}_{(ab)\phantom{cd}d}^{\phantom{(ab)}cd}-\frac{1}{2}\nabla_c\nu^{(1)}{}^{c\phantom{ab}d}_{\phantom{c}ab\phantom{d}d}-\nabla_d\nu^{(1)}{}_{(ab)\phantom{c}c}^{\phantom{(ab)}c\phantom{c}d}+\nabla_d\nu^{(1)}{}^{c\phantom{abc}d}_{\phantom{c}abc}+\frac{1}{2}\nabla_d\nu^{(1)}{}^{c\phantom{c(ab)}d}_{\phantom{c}c(ab)}\,.
\end{eqnarray}
Here, we have written
\begin{equation}
  T^{(1)}_{ab}\equiv\wlim_{\lambda\to0}\frac{T_{ab}(\lambda)-T_{ab}^{(0)}}{\lambda} \,.
\label{T1}
\end{equation}
This weak limit exists by virtue of assumption (v) and the fact that 
$g_{ab}(\lambda)$ satisfies Einstein's equation.

Our equation \eqref{eqn:gammaL-1} for $\gamma^{(L)}_{ab}$ takes the
form of a modified linearized Einstein equation. The terms on the left
side are linear in $\gamma^{(L)}_{ab}$ and, in addition to the usual
terms appearing in the linearized Einstein tensor, contain terms
proportional to $\mu_{abcdef}$.  The right side contains, in addition
to the matter source $T^{(1)}_{ab}$, numerous ``effective source
terms'' arising from $\mu^{(1)}_{abcdef}$, $\nu^{(1)}_{abcde}$ and
$\omega^{(1)}_{abcdefgh}$.

Some relations between $\mu^{(1)}_{abcdef}$, $\nu^{(1)}_{abcde}$ and
$\omega^{(1)}_{abcdefgh}$ can be derived from Einstein's equation. We
previously derived the relation
$\alpha_{aeb\phantom{e}cd}^{\phantom{aeb}e}=0$ (see \eqref{aht2}) by
starting with Einstein's equation in ``Ricci form''
\eqref{eqn:RicciRelation2}, multiplying it by $h_{ef}(\lambda)$ and
taking the weak limit as $\lambda\to0$.  In a similar manner, if we
multiply \eqref{eqn:RicciRelation2} by
$h^{(S)}_{cd}(\lambda)h^{(S)}_{ef}(\lambda)/\lambda$, we obtain the
following equation satisfied by $\omega^{(1)}_{abcdefgh}$:
\begin{eqnarray}\label{eqn:omega1}
  \omega^{(1)}{}^g_{\phantom{g}acdefbg}+\omega^{(1)}{}^g_{\phantom{g}bcdefag}-\omega^{(1)}{}^g_{\phantom{g}gcdefab}-\omega^{(1)}{}^{\phantom{abcdef}g}_{abcdef\phantom{g}g}&&\nonumber\\
  +\omega^{(1)}{}^g_{\phantom{g}aefcdbg}+\omega^{(1)}{}^g_{\phantom{g}befcdag}-\omega^{(1)}{}^g_{\phantom{g}gefcdab}-\omega^{(1)}{}_{abefcd\phantom{g}g}^{\phantom{abefcd}g}&=&0\,.
\end{eqnarray}
In deriving this, we used the fact that
\begin{equation}
  \wlim_{\lambda\to0}\frac{1}{\lambda}h^{(S)}_{ab}(\lambda)h^{(S)}_{cd}(\lambda)T_{ef}(\lambda)=0\,,
\end{equation}
which follows from our lemma from section 2 and the assumption that
$T_{ab}(\lambda)$ satisfies the weak energy condition, in the same way
that we showed $\kappa_{abcd}=0$ (see \eqref{rhott2}).

Similarly, if we subtract the background Einstein equation in ``Ricci
form'' from \eqref{eqn:RicciRelation2}, multiply by
$h^{(S)}_{cd}(\lambda)/\lambda$, and take the $\lambda\to0$ weak
limit, we can derive an analog of \eqref{aht2} satisfied by
$\mu^{(1)}_{abcdef}$,
\begin{eqnarray}\label{eqn:mu1}
  &&\frac{1}{2}\mu^{(1)}{}_{abcd\phantom{e}e}^{\phantom{abcd}e}+\frac{1}{2}\mu^{(1)}{}^e_{\phantom{e}cdab}-\mu^{(1)}{}^e_{\phantom{e}(a|cd|b)e}\nonumber\\
  &=&8\pi\kappa^{(1)}{}_{cdab}-4\pi g^{(0)}_{ab}\kappa^{(1)}{}_{cd\phantom{e}e}^{\phantom{cd}e}+\frac{1}{2}\gamma^{(L)}{}^{ef}\left(\mu_{abefcd}+\mu_{efabcd}-2\mu_{e(ab)fcd}\right)+\frac{1}{4}\omega^{(1)}{}_{abcd\phantom{ef}ef}^{\phantom{abcd}ef}\nonumber\\
  &&+\frac{1}{2}\omega^{(1)}{}_{ab\phantom{ef}cdef}^{\phantom{ab}ef}-\frac{1}{2}\omega^{(1)}{}^{\phantom{(a}e\phantom{|cd|b)e}f}_{(a\phantom{e}|cd|b)e\phantom{f}f}-\omega^{(1)}{}^{\phantom{(a}e\phantom{|e|}f}_{(a\phantom{e}|e|\phantom{f}b)fcd}+\frac{1}{4}\omega^{(1)}{}^{e\phantom{ecdab}f}_{\phantom{e}ecdab\phantom{f}f}-\frac{1}{2}\omega^{(1)}{}^{e\phantom{ecda}f}_{\phantom{e}ecda\phantom{f}bf}+\frac{1}{2}\omega^{(1)}{}^{ef}_{\phantom{ef}cdafbe}\nonumber\\
  &&+\frac{1}{2}\omega^{(1)}{}^{ef}_{\phantom{ef}efabcd}+\frac{1}{2}\nabla_{(a}\nu^{(1)}{}^{\phantom{b)cd}e}_{b)cd\phantom{e}e}+\frac{1}{2}\nabla_{(a}\nu^{(1)}{}^e_{\phantom{e}b)ecd}+\frac{1}{2}\nabla_e\nu^{(1)}{}^{\phantom{(ab)}e}_{(ab)\phantom{e}cd}-\frac{1}{2}\nabla_e\nu^{(1)}{}^e_{\phantom{e}abcd}\,.
\end{eqnarray}
In this equation we have defined the quantity
\begin{equation}
  \kappa^{(1)}_{abcd}=\wlim_{\lambda\to0}\frac{1}{\lambda}h_{ab}^{(S)}(\lambda)T_{cd}(\lambda)\,.
\end{equation}
The existence of this limit is guaranteed by Einstein's equation
together with our other assumptions. Note that since
$h^{(S)}_{ab}(\lambda)/\lambda$ does not converge uniformly to zero as
$\lambda\to0$ (although it is uniformly bounded), the weak energy
condition does not imply the vanishing of
$\kappa^{(1)}_{abcd}$ as it did for $\kappa_{abcd}$.

We can simplify the above equations as follows. As with the background
case, define $\alpha^{(1)}_{abcdef}=\mu^{(1)}_{[c|[ab]|d]ef}$ and
$\beta^{(1)}_{abcdef}=\mu^{(1)}_{(abcd)ef}$, and express the equations
in terms of these quantities.  A similar breakup is also possible for
$\omega^{(1)}_{abcdefgh}$.  First, split it into the parts with are
symmetric and antisymmetric in the first two indices,
$\omega^{(1,S)}_{abcdefgh}=\omega^{(1)}_{(ab)cdefgh}$ and
$\omega^{(1,A)}_{abcdefgh}=\omega^{(1)}_{[ab]cdefgh}$.  By
\eqref{eqn:omega-symmetry}, we have
\begin{equation}
  \omega^{(1,A)}_{abcdefgh}=-\omega^{(1,A)}_{abefcdgh}\,.
\end{equation}
The symmetric part is then decomposed into
$\omega^{(1,\alpha)}_{abcdefgh}=\omega^{(1,S)}_{[e|[a|cd|b]|f]gh}$ and
$\omega^{(1,\beta)}_{abcdefgh}=\omega^{(1,S)}_{(ab|cd|ef)gh}$, with
inverse transformation
\begin{eqnarray}
  \omega^{(1,S)}_{abcdefgh}&=&-\frac{4}{3}\left(\omega^{(1,\alpha)}_{e(a|cd|b)fgh}+\omega^{(1,\alpha)}_{g(a|cd|b)hef}-\omega^{(1,\alpha)}_{g(e|cd|f)hab}\right)\nonumber\\
  &&+\omega^{(1,\beta)}_{abcdefgh}+\omega^{(1,\beta)}_{abcdghef}-\omega^{(1,\beta)}_{efcdghab}\,.
\end{eqnarray}
Substituting for our new quantities, \eqref{eqn:omega1} can be
rewritten as
\begin{equation}\label{eqn:omegaalpha}
  \omega^{(1,\alpha)}{}^{\phantom{a}g}_{a\phantom{g}cdbgef}+\omega^{(1,\alpha)}{}^{\phantom{a}g}_{a\phantom{g}efbgcd}=0\,,
\end{equation}
whereas \eqref{eqn:mu1} becomes
\begin{eqnarray}\label{eqn:alpha1}
  \alpha^{(1)}{}^{\phantom{a}e}_{a\phantom{e}becd}&=&4\pi\kappa^{(1)}{}_{cdab}-2\pi g^{(0)}_{ab}\kappa^{(1)}{}^{\phantom{cd}e}_{cd\phantom{e}e}+\alpha_{aebfcd}\gamma^{(L)}{}^{ef}+\frac{1}{4}\omega^{(1,A)}{}^{\phantom{(a}e\phantom{b)ecd}f}_{(a\phantom{e}b)ecd\phantom{f}f}-\frac{1}{2}\omega^{(1,A)}{}^{\phantom{(a}e\phantom{b)}f}_{(a\phantom{e}b)\phantom{f}cdef}\nonumber\\
  &&-\frac{1}{4}\omega^{(1,A)}{}^{ef}_{\phantom{ef}aebfcd}+\frac{1}{2}\omega^{(1,\alpha)}{}^{\phantom{a}e\phantom{cdb}f}_{a\phantom{e}cdb\phantom{f}ef}+\omega^{(1,\alpha)}{}^{\phantom{(a}e\phantom{|cde|}f}_{(a\phantom{e}|cde|\phantom{f}b)f}+\omega^{(1,\alpha)}{}^{\phantom{a}e\phantom{e}f}_{a\phantom{e}e\phantom{f}bfdc}-\frac{1}{4}\omega^{(1,\alpha)}{}^{ef}_{\phantom{ef}baefdc}\nonumber\\
  &&+\frac{1}{4}\nabla_{(a}\nu^{(1)}{}_{b)cd\phantom{e}e}^{\phantom{b)cd}e}+\frac{1}{4}\nabla_{(a}\nu^{(1)}{}^e_{\phantom{e}b)ecd}+\frac{1}{4}\nabla_e\nu^{(1)}{}_{(ab)\phantom{e}cd}^{\phantom{(ab)}e}-\frac{1}{4}\nabla_e\nu^{(1)}{}^e_{\phantom{e}abcd}\,.
\end{eqnarray}
Finally we can use \eqref{eqn:omegaalpha} and \eqref{eqn:alpha1} to
simplify our version of the linearized Einstein equation,
\eqref{eqn:gammaL-1}, resulting in
\begin{eqnarray}\label{eqn:gammaL}
  &&\nabla^c\nabla_{(a}\gamma^{(L)}_{b)c}-\frac{1}{2}\nabla^c\nabla_c\gamma^{(L)}_{ab}-\frac{1}{2}\nabla_a\nabla_b\gamma^{(L)c}_{\phantom{(L)c}c}-\frac{1}{2}g^{(0)}_{ab}\left(\nabla^c\nabla^d\gamma^{(L)}_{cd}-\nabla^c\nabla_c\gamma^{(L)d}_{\phantom{(L)d}d}\right)\nonumber\\
  &&+\frac{1}{2}g^{(0)}_{ab}R^{cd}(g^{(0)})\gamma^{(L)}_{cd}-\frac{1}{2}R(g^{(0)})\gamma^{(L)}_{ab}+\Lambda\gamma^{(L)}_{ab}+2\gamma^{(L)}{}^{cd}\alpha_{(a\phantom{e}b)cde}^{\phantom{(a}e}\nonumber\\
  &=&8\pi T^{(1)}_{ab}+\alpha^{(1)}{}_{a\phantom{c}b\phantom{d}cd}^{\phantom{a}c\phantom{b}d}-2\pi\kappa^{(1)}{}_{ab\phantom{c}c}^{\phantom{ab}c}-8\pi\kappa^{(1)}{}_{(a\phantom{c}b)c}^{\phantom{(a}c}+2\pi g^{(0)}_{ab}\left\{\kappa^{(1)}{}^{c\phantom{c}d}_{\phantom{c}c\phantom{d}d}-\kappa^{(1)}{}^{cd}_{\phantom{cd}cd}\right\}\nonumber\\
  &&-\frac{1}{4}\omega^{(1,A)}{}_{(a\phantom{c}b)\phantom{d}cd\phantom{e}e}^{\phantom{(a}c\phantom{b)}d\phantom{cd}e}+\frac{1}{2}\omega^{(1,A)}{}_{(a\phantom{c}b)\phantom{d}c\phantom{e}de}^{\phantom{(a}c\phantom{b)}d\phantom{c}e}+\frac{1}{8}\omega^{(1,A)}{}^{cd\phantom{abc}e}_{\phantom{cd}abc\phantom{e}de}+\frac{1}{4}\omega^{(1,A)}{}^{cd\phantom{acbd}e}_{\phantom{cd}acbd\phantom{e}e}\nonumber\\
  &&+\frac{1}{16}g^{(0)}_{ab}\left\{\omega^{(1,A)}{}^{cd\phantom{c}e\phantom{de}f}_{\phantom{cd}c\phantom{e}de\phantom{f}f}+2\omega^{(1,A)}{}^{cd\phantom{c}e\phantom{d}f}_{\phantom{cd}c\phantom{e}d\phantom{f}ef}\right\}-2\omega^{(1,\alpha)}{}^{\phantom{(a}ced}_{(a\phantom{ced}b)dce}-\omega^{(1,\alpha)}{}_{(a\phantom{c}b)d\phantom{de}ce}^{\phantom{(a}c\phantom{b)d}de}-\omega^{(1,\alpha)}{}_{(a\phantom{ce}b)\phantom{d}cde}^{\phantom{(a}ce\phantom{b)}d}\nonumber\\
  &&+\frac{1}{8}\omega^{(1,\alpha)}{}^{cd\phantom{abcd}e}_{\phantom{cd}abcd\phantom{e}e}+\frac{1}{4}\omega^{(1,\alpha)}{}^{cd\phantom{bac}e}_{\phantom{cd}bac\phantom{e}de}+\frac{1}{8}g^{(0)}_{ab}\left\{\omega^{(1,\alpha)}{}^{c\phantom{dec}edf}_{\phantom{c}dec\phantom{edf}f}-2\omega^{(1,\alpha)}{}^{c\phantom{dec}efd}_{\phantom{c}dec\phantom{efd}f}\right\}\nonumber\\
  &&-\frac{1}{2}\nabla_{(a}\nu^{(1)}{}^{c\phantom{b)c}d}_{\phantom{c}b)c\phantom{d}d}+\frac{1}{4}\nabla_{(a}\nu^{(1)}{}^{c\phantom{b)}d}_{\phantom{c}b)\phantom{d}cd}+\frac{1}{4}\nabla_c\nu^{(1)}{}^{c\phantom{ab}d}_{\phantom{c}ab\phantom{d}d}-\frac{3}{4}\nabla_d\nu^{(1)}{}^{\phantom{(ab)}c\phantom{c}d}_{(ab)\phantom{c}c}+\frac{1}{4}\nabla_d\nu^{(1)}{}^{c\phantom{abc}d}_{\phantom{c}abc}\nonumber\\
  &&-\frac{1}{2}\nabla_d\nu^{(1)}{}^{c\phantom{(a}d}_{\phantom{c}(a\phantom{d}b)c}+\frac{1}{4}g^{(0)}_{ab}\left\{\nabla_d\nu^{(1)}{}^{c\phantom{c}de}_{\phantom{c}c\phantom{de}e}+\nabla_e\nu^{(1)}{}^{c\phantom{c}d\phantom{d}e}_{\phantom{c}c\phantom{d}d}\right\}\,.
\end{eqnarray}
Equations \eqref{eqn:omegaalpha}, \eqref{eqn:alpha1} and
\eqref{eqn:gammaL} describe the long
wavelength perturbations.  It should be noted that,
just as $\beta_{abcdef}$ was absent from our background equations,
$\beta_{abcdef}$, $\beta^{(1)}_{abcdef}$, and
$\omega^{(1,\beta)}_{abcdefgh}$ are all absent from our perturbation
equations.

Equations \eqref{eqn:omegaalpha}, \eqref{eqn:alpha1} and
\eqref{eqn:gammaL} have been written down in an arbitrary gauge. We
will make a specific choice of gauge in subsection IVA below, but for
now we note that we can apply any one-parameter family of
diffeomorphisms, $\phi_\lambda$, to $g_{ab}(\lambda)$ that preserves
conditions (i)--(v). Burnett \cite{Burnett:1989gp} has analyzed the
properties of gauge transformations associated with one-parameter
families of diffeomorphisms that are not smooth in $\lambda$. Here, we
simply note that any smooth, one-parameter group of diffeomorphisms
$\phi_\lambda$ generates gauge transformations that are easily seen to
preserve conditions (i)--(v). Under such gauge transformations, it is
not difficult to see that $\gamma^{(L)}_{ab} \to \gamma^{(L)}_{ab} +
{\mathcal L}_\xi g^{(0)}_{ab}$, where $\xi^a$ is the vector field that
generates $\phi_\lambda$ and $ {\mathcal L}$ denotes the Lie
derivative.  Thus, $\gamma^{(L)}_{ab}$ has the same gauge freedom
arising from smooth $\phi_\lambda$ as in ordinary linearized
perturbation theory. This freedom can be used to impose the same types
of gauge conditions on $\gamma^{(L)}_{ab}$ as in ordinary linearized
perturbation theory. It is also not difficult to see that
$h^{(S)}_{ab}(\lambda) \to \phi^*_\lambda h^{(S)}_{ab}(\lambda) +
j_{ab}(\lambda)$, where $j_{ab}(\lambda) = O(\lambda^2)$ and is
jointly smooth in $\lambda$ and the spacetime point. By using this
gauge transformation property of $h^{(S)}_{ab}(\lambda)$, it is
possible to show that $\mu_{abcdef}$, $\nu^{(1)}_{abcde}$,
$\omega^{(1)}_{abcdefgh}$, and $ \kappa^{(1)}_{abcd}$ are gauge
invariant under gauge transformations arising from smooth
$\phi_\lambda$, whereas $\mu^{(1)}_{abcdef} \to \mu^{(1)}_{abcdef} +
{\mathcal L}_\xi \mu_{abcdef}$ and $T^{(1)}_{ab} \to T^{(1)}_{ab} +
{\mathcal L}_\xi T^{(0)}_{ab}$.

We turn our attention now to the short wavelength perturbations.
Without making any approximations it is straightforward to write down
an equation satisfied by $h^{(S)}_{ab}(\lambda)$: Simply substitute
$g_{ab}(\lambda)=g^{(0)}_{ab}+\lambda\gamma^{(L)}_{ab}+h^{(S)}_{ab}(\lambda)$
into the exact Einstein equation.  We may write this equation in the
form
\begin{eqnarray}
  G^{(1)}_{ab}(g^{(0)},h^{(S)}(\lambda))+\Lambda h_{ab}^{(S)}(\lambda)&=&8\pi T_{ab}(\lambda)-G_{ab}(g^{(0)})-\Lambda g^{(0)}_{ab}-\lambda G_{ab}^{(1)}(g^{(0)},\gamma^{(L)})\nonumber\\
  &&-\lambda\Lambda\gamma^{(L)}_{ab}-\sum_{n=2}^{\infty}G_{ab}^{(n)}(g^{(0)},\lambda\gamma^{(L)}+h^{(S)}(\lambda))\,,
\label{hSeq}
\end{eqnarray}
where we have grouped linear terms in $h^{(S)}_{ab}(\lambda)$ on the
left hand side. Here,
$G_{ab}^{(n)}(g^{(0)},\lambda\gamma^{(L)}+h^{(S)}(\lambda))$ denotes
the $n$th order Einstein tensor expanded about $g^{(0)}_{ab}$ of the
perturbation $\lambda\gamma^{(L)}_{ab} +h^{(S)}_{ab}(\lambda)$.

Unfortunately, it does not appear possible to simplify \eqref{hSeq} to
obtain suitable approximate solutions without introducing additional
assumptions. Of course, if we do not simplify \eqref{hSeq}, then we
have not made any progress beyond asserting that we must solve
Einstein's equation.  In the next section, we will introduce
additional assumptions relevant to the case of cosmological
perturbations and argue that to obtain an accurate description of the
metric to order $\lambda$, we may replace \eqref{hSeq} by the
equations of Newtonian gravity with local matter sources.

For the remainder of this section, we shall compare our general
analysis to that given by Isaacson
\cite{Isaacson:1967zz,Isaacson:1968zza} and others (see, e.g.,
\cite{MTW}), who were interested in describing the self-gravitating
effects of gravitational radiation. We therefore restrict attention to
the vacuum case ($T_{ab}(\lambda) = 0$). These authors work with the
quantity $h_{ab}(\lambda)$ rather than introducing $\mu_{abcdef}$.
Suppose one is merely interested in determining the background metric
$g^{(0)}_{ab}$, i.e., one is not interested in obtaining an accurate
(to order $\lambda$) description of the deviation of the metric from
$g^{(0)}_{ab}$. Then one would need only to calculate
$h_{ab}(\lambda)$ to sufficient accuracy that one could determine
$\mu_{abcdef}$ and, thereby, $t^{(0)}_{ab}$ (see
\eqref{eqn:BackgroundEinstein-messy}).  In particular, one would not
be interested in computing $\gamma^{(L)}_{ab}$, so one could ignore
the equations we have derived above for
$\gamma^{(L)}_{ab}$. Furthermore, in order to obtain
$h^{(S)}_{ab}(\lambda)$ to sufficient accuracy, it appears plausible
that one could make $O(1)$ modifications to \eqref{hSeq} as $\lambda
\to 0$ and still determine $h^{(S)}_{ab}(\lambda)$ to sufficient
accuracy, provided that these $O(1)$ modifications have vanishing weak
limit. Here, by the phrase ``it appears plausible'' we mean that we
believe it is likely that one could introduce additional reasonable
assumptions on the one-parameter family $g_{ab}(\lambda)$ so that
these modifications to \eqref{hSeq} could be made without affecting
$g^{(0)}_{ab}$.  The reason for this belief is that $O(1)$ error terms
in \eqref{hSeq} should---under suitable further assumptions similar to
ones indicated at the end of subsection IVB below---give rise to
$O(\lambda^2)$ errors in $h^{(S)}_{ab}(\lambda)$, which should not
affect $\mu_{abcdef}$.

A candidate modification of \eqref{hSeq} in the vacuum case would be
to drop the entire right side of this equation (together with the term
$\Lambda h_{ab}^{(S)}(\lambda)$ on the left side), to obtain simply
the linearized Einstein equation for $h_{ab}^{(S)}$ off of
$g^{(0)}_{ab}$,
\begin{equation}
G^{(1)}_{ab}(g^{(0)},h^{(S)}(\lambda)) = 0 \,.
\end{equation}
However, if $G_{ab}(g^{(0)}) \neq 0$, the linearized Einstein equation
off of $g^{(0)}_{ab}$ does not appear to have an initial value
formulation, so this modification of \eqref{hSeq} is probably not
suitable. (Note that the linearized Einstein equation off of a
non-solution is not gauge invariant, so one cannot employ a choice of
gauge to simplify the equation and/or put it in hyperbolic form.) A
better candidate modification would be an equation of the same form
that the linearized Einstein equation would take in the Lorenz gauge
if perturbed off of a vacuum spacetime. Since the linearized Einstein
equation off of a non-vacuum spacetime is inconsistent with the Lorenz
gauge condition, there is no unique choice of such an equation---in
particular, one can add new terms involving the background Ricci
tensor---and, indeed, Isaacson \cite{Isaacson:1967zz,Isaacson:1968zza}
and Misner, Thorne, and Wheeler \cite{MTW} give slightly different
forms of the proposed equation (compare Eq.~(5.12) of
\cite{Isaacson:1967zz} with Eq.~(35.68) of \cite{MTW}). In fact, since
terms involving the product of the background curvature with
$h_{ab}^{(S)}$ are $O(\lambda)$, it would appear simplest to drop all
of these terms and work with the wave equation
\begin{equation}
\nabla^c \nabla_c \bar{h}_{ab}^{(S)} = 0 \,,
\label{hSwave}
\end{equation}
where
$\bar{h}^{(S)}_{ab}=h_{ab}^{(S)}-\frac{1}{2}g^{(0)}_{ab}h^{(S)}{}^c_{\phantom{c}c}$.
As with the equations used in \cite{Isaacson:1967zz} and \cite{MTW},
this equation is inconsistent with the Lorenz gauge condition
$\nabla^a \bar{h}_{ab}^{(S)} = 0$.  However, if one constrains the
initial data for $\bar{h}_{ab}^{(S)}$ for solutions to \eqref{hSwave}
so that $\nabla^a \bar{h}_{ab}^{(S)}$ and its first time derivative
vanish initially, then the Lorenz gauge condition should hold to
$O(\lambda)$ at later times (in compact regions of spacetime). Thus,
the solutions to \eqref{hSwave} with these initial data restrictions
should satisfy Einstein's equation \eqref{hSeq} to the desired $O(1)$
accuracy.  It should also be possible to impose the gauge condition
${h^{(S)a}}_a = 0$ to $O(\lambda)$ accuracy.

If the above arguments are correct, then in order to obtain the
possible background metrics $g_{ab}^{(0)}$, it should suffice to
simultaneously solve \eqref{Eeg0} and \eqref{hSwave} (or equivalently,
Eq.~(5.12) of \cite{Isaacson:1967zz} or Eq.~(35.68) of \cite{MTW})
with appropriate restrictions on initial data, to obtain
$g_{ab}^{(0)}$ and a one-parameter family $h_{ab}^{(S)}(\lambda)$
satisfying our conditions (ii)--(iv), where, in \eqref{Eeg0},
$t_{ab}^{(0)}$ is given by \eqref{eqn:BackgroundEinstein-messy} and
$\mu_{abcdef}$ is given by \eqref{mu} with $h_{ab}(\lambda)$ replaced
by our one-parameter family of solutions to \eqref{hSwave}. One could
then attempt to establish properties of $t_{ab}^{(0)}$---in
particular, the vanishing of its trace and the positivity of its
energy density---by working with the expressions for it in terms of
$h_{ab}^{(S)}(\lambda)$ in a particular gauge. Needless to say,
numerous extremely murky mathematical issues arise if one proceeds in
this manner.  Our analysis of section II, as well as the work of
Burnett \cite{Burnett:1989gp} in the vacuum case, proved rigorous
results about $t_{ab}^{(0)}$ without introducing any approximate
equations satisfied by $h_{ab}^{(S)}(\lambda)$, and thus completely
bypassed these murky issues.

Finally, we note that if one wished to know the deviation of the
metric from $g^{(0)}_{ab}$ to $O(\lambda)$---as would not necessarily
be of interest in studying gravitational radiation but is of
considerable interest in cosmology---then, of course, it would be
necessary to know $\gamma^{(L)}_{ab}$. Although we argued above that
one has considerable freedom in modifying the equation satisfied by
$h_{ab}^{(S)}$ without affecting $g^{(0)}_{ab}$, it is clear that one
cannot make any modification to the equation \eqref{eqn:gammaL}
satisfied by $\gamma^{(L)}_{ab}$ without introducing $O(1)$ errors in
$\gamma^{(L)}_{ab}$ and thus $O(\lambda)$ errors in
$g_{ab}(\lambda)$. Furthermore, in order to calculate the effective
source term
$\alpha^{(1)}{}_{a\phantom{c}b\phantom{d}cd}^{\phantom{a}c\phantom{b}d}$
in \eqref{eqn:gammaL}, it appears that one would need to know
$h_{ab}^{(S)}$ to $O(\lambda^2)$, in which case one could not drop the
quadratic terms in $h_{ab}^{(S)}$ in \eqref{hSeq}. Thus, in the case
of self-gravitating gravitational radiation, if one wished to know the
deviation of the metric from $g^{(0)}_{ab}$ to $O(\lambda)$, one would
have to solve \eqref{Eeg0}, \eqref{eqn:gammaL}, and some suitable
simplification of \eqref{hSeq}. This would comprise an extremely
complicated system.  Fortunately, in the case of cosmology, we will
now argue that, under additional assumptions, significant
simplifications occur.

\section{Cosmological Perturbation Theory}

Up to this point we have not assumed any symmetries or other special
properties of the background metric $g^{(0)}_{ab}$. We also have not
made any restrictions on the matter content, $T_{ab} (\lambda)$, other
than that it satisfy the weak energy condition, nor have we imposed
any restrictions on the perturbations. In this section, we will be
concerned with the case of main interest in cosmology, where
$g^{(0)}_{ab}$ has FLRW symmetry, there is negligible gravitational
radiation content (in particular, $t^{(0)}_{ab}=0$), and the matter
content satisfies suitable Newtonian assumptions. We will argue that,
under these assumptions, to leading order in $\lambda$,
$h^{(S)}_{ab}(\lambda)$ is given by local Newtonian gravity, i.e., in
a neighborhood of any point $x$, $h^{(S)}_{ab}$ can be calculated to
sufficient accuracy using Newtonian gravity, taking into account only
the matter distribution within a suitable neighborhood of $x$. (The
effects of more distant matter are taken into account by
$\gamma^{(L)}_{ab}$.) With our ``no gravitational radiation''
assumption for the background and our Newtonian approximation for
$h^{(S)}_{ab}$, our equation \eqref{eqn:gammaL} for
$\gamma^{(L)}_{ab}$ simplifies considerably, yielding the linearized
Einstein equation with an additional effective source that agrees with
recent results of \cite{Baumann:2010tm} (see also
\cite{Futamase:1996fk}).

In our analysis, it will be important to make a convenient choice of
gauge. Since we are taking nonlinear effects at small scales into
account, we cannot simply impose the usual cosmological gauge choices
for perturbations, i.e., we must make our gauge choice at the
nonlinear level.  When studying non-linear perturbations off of a flat
background it is often convenient to work with ``wave map
coordinates'' (usually called ``harmonic coordinates'' in the
literature), particularly in the context of the post-Newtonian
expansion \cite{lrr-2007-2}.  Since our background metric
$g_{ab}^{(0)}$ is not flat, the usual definition of wave-map
coordinates is not convenient, but we may instead impose a generalized
wave-map gauge condition with respect to the background metric
$g_{ab}^{(0)}$ \cite{Choquet-Bruhat:2009}.

Our gauge choice will be introduced in subsection A in the context of
a general background metric $g_{ab}^{(0)}$ (i.e., without assuming
FLRW symmetry). In subsection B, we restrict to a FLRW background, we
make our Newtonian assumptions, and argue that $h^{(S)}_{ab}$ is given
by local Newtonian gravity. It should be emphasized that the arguments
of section B have the character of plausibility arguments rather than
proofs. Finally, the simplifications to the equation for
$\gamma^{(L)}_{ab}$ will be obtained in subsection C.

\subsection{Generalized wave map (harmonic) gauge}

Given a one-parameter family of metrics $g_{ab} (\lambda)$ on our
spacetime manifold $M$ that satisfies our assumptions (i)--(v), we may
apply any one-parameter family of diffeomorphisms, $\phi(\lambda):M\to
M$, which preserve these conditions, where without loss of generality,
we may assume that $\phi(0)$ is the identity map.  As already noted in
the paragraph below \eqref{eqn:gammaL}, it is clear that any
$\phi(\lambda)$ that is jointly smooth in $\lambda$ and the spacetime
point $x$ will preserve conditions (i)--(v), but there also should
exist a wide class of $\phi(\lambda)$ that are not smooth in $\lambda$
that preserve these conditions.  The properties of such
$\phi(\lambda)$ and the transformations that they induce on
$\mu_{abcdef}$ were analyzed by \cite{Burnett:1989gp} (under his
assumptions, which differ slightly from our assumptions
(i)--(iv)). Unfortunately, although Burnett's analysis can be used to
prove important properties of gauge transformations, such as the
invariance of $t^{(0)}_{ab}$ under all allowed gauge transformations,
it is very difficult to prove any existence results that establish
that specific gauge conditions can be imposed on an arbitrary one
parameter family of metrics $g_{ab} (\lambda)$ satisfying our
conditions.

In this section, we will assume that we can impose the ``generalized
wave-map gauge condition'' on the metric $g_{ab} (\lambda)$, namely
\begin{equation}
  g^{ab}(\lambda)C^c_{\phantom{c}ab} (\lambda)=0\,,
  \label{wmg}
\end{equation}
where $C^c_{\phantom{c}ab} (\lambda)$ is given by \eqref{Ccab}. 
Note that this condition depends upon the background metric, $g_{ab}^{(0)}$,
since the derivative operator, $\nabla_a$, of $g_{ab}^{(0)}$ appears
in the definition of $C^c_{\phantom{c}ab} (\lambda)$.
We can impose the gauge condition \eqref{wmg} on $g_{ab} (\lambda)$
by applying the diffeomorphism 
$x^\mu \to \phi^\mu(\lambda,x)$ to $g_{ab} (\lambda)$, 
where $x^\mu$ are arbitrarily chosen
coordinates on $M$ and $\phi^\mu$ 
satisfies\footnote{The diffeomorphisms defined by
\eqref{wmg2} do not depend on the choice of coordinates $x^\mu$,
since this equation can be derived from the coordinate invariant action 
\begin{equation}
  E[\phi]=\int_M g^{\mu\nu}(\lambda)g_{\bar{\mu}\bar{\nu}}^{(0)}\frac{\partial \phi^{\bar{\mu}}}{\partial x^\mu}\frac{\partial \phi^{\bar{\nu}}}{\partial x^\nu} \sqrt{-g(\lambda)}\ud^4 x\,. \nonumber
\end{equation}}
\begin{equation}
  \frac{1}{\sqrt{-g(\lambda)}}\frac{\partial}{\partial x^\mu}\left(\sqrt{-g(\lambda)}g^{\mu\nu}(\lambda)\frac{\partial \phi^{\bar{\alpha}}}{\partial x^\nu}\right)+\Gamma^{(0)}{}^{\bar{\alpha}}_{\phantom{\alpha}\bar{\mu}\bar{\nu}}g^{\mu\nu}(\lambda)\frac{\partial \phi^{\bar{\mu}}}{\partial x^\mu}\frac{\partial \phi^{\bar{\nu}}}{\partial x^\nu}=0\,.
\label{wmg2}
\end{equation}
This equation is a (nonlinear) wave equation, so we can always find
(local) solutions.  To see that solutions to this equation give rise
to the condition \eqref{wmg}, we note that if we use
$\phi^\mu(\lambda,x)$ as coordinates for the $\lambda$th spacetime,
then we may replace $\partial \phi^\alpha/\partial x^\beta$ by
${\delta^\alpha}_\beta$, and \eqref{wmg2} reduces to
\begin{eqnarray}
  0&=&\frac{1}{\sqrt{-g(\lambda)}}\partial_\mu\left(\sqrt{-g(\lambda)}g^{\mu\alpha}(\lambda)\right)+\Gamma^{(0)}{}^\alpha_{\phantom{\alpha}\mu\nu}g^{\mu\nu}(\lambda)\nonumber\\
  &=&g^{\mu\nu}(\lambda)\Gamma^{\alpha}_{\phantom{\alpha}\mu\nu}(\lambda)-g^{\mu\nu}(\lambda)\Gamma^{(0)}{}^{\alpha}_{\phantom{\alpha}\mu\nu}\nonumber\\
  &=&-g^{\mu\nu}(\lambda)C^\alpha_{\phantom{\alpha}\mu\nu}\,.
\end{eqnarray}
We will refer to the coordinates $\phi^\mu(\lambda,x)$ as {\it
  generalized wave map coordinates} for $g_{ab}(\lambda)$ relative to
the coordinates $x^\mu$ for $g_{ab}^{(0)}$.  Note that in the case
where $g_{ab}^{(0)}=\eta_{ab}$ and $x^\mu$ are Minkowski coordinates,
the generalized wave map gauge condition reduces to the condition
$g^{\mu\nu}(\lambda)\Gamma^{\alpha}_{\phantom{\alpha}\mu\nu}(\lambda)=0$, and the
coordinates $\phi^\mu(\lambda,x)$ satisfy a linear wave equation. Such
coordinates are usually referred to as ``harmonic coordinates'' in the
literature.

Although we can always (locally) solve \eqref{wmg2} and thus (locally)
put each $g_{ab} (\lambda)$ in our one-parameter family in wave map
gauge, we have no guarantee that the resulting new one-parameter
family of metrics will satisfy our conditions (i)--(v). In the
following, we shall simply assume that this is the case, i.e., that we
have a one-parameter family of metrics $g_{ab} (\lambda)$ that
satisfies conditions (i)--(v) as well as our gauge condition
\eqref{wmg}. This corresponds to a strengthening of our assumptions.

When the generalized wave map gauge condition is satisfied, 
it is very convenient to work with the variable 
\begin{equation}
  \hh^{ab}(\lambda)\equiv g^{(0)}{}^{ab}-\sqrt{\frac{g(\lambda)}{g^{(0)}}}g^{ab}(\lambda)\,.
\end{equation}
instead of $h_{ab}(\lambda)=g_{ab}(\lambda)-g^{(0)}_{ab}$. (Note that
$\sqrt{g(\lambda)/g^{(0)}}$ is the proportionality factor 
between the volume elements of $g_{ab} (\lambda)$ and $g_{ab}^{(0)}$,
so this quantity does not depend on any choice of coordinates $x^\mu$ on $M$.)
We have
\begin{eqnarray}
  \nabla_b \hh^{ab}(\lambda)&=&-\sqrt{\frac{g(\lambda)}{g^{(0)}}}\nabla_bg^{ab}(\lambda)-\sqrt{\frac{g(\lambda)}{g^{(0)}}}g^{ab}(\lambda)C^c_{\phantom{c}cb}\nonumber\\
&=&-\sqrt{\frac{g(\lambda)}{g^{(0)}}}\left(\nabla_bg^{ab}(\lambda)+\frac{1}{2}g^{ab}(\lambda)g^{cd}(\lambda)\nabla_bg_{cd}(\lambda)\right)\nonumber\\
&=&\sqrt{\frac{g(\lambda)}{g^{(0)}}}g^{bc}(\lambda)C^a_{\phantom{a}bc} \nonumber \\
&=& 0 \, , 
\label{hhwmg}
\end{eqnarray}
where on the first line we used the fact that
\begin{equation}
  C^b_{\phantom{b}ba}=\frac{1}{2}\nabla_a\log\left(\frac{g(\lambda)}{g^{(0)}}\right)\,.
\end{equation}
Note that in linearized gravity, $\hh^{ab}$ reduces to the
trace-reversed metric perturbation, i.e.,
$\hh^{ab}(\lambda)\to\bar{h}^{ab}(\lambda)=h^{ab}(\lambda)-\frac{1}{2}g^{(0)}{}^{ab}h^c_{\phantom{c}c}(\lambda)$,
and \eqref{hhwmg} reduces to the Lorenz gauge condition $\nabla_a
\bar{h}^{ab} = 0$.  One should keep in mind, though, that beyond
lowest order in $\lambda$, $\hh^{ab}(\lambda)$ is not the
trace-reversed metric perturbation.

We now express the exact Einstein equation in wave-map
gauge in terms of $\hh^{ab}$ and the background derivative
operator.  Starting with \eqref{eqn:EinsteinRelation}, we use the
background equation \eqref{Eeg0}, the gauge condition
$\nabla_b\hh^{ab}=0$, as well as the fact that
\begin{equation}
  \nabla_ag^{bc}(\lambda)=-\sqrt{\frac{g^{(0)}}{g(\lambda)}}\left(\nabla_a \hh^{bc}-\frac{1}{2}g^{bc}(\lambda)g_{de}(\lambda)\nabla_a \hh^{de}\right)\,,
\end{equation}
to show (after a long computation) that
\begin{eqnarray}\label{eqn:h-wavemap}
  &&\nabla^c\nabla_c \hh^{ab} - 2 R^{a\phantom{cd}b}_{\phantom{a}cd}(g^{(0)})\hh^{cd}+2R_c^{\phantom{c}(a}(g^{(0)})\hh^{b)c}-R(g^{(0)})\hh^{ab}-g^{(0)}{}^{ab}R_{cd}(g^{(0)})\hh^{cd}\nonumber\\
  &&-2G^{ab}(g^{(0)})g^{(0)}_{cd}\hh^{cd}+2\Lambda\left(\hh^{ab}-\frac{1}{2}g^{(0)}{}^{ab}g^{(0)}_{cd}\hh^{cd}\right)\nonumber\\
  &=&-16\pi\frac{g(\lambda)}{g^{(0)}}\left(T^{ab}(\lambda)-T^{(0)}{}^{ab}+\mathfrak{t}^{ab}(\lambda)-t^{(0)}{}^{ab}\right)\,.
\end{eqnarray}
Here, the terms that are non-linear in $\hh^{ab}$ have been absorbed into
the quantity,
\begin{eqnarray}
  16\pi\frac{g(\lambda)}{g^{(0)}}\mathfrak{t}^{ab}(\lambda)&\equiv&-2G^{ab}\left(1-\frac{g(\lambda)}{g^{(0)}}-g^{(0)}_{cd}\hh^{cd}\right)+R_{cd}(g^{(0)})\hh^{cd}\hh^{ab}-2R_{ced}^{\phantom{ced}(a}(g^{(0)})\hh^{b)e}\hh^{cd}\nonumber\\
  &&-\hh^{cd}\nabla_c\nabla_d\hh^{ab}+\nabla_d\hh^{ca}\nabla_c\hh^{db}+g^{ef}(\lambda)g_{cd}(\lambda)\nabla_e\hh^{ca}\nabla_f\hh^{db}\nonumber\\
  &&+\frac{1}{2}g^{ab}(\lambda)g_{cd}(\lambda)\nabla_e\hh^{fc}\nabla_f\hh^{ed}-2g_{cd}(\lambda)g^{f(a}(\lambda)\nabla_e\hh^{b)c}\nabla_f\hh^{ed}\nonumber\\
  &&+\frac{1}{8}\left(2g^{ag}(\lambda)g^{bh}(\lambda)-g^{ab}(\lambda)g^{gh}(\lambda)\right)\left(2g_{cd}(\lambda)g_{ef}(\lambda)-g_{ed}(\lambda)g_{cf}(\lambda)\right)\nabla_h\hh^{ed}\nabla_g\hh^{cf}\nonumber\\
  &&-2\Lambda\left(\frac{g(\lambda)}{g^{(0)}}\left(g^{ab}(\lambda)-g^{(0)}{}^{ab}\right)+\hh^{ab}-\frac{1}{2}g^{(0)}{}^{ab}g^{(0)}_{cd}\hh^{cd}\right)\,.
\label{matht}
\end{eqnarray}

Equation \eqref{eqn:h-wavemap} is of the form
\begin{equation}
\mathscr{L}^{ab}(\hh)=-16\pi S^{ab} \, .
\label{LS}
\end{equation}
where $\mathscr{L}^{ab}$
takes the form of a linear wave operator acting on $\hh^{ab}$.  
Consequently, re-introducing our (arbitrarily chosen)
coordinates $x^\mu$ of the background spacetime,
we may rewrite \eqref{eqn:h-wavemap} in the following
equivalent integral form:
\begin{equation}\label{eqn:h-integral}
  \hh^{\alpha\beta}(\lambda,x)=4\int_M G_{\mathrm{ret}\phantom{\alpha\beta}\mu'\nu'}^{\phantom{\mathrm{ret}}\alpha\beta}(x,x')S^{\mu'\nu'}(\lambda,x')\sqrt{-g^{(0)}(x')}\ud^4x'+\hh^{\alpha\beta}_{\mathrm{hom}}(\lambda,x)\,,
\end{equation}
where 
$G_{\mathrm{ret}\phantom{\alpha\beta}\mu'\nu'}^{\phantom{\mathrm{ret}}\alpha\beta}(x,x')$
is the retarded Green's function for $\mathscr{L}^{ab}$
and $\hh^{\alpha\beta}_{\text{hom}}$ is a solution to
$\mathscr{L}^{ab}(\hh_{\text{hom}})=0$.
We emphasize that \eqref{eqn:h-integral} is not a {\em
  solution} to \eqref{eqn:h-wavemap} since the source
$S^{\alpha\beta}$ depends on $\hh^{\alpha\beta}$.  Rather, it is
simply a re-writing of \eqref{eqn:h-wavemap} in an integral form.

\subsection{Local Newtonian gravity}

For the remainder of this section, we restrict attention to the case
where the background spacetime $g_{ab}^{(0)}$ has FLRW symmetry. It
will be useful to work in coordinates where the metric components are
nonsingular. Thus, instead of the more common choice of polar-type
coordinates, we will write the background metric in the form
\begin{equation}
  ds^{(0)}{}^2=-d\tau^2+a^2(\tau) \left(1+k(x^2+y^2+z^2)/4\right)^{-2} [dx^2 + dy^2 +dz^2]\,,
\label{FLRW}
\end{equation}
where $k=0,\pm1$, depending on the spatial curvature.

The ``long wavelength part'' of the leading order 
in $\lambda$ part of $\hh^{ab}(\lambda)$ is given by
\begin{equation}
\bar{\gamma}^{(L)}{}^{ab}\equiv\wlim_{\lambda\to0}\frac{\hh^{ab}}{\lambda} = \gamma^{(L)ab} - \frac{1}{2} g^{(0)ab} {\gamma^{(L)c}}_c\,.
\label{gammabar}
\end{equation}
It follows directly from \eqref{hhwmg} that 
$\bar{\gamma}^{(L)}{}^{ab}$ satisfies the Lorenz gauge condition
\begin{equation}
\nabla_a \bar{\gamma}^{(L)}{}^{ab} = 0 \, .
\label{lorenz}
\end{equation}
The ``short wavelength part'' of $\hh^{ab}(\lambda)$ is given by
\begin{equation}
\hh^{(S)}{}^{ab}(\lambda) = \hh^{ab}(\lambda) - \lambda \bar{\gamma}^{(L)}{}^{ab}\,.
\end{equation}
By \eqref{eqn:h-wavemap}, it satisfies (in the notation of \eqref{LS})
\begin{equation}
  \mathscr{L}^{ab}(\hh^{(S)}(\lambda))=-16\pi\left(S^{ab}(\lambda)-\lambda\wlim_{\lambda'\to0}\frac{S^{ab}(\lambda')}{\lambda'}\right)\,,
\label{hhSwmg}
\end{equation}
where the weak limit appearing in this equation exists by virtue of
assumption (v) of section III.  Note that $S^{ab}(\lambda)$ still
depends on the full perturbation, i.e., $\hh^{ab}$ cannot be replaced
by $\hh^{(S)}{}^{ab}$ in the expression for $S^{ab}$.  Equation
\eqref{hhSwmg} can be rewritten in integral form as
\begin{eqnarray}\label{eqn:hS-wavemap-integral}
  \hh^{(S)}{}^{\alpha\beta}(\lambda,x)&=&4\int_M G_{\mathrm{ret}\phantom{\alpha\beta}\mu'\nu'}^{\phantom{\mathrm{ret}}\alpha\beta}(x,x')\left(S^{\mu'\nu'}(\lambda,x')-\lambda\wlim_{\lambda'\to0}\frac{S^{\mu'\nu'}(\lambda',x')}{\lambda'}\right)\sqrt{-g^{(0)}(x')}\ud^4x'\nonumber\\
  &&+\hh^{(S)}_{\text{hom}}{}^{\alpha\beta}(\lambda,x)\,.
\end{eqnarray}

Our aim for the remainder of this subsection is to argue that---in the
absence of gravitational radiation and under suitable assumptions
concerning the behavior of the matter distribution
$T_{ab}(\lambda)$---to leading order in $\lambda$,
$\hh^{(S)}{}^{ab}(\lambda)$ near point $x$ is described by Newtonian
gravity, taking into account only the matter distribution in a
suitable local neighborhood of $x$. However, in order to derive this
conclusion, we must make significant additional assumptions about our
one-parameter family $g_{ab}(\lambda)$, and severe difficulties arise
if one attempts to formulate these assumptions in a mathematically
precise manner.  The reason is that, although there are simple,
precise limits that one can take of general relativistic spacetimes to
obtain Newtonian gravity
\cite{Futamase:1983,Oliynyk:2009qa,Oliynyk:2009bc}, these limits would
not be compatible with our assumptions (unless $h_{ab}=O(\lambda^2)$),
and thus would not be suitable for our use.  We believe that it should
be possible to concoct a mathematically consistent set of assumptions
that would enable us to rigorously justify our conclusions below, but
we do not see a simple and/or elegant way of formulating such
assumptions, and we do not feel that it would be illuminating to
attempt to derive our results from a complicated list of technical
assumptions whose intrinsic plausibility is not much greater than that
of the conclusions we wish to draw. Thus, in the discussion below in
this subsection, although we will clearly indicate the nature of the
assumptions that are needed, we will not attempt to formulate all of
our assumptions in a mathematically precise manner, and we will
thereby resort to plausibility arguments to obtain our conclusions.

We are interested in obtaining $\hh^{(S)}{}^{ab}(\lambda)$ near a
point $x$ at a time roughly corresponding to the present time in the
actual universe.  We first argue that although the retarded Green's
function integral extends all the way back to the ``big bang'', it
should suffice to integrate only over the ``recent universe''
(corresponding, say, to $z \le 1000$ in the present, actual
universe). There are two reasons why this should be so: (1) The
universe is expected to be very nearly homogeneous and isotropic in
the distant past, so the source term
$\left(S^{ab}-\lambda\wlim_{\lambda'\to0}\left[S^{ab}/\lambda'\right]\right)$
should be negligibly small. (2) The nature of the retarded Green's
function in an expanding universe is such as to make the influence of
distant sources small (on account of redshift and intensity
diminution).  Similarly, we assume that the gravitational radiation
content of the present universe arising from the ``big bang'' is
negligible. Consequently, we discard the last term,
$\hh^{(S)}_{\text{hom}}{}^{\alpha\beta}$, in
\eqref{eqn:hS-wavemap-integral}.  Thus, our integral relation for
$\hh^{(S)}{}^{ab}$ becomes
\begin{equation}\label{eqn:hS-wavemap-integral2}
  \hh^{(S)}{}^{\alpha\beta}(\lambda,x) = 4\int_{\mathcal W}G_{\mathrm{ret}\phantom{\alpha\beta}\mu'\nu'}^{\phantom{\mathrm{ret}}\alpha\beta}(x,x')\left(S^{\mu'\nu'}(\lambda,x')-\lambda\wlim_{\lambda'\to0}\frac{S^{\mu'\nu'}(\lambda',x')}{\lambda'}\right)\sqrt{-g^{(0)}(x')}\ud^4x' \, ,
\end{equation}
where $\mathcal W$ is the (compact) region corresponding to the
``recent universe'' (i.e., $z \leq 1000$ in the present, actual
universe).

Next, we argue that, in order to calculate $\hh^{(S)}{}^{\alpha\beta}$
to $O(\lambda)$ at $x$, it suffices to perform the integral in
\eqref{eqn:hS-wavemap-integral2} only over a small neighborhood
$\mathcal V$ of $x$.  In other words, we argue that---for any fixed
neighborhood, $\mathcal V$, of $x$---as $\lambda \to 0$, the
contribution to the integral in \eqref{eqn:hS-wavemap-integral2} from
the region $x' \in {\mathcal W} \setminus {\mathcal V}$ should vanish
faster than $\lambda$ as $\lambda \to 0$, i.e.,
\begin{equation}\label{eqn:hS-wavemap-integral3}
\int_{{\mathcal W} \setminus {\mathcal V}} G_{\mathrm{ret}\phantom{\alpha\beta}\mu'\nu'}^{\phantom{\mathrm{ret}}\alpha\beta}(x,x')\left(\frac{S^{\mu'\nu'}(\lambda,x')}{\lambda}-\wlim_{\lambda'\to0}\frac{S^{\mu'\nu'}(\lambda',x')}{\lambda'}\right)\sqrt{-g^{(0)}(x')}\ud^4x' \to 0 \, .
\end{equation}
To see this, we note that if
$G_{\mathrm{ret}\phantom{\alpha\beta}\mu'\nu'}^{\phantom{\mathrm{ret}}\alpha\beta}(x,x')$
were smooth (and if the sharp boundaries of the region of integration
were replaced by smooth cut-off functions), then we would be
integrating the quantity $(S^{\mu \nu}/\lambda - \wlim_{\lambda'\to0}
[S^{\mu \nu}/\lambda'])$ with a test function. Since, clearly, the
weak limit as $\lambda \to 0$ of this quantity is $0$, it follows
immediately that \eqref{eqn:hS-wavemap-integral3} would hold. Of
course,
$G_{\mathrm{ret}\phantom{\alpha\beta}\mu'\nu'}^{\phantom{\mathrm{ret}}\alpha\beta}(x,x')$
is not smooth, so \eqref{eqn:hS-wavemap-integral3} cannot be expected
to hold without further restrictions on $S^{ab}$.  However, in a
normal neighborhood of $x$, the retarded Green's function
$G_{\mathrm{ret}\phantom{\alpha\beta}\mu'\nu'}^{\phantom{\mathrm{ret}}\alpha\beta}(x,x')$
for the linear operator $\mathscr{L}^{ab}$ takes the form
\begin{equation}\label{eqn:Greensfunction}
  G_{\mathrm{ret}\phantom{\alpha\beta}\mu'\nu'}^{\phantom{\mathrm{ret}}\alpha\beta}(x,x')=U^{\alpha\beta}_{\phantom{\alpha\beta}\mu'\nu'}(x,x')\delta_+(\sigma)+V^{\alpha\beta}_{\phantom{\alpha\beta}\mu'\nu'}(x,x')\theta_+(-\sigma)\,,
\end{equation}
where $U^{\alpha\beta}_{\phantom{\alpha\beta}\mu'\nu'}(x,x')$ and
$V^{\alpha\beta}_{\phantom{\alpha\beta}\mu'\nu'}(x,x')$ are smooth
bitensors and $\sigma$ is the squared geodesic distance between $x$
and $x'$ (see \cite{lrr-2004-6} for details). Thus, apart from the
singularity at $x'=x$ (which is excluded from the region of
integration in \eqref{eqn:hS-wavemap-integral3}), the singularities of
$G_{\mathrm{ret}\phantom{\alpha\beta}\mu'\nu'}^{\phantom{\mathrm{ret}}\alpha\beta}(x,x')$
are of the form of a restriction to the past lightcone (i.e.,
$\delta_+(\sigma)$) and a cutoff at the past lightcone (i.e.,
$\theta_+(-\sigma)$). If $S^{ab}$ is not rapidly varying with time (as
should be the case under our Newtonian assumptions below), these
singularities should be quite benign, so it seems not unreasonable
that \eqref{eqn:hS-wavemap-integral3} will hold under suitable
assumptions.

We have argued that \eqref{eqn:hS-wavemap-integral3} should hold for
an arbitrary neighborhood $\mathcal V$ of $x$. However, the vanishing
of the contribution to $\hh^{(S)}{}^{\alpha\beta}(\lambda,x)$ from
outside of $\mathcal V$ to $O(\lambda)$ holds only in the limit as
$\lambda \to 0$, and the smaller we take $\mathcal V$, the smaller we
must take $\lambda$ in order to get a good approximation to
$\hh^{(S)}{}^{\alpha\beta}(\lambda,x)$ by integrating only over
$\mathcal V$. At any finite $\lambda$, we cannot take $\mathcal V$ to
be arbitrarily small and still get a good approximation to
$\hh^{(S)}{}^{\alpha\beta}(\lambda,x)$.  How large must we take
$\mathcal V$ at finite $\lambda$?

To propose an answer to this question, we restrict consideration to
the case where the matter content satisfies suitable Newtonian
behavior (at least in the ``recent universe'').  We shall assume that
as $\lambda\to0$ we have $T_{00} (\lambda) = O(1/\lambda)$, $T_{0i}
(\lambda) = O(1/\lambda^{1/2})$ and $T_{ij} (\lambda) = O(1)$, so
that, for small $\lambda$, the energy density is much greater than the
momentum density, and the momentum density is much greater than the
stress.  In addition we shall assume that spatial differentiation of
components of the stress-tensor results in blow-up as $\lambda \to 0$
that is a factor of $\lambda^{-1}$ faster than the undifferentiated
components (so, e.g., $\partial_i T_{00} (\lambda) = O(1/\lambda^2)$),
but that time differentiation results in a blow-up of only a factor of
$\lambda^{-1/2}$ faster (so, e.g., $\partial_0 T_{00} (\lambda) =
O(1/\lambda^{3/2})$).

We now introduce the notion of the {\em scale of homogeneity},
$l(\lambda,\tau_0)$, at cosmic time $\tau_0$ as follows.  First, we
define a fiducial window function which is to be used for
averaging. Fix a time interval $\Delta \tau \ll \tau_0$ and let
$V_{R,x_0}$ to be the ``cylinder'', centered at $x_0$, of ``height''
$\Delta \tau$ and proper spatial radius $R$.  Let $\chi_{R,x_0}(x)$ be
a smooth non-negative function which is equal to one on $V_{R,x_0}$,
and falls rapidly to zero outside of this region. Let $\mathcal H$
denote the ``Hubble volume'' relative to $x$ at time $\tau_0$, i.e.,
the ball of proper radius equal to the Hubble radius, $R_H$, centered
at point $x$. We define
\begin{equation}
  l(\lambda,\tau_0)=\inf\left\{R:\left|\int\left(T_{00}(\lambda)-T_{00}^{(0)}\right)\chi_{R,\tau_0,\boldsymbol{x}_0}\right|<\left|\int T_{00}^{(0)}\chi_{R,\tau_0,\boldsymbol{x}_0}\right|,\,\forall \boldsymbol{x}_0 \in {\mathcal H} \right\}\,.
\label{homscale}
\end{equation}
In other words $l$ is the smallest radius such that averaging over a ball
of radius $l$ centered at any point $x_0$ lying within the Hubble volume 
relative to $x$ always
yields $|\delta \rho|/\rho^{(0)} < 1$. 
Note that our definition of $l(\lambda,\tau_0)$ depends on our
choice of ``window function''
$\chi_{R,x_0}(x)$.

We now claim that $l(\lambda,\tau_0)\to0$ as $\lambda \to 0$ (and
thus, in particular, $l$ is always finite at sufficiently small
$\lambda$).  To prove this, we note that if this result did not hold,
we could find an $l_0>0$ and a sequence $\{\lambda_n:n\in\mathbb{N}\}$
converging to zero, such that $l(\lambda_n,\tau_0) > l_0$ for all
$n\in\mathbb{N}$. Consequently, there would exist a sequence of points
$\{\boldsymbol{x}_n\}\subset {\mathcal H}$ such that
\begin{equation}
\left|\int\left(T_{00}(\lambda_n)-T_{00}^{(0)}\right)\chi_{l_0,\tau_0,\boldsymbol{x}_n}\right|\ge\left|\int T_{00}^{(0)}\chi_{l_0,\tau_0,\boldsymbol{x}_n}\right| \, .
\end{equation}
Since $\mathcal H$ is compact, there exists a subsequence---which we
also denote as $\{\boldsymbol{x}_n\}$---converging to some
$\boldsymbol{z}\in {\mathcal H}$.  By the triangle inequality, we have
\begin{equation}
\left|\int\left(T_{00}(\lambda_n)-T_{00}^{(0)}\right)\chi_{l_0,\tau_0,\boldsymbol{z}}\right|
+
\left|\int\left(T_{00}(\lambda_n)-T_{00}^{(0)}\right)\left[\chi_{l_0,\tau_0,\boldsymbol{x}_n}
- \chi_{l_0,\tau_0,\boldsymbol{z}}\right] \right| \geq 
\left|\int T_{00}^{(0)} \chi_{l_0,\tau_0,\boldsymbol{x}_n}\right|\,.
\end{equation}
Taking the limit as $n \to \infty$, we see that the first term on the
left side vanishes because $T_{00}(\lambda)$ converges weakly to
$T_{00}^{(0)}$ and the second term vanishes by the lemma of section
II. On the other hand, the right side is bounded away from $0$, thus
yielding a contradiction, thereby proving the desired result that
$l(\lambda,\tau_0)\to0$ as $\lambda \to 0$.

Returning to the question posed four paragraphs above, since
$T_{00}(\lambda)$ provides the dominant contribution to the source
term $S^{ab}$, it seems clear that if, at finite $\lambda$, ${\mathcal
  V}(\lambda)$ is chosen to be so small that it does not include all
source contributions lying within a homogeneity scale
$l(\lambda,\tau_0)$ about point $x$, then we cannot expect the source
contributions from outside of ${\mathcal V}(\lambda)$ to consistently
average to zero to a good approximation. On the other hand, if
${\mathcal V}(\lambda)$ is of order of the homogeneity scale or
larger, then it seems plausible that (with suitable additional
assumptions) a good approximation to
$\hh^{(S)}{}^{\alpha\beta}(\lambda,x)$ will be obtained.  Thus, we
have argued that $\hh^{(S)}{}^{\alpha\beta}(\lambda,x)$ should be well
approximated by
\begin{equation}\label{eqn:hS-wavemap-integral4}
  \hh^{(S)}{}^{\alpha\beta}(\lambda,x) = 4\int_{{\mathcal V}(\lambda)}G_{\mathrm{ret}\phantom{\alpha\beta}\mu'\nu'}^{\phantom{\mathrm{ret}}\alpha\beta}(x,x')\left(S^{\mu'\nu'}(\lambda,x')-\lambda\wlim_{\lambda'\to0}\frac{S^{\mu'\nu'}(\lambda',x')}{\lambda'}\right)\sqrt{-g^{(0)}(x')}\ud^4x' \, ,
\end{equation}
where ${\mathcal V}(\lambda)$ may be taken to be a ``cylinder''
centered at $x$ of proper spatial radius $L$ and proper time
``height'' $2L$, where $L\gtrsim l(\lambda,\tau_0)$.  Note that in
order to obtain the conclusion that we need only integrate over the
local region ${\mathcal V}(\lambda)$, it is essential that we have
removed the ``long wavelength part'' of $\hh^{\alpha\beta}$.

We now assume that $l(\lambda,\tau_0)\ll R_C$ so that we can choose
$L$ such that $L\ll R_C$, where $R_C$ denotes the length scale of
curvature of the background metric $g_{ab}^{(0)}$ (i.e., the Hubble
radius, assuming spatial curvature is negligible).  In the actual
universe at the present time, we have $R_C\sim3000$~Mpc and
$l(\lambda,\tau_0)\lesssim100$~Mpc, so we should easily satisfy all
required criteria with $L\sim100$~Mpc.  In that case, it should
suffice to take only the leading order terms in the Hadamard expansion
for $U^{\alpha\beta}_{\phantom{\alpha\beta}\mu'\nu'}$ and
$V^{\alpha\beta}_{\phantom{\alpha\beta}\mu'\nu'}$ as well as the
leading order approximation to $\sigma$ in the Green's function
expression \eqref{eqn:Greensfunction}. This yields
\begin{eqnarray}
  G_{\mathrm{ret}\phantom{\alpha\beta}\mu'\nu'}^{\phantom{\mathrm{ret}}\alpha\beta}(\tau,\boldsymbol{0},\tau',\boldsymbol{x}')&=&\delta^{(\alpha}_{\mu'}\delta^{\beta)}_{\nu'}\frac{\delta\left(\tau'-\tau+a(\tau)r'\right)}{a(\tau)r'}\nonumber\\
  &&+V^{\alpha\beta}_{\phantom{\alpha\beta}\mu'\nu'}(\tau,\boldsymbol{0},\tau,\boldsymbol{0})\theta\left(-\tau'+\tau-a(\tau)r'\right)\,,
\end{eqnarray}
where we have now put the field evaluation point $x$ at the spatial
origin of our coordinate system \eqref{FLRW}. Since
$V^{\alpha\beta}_{\phantom{\alpha\beta}\mu'\nu'}
(\tau,\boldsymbol{0},\tau,\boldsymbol{0})$ is proportional to the
curvature of $g_{ab}^{(0)}$, it is clear that the contribution of the
second term will be down by a factor of $(L/R_C)^2$ from the first
term, so we neglect this contribution. We also neglect ``retardation
effects'', i.e., the difference between evaluating the source at time
$\tau - a(\tau) r'$ and time $\tau$.  Our formula
\eqref{eqn:hS-wavemap-integral4} for $\hh^{(S)}{}^{\alpha\beta}$ then
reduces to
\begin{equation}\label{eqn:hconstanttau}
  \hh^{(S)}{}^{\alpha\beta}(\lambda,\tau,\boldsymbol{0})\approx 4\int \ud\Omega'
\int_0^{\frac{L}{a(\tau)}} \frac{1}{r'}
\left(S^{\alpha\beta}(\lambda,\tau,\boldsymbol{x}')-\lambda\wlim_{\lambda'\to0}\frac{S^{\alpha\beta}(\lambda',\tau,\boldsymbol{x}')}{\lambda'}\right) a^2(\tau)r'^2 \ud r' \,.
\end{equation}

Next, we assume that $l(\lambda,\tau_0)$ not only goes to zero as
$\lambda \to 0$ (as we have proven above) but is $O(\lambda)$ as
$\lambda \to 0$, i.e., $l(\lambda,\tau_0)/\lambda$ remains
bounded\footnote{This precludes behavior wherein, e.g.,
  $T_{00}(\lambda)$ is $O(1)$ as $\lambda \to 0$ but its scale of
  spatial variation goes as $\lambda^{1/2}$ rather than
  $\lambda$. Such behavior would not be excluded by our previous
  assumptions.}  as $\lambda \to 0$.  Thus, if we choose $L$ such that
$L/\lambda$ remains bounded as $\lambda\to0$, by inspection of
\eqref{eqn:hconstanttau}, it is then clear that terms in $S^{ab}$ that
are $o(1/\lambda)$ as $\lambda \to 0$ will make only $o(\lambda)$
contributions to $\hh^{(S)}{}^{\alpha\beta}$. However, we assumed
above that as $\lambda \to 0$, we have $T_{00} (\lambda) =
O(1/\lambda)$, $T_{0i} (\lambda) = O(1/\lambda^{1/2})$ and $T_{ij}
(\lambda) = O(1)$. We now make a final, additional assumption that
terms of the form $\nabla_c\nabla_d\hh^{ab}$ are $O(1/\lambda)$. In
that case, it follows immediately that $\mathfrak{t}^{ab}$ is $O(1)$
as $\lambda \to 0$ (see \eqref{matht}). It then follows that to obtain
$\hh^{(S)}{}^{\alpha\beta}$ to $O(\lambda)$ accuracy, we need only
take into account the contribution to $S^{ab}$ from $T_{00}$. Thus, to
$O(\lambda)$ accuracy, we obtain
\begin{equation}\label{eqn:hconstanttau2}
  \hh^{(S)}{}^{00}(\lambda,\tau,\boldsymbol{0})= 4\int \ud\Omega'
\int_0^{\frac{L}{a(\tau)}} \frac{1}{r'}
\left(T^{00}(\lambda,\tau,\boldsymbol{x}')-T^{(0)00}(\tau,\boldsymbol{x}') -\lambda T^{(1)00}(\tau,\boldsymbol{x}')
\right) a^2(\tau)r'^2 \ud r'
\end{equation}
and $\hh^{(S)}{}^{\mu \nu} = 0$ for all other components of
$\hh^{(S)}{}^{ab}$.  We write
\begin{equation}
\phi \equiv -\frac{1}{4}\hh^{(S)}{}^{00} \, .
\label{phidef}
\end{equation}
Then $\phi(x)$ differs from the familiar formula for the gravitational
potential arising in ordinary Newtonian gravity due to the matter
lying within a ball of proper radius $L$ about $x$ only in the
following ways: (a) There is a factor of $a^2(\tau)$ in
\eqref{eqn:hconstanttau2}, which arises from the trivial scaling
difference between the spatial coordinates of \eqref{FLRW} and
ordinary Cartesian coordinates. (b) $T^{(0)00}$ is subtracted in the
integrand of \eqref{eqn:hconstanttau2} because the FLRW time slicing
differs from a locally Minkowskian time slicing (see section IA of
\cite{Holz:1997ic}); equivalently, the effects of $T^{(0)00}$ have
already been taken into account via the dynamics of the FLRW
background. (c) $\lambda T^{(1)00}$ is subtracted in the integrand of
\eqref{eqn:hconstanttau2} because its effects were already taken into
account by $\gamma^{(L)}_{ab}$. As discussed above, this subtraction
of $\lambda T^{(1)00}$ gives the integral much better convergence
properties.  Thus, we conclude that the leading order short wavelength
deviation from the FLRW background $g^{(0)}_{ab}$ is described by
Newtonian gravity, taking into account only the matter distribution
lying within a region about $x$ whose size is of order the homogeneity
lengthscale.

The motion of matter is given by
\begin{equation}
0 = \nabla_a (\lambda) T^{ab} (\lambda) = \nabla_a T^{ab} (\lambda)
+ C^a_{\phantom{a}ac}(\lambda) T^{cb}(\lambda) + C^b_{\phantom{b}ac}(\lambda) T^{ac} (\lambda) \, .
\label{eom}
\end{equation}
We may write
\begin{equation}
C^a_{\phantom{a}bc} (\lambda) = C^{(S)}{}^a_{\phantom{a}bc}(\lambda) + \lambda C^{(1)}{}^a_{\phantom{a}bc}\,,
\label{Ctot}
\end{equation}
where, to leading order in wave map gauge, we have
\begin{eqnarray}
C^{(S)}{}^0_{\phantom{0}0i}&=&\nabla_i\phi\,,\\
C^{(S)}{}^i_{\phantom{i}00}&=&\nabla^i\phi\,,\\
C^{(S)}{}^i_{\phantom{i}jk}&=&g_{jk}^{(0)}\nabla^i\phi-2\delta^i_{\phantom{i}(j}\nabla_{k)}\phi\,,
\end{eqnarray}
(with other components zero), and
\begin{equation}
C^{(1)}{}^a_{\phantom{a}bc} = \frac{1}{2}g^{(0)}{}^{cd} \left\{\nabla_a\gamma^{(L)}_{bd}+\nabla_b\gamma^{(L)}_{ad}-\nabla_d \gamma^{(L)}_{ab}\right\}\, .
\label{C1}
\end{equation}
The dominant terms in \eqref{eom} arise from
$C^{(S)}{}^i_{\phantom{i}00}(\lambda)T^{00}(\lambda)$ and correspond
to the ordinary Newtonian gravitational effects on the motion of
matter.  Although, for small $\lambda$, the contributions from
$C^{(1)}{}^a_{\phantom{a}bc}$ will be much smaller than those arising
from $\phi$, it is important not to discard the terms in
$C^{(1)}{}^a_{\phantom{a}bc}$ since they can produce large scale,
coherent motions.

\subsection{Behavior of $\gamma^{(L)}_{ab}$ with dust source}

In this subsection, we will simplify the rather complicated equations
for $\gamma^{(L)}_{ab}$ derived in section III under the assumption that
$\hh^{(S)}{}^{ab}(\lambda)$ is of ``Newtonian form''. More precisely, 
we assume that the following quantities are uniformly
bounded as $\lambda\to0$:
\begin{equation}
  \begin{matrix}
    \frac{1}{\lambda}\hh^{(S)}{}^{00}\,, & \frac{1}{\lambda^{1/2}}\nabla_0\hh^{(S)}{}^{00}\,, & \nabla_i\hh^{(S)}{}^{00}\,, \\
    & \frac{1}{\lambda}\nabla_0\hh^{(S)}{}^{0j}\,, & \frac{1}{\lambda^{1/2}}\nabla_i\hh^{(S)}{}^{0j}\,, \\
    &  & \frac{1}{\lambda}\nabla_i\hh^{(S)}{}^{jk}\,.
  \end{matrix}
  \label{newt}
\end{equation}
The remaining components ($\hh^{(S)}{}^{0i}$, $\hh^{(S)}{}^{ij}$, and
$\nabla_0\hh^{(S)}{}^{ij}$) are assumed to be $o(\lambda)$.  These
assumptions can be justified by generalizations of the arguments made
in the previous subsection, but here we will simply assume that they
are valid.

We will also assume that the matter stress-energy takes the form of
``dust''
\begin{equation}
  T_{ab}(\lambda)=\rho(\lambda)u_a(\lambda)u_b(\lambda)\,,
\end{equation}
where $u^a(\lambda)$ has norm $-1$ with respect to the metric
$g_{ab}(\lambda)$, and $\rho(\lambda)\ge0$. (Recall that we have
incorporated a cosmological constant into Einstein's equation, so the
possible presence of ``dark energy'' of that form has already been
taken into account.  This assumption of the dust form of the
stress-energy tensor as opposed to a more general form of
non-relativistic matter is made here mainly for the purpose of
obtaining definite equations involving familiar quantities.)  We
further assume that as $\lambda \to 0$, $u^a(\lambda)$ converges
uniformly to $u^{(0)}{}^a$, where in the coordinates of \eqref{FLRW},
we have $u^{(0)}{}^\mu=(1,0,0,0)$.  Note that since $u_a(\lambda)\to
u_a^{(0)}$ uniformly as $\lambda\to0$, it follows by our lemma of
section II and the positivity of $\rho(\lambda)$ that
\begin{equation}
  T^{(0)}_{ab}=\rho^{(0)}u^{(0)}_au^{(0)}_b\,,
\end{equation}
where $\rho^{(0)}=\wlim_{\lambda\to0}\rho(\lambda)$.

Let $v^a(\lambda)$ denote the peculiar velocity of the dust relative
to the ``Hubble flow'' $u^{(0)}{}^a$, i.e., $v^a(\lambda)$ is the
projection of $u^a(\lambda)$ orthogonal to $u^{(0)}{}^a$ in the metric
$g_{ab}^{(0)}$. In accord with usual Newtonian limits, we assume that
$v^a(\lambda)/\lambda^{1/2}$ is uniformly bounded as $\lambda\to0$.
It follows that
\begin{equation}
  u^a(\lambda)=\left(1+\frac{1}{2}h_{cd}(\lambda)u^{(0)}{}^cu^{(0)}{}^d+\frac{1}{2}v_c(\lambda)v^c(\lambda)\right)u^{(0)}{}^a+v^a(\lambda)+o(\lambda)\,,
\end{equation}
or, equivalently,
\begin{equation}
  u_a(\lambda)=g_{ab}(\lambda)u^b(\lambda)=\left(1+\frac{1}{2}h_{cd}(\lambda)u^{(0)}{}^cu^{(0)}{}^d+\frac{1}{2}v_c(\lambda)v^c(\lambda)\right)u^{(0)}_a+v_a(\lambda)+h_{ab}(\lambda)u^{(0)}{}^b+o(\lambda) \, .
\end{equation}

The results we shall now derive in this subsection will be rigorous
consequences of the assumptions that have been stated above.  We shall
first show that, under the above assumptions\footnote{In fact, we can
  show, without any Newtonian assumptions or gauge choice, that if
  $t^{(0)}_{ab}=0$, then $\alpha_{abcdef}=0$ and
  $\omega^{(1,\alpha)}_{abcdefgh}=0$.  However, this is not enough to
  tell us anything about $\nu^{(1)}_{abcde}$ or the other components
  of $\mu_{abcdef}$ or $\omega^{(1)}_{abcdefgh}$.}, the quantities
$\mu_{abcdef}$, $\omega^{(1,\alpha)}_{abcdefgh}$, and
$\nu^{(1)}_{abcde}$ all vanish, so all of the ``backreaction tensors''
that appear in the equations for $\gamma^{(L)}_{ab}$ of section III
vanish, except for $\mu^{(1)}_{abcdef}$ and $\kappa^{(1)}_{abcd}$.

Taking the weak limit of $\hh^{cd}$
times \eqref{eqn:h-wavemap} (or, equivalently, using \eqref{aht2}
directly) we find that our Newtonian assumptions imply
\begin{equation}
 \wlim_{\lambda\to0}\partial_i\phi\partial^i\phi = 0 \, ,
\end{equation}
where $\phi$ was defined by \eqref{phidef}.
Since the spatial metric $g^{(0)}_{ij}$ is positive definite, this
implies that for any test function $f$, we have
$\|\partial_i(f\phi)\|_{L^2}\to0$ as $\lambda\to0$.  Now,
$h_{ab}^{(S)}=\hh^{(S)}_{ab}-\frac{1}{2}g^{(0)}_{ab}\hh^{(S)}{}^c_{\phantom{c}c}+O(\hh\hh)$,
so the only possible contributions to
$\omega^{(1)}_{abcdefgh}$ come from terms proportional to
\begin{equation}
  \wlim_{\lambda\to0}\frac{1}{\lambda}\phi\partial_i\phi\partial_j\phi\,.
\end{equation}
Let $f$ be any test function, and let $g$ be a non-negative test
function which is equal to 1 on the support of $f$.  Then we have
\begin{eqnarray}
  \left|\int\left(\wlim_{\lambda\to0}\frac{1}{\lambda}\phi\partial_i\phi\partial_j\phi\right)f\ud^4x\right|
   &=&\left|\lim_{\lambda\to0}\int\frac{1}{\lambda}\phi\partial_i\phi\partial_j\phi f \ud^4x\right|\nonumber\\
  &=&\left|\lim_{\lambda\to0}\int\frac{1}{\lambda}\phi\partial_i\phi\partial_j\phi fg^2\ud^4x\right|\nonumber\\
  &\leq&C\lim_{\lambda\to0}\int|\partial_i(g\phi)\partial_j(g\phi)|\ud^4x\nonumber\\
  &\le&C\lim_{\lambda\to0}\|\partial_i(g\phi)\|_{L^2}\|\partial_j(g\phi)\|_{L^2}\nonumber\\
  &=&0\,.
\end{eqnarray}
Thus, we have $\omega^{(1)}_{abcdefgh}=0$.  Similar
arguments show that $\mu_{abcdef}=0$.  Finally, $\nu^{(1)}_{abcde}$
can only depend on terms proportional to
\begin{equation}
  \wlim_{\lambda\to0}\frac{1}{\lambda}\phi\partial_i\phi=\wlim_{\lambda\to0}\partial_i\left(\frac{\phi^2}{2\lambda}\right)=0\,,
\end{equation}
so it vanishes as well, as we desired to show.  With this
simplification, \eqref{eqn:gammaL} reduces to
\begin{eqnarray}\label{eqn:gammaLbar-dust}
  &&-\frac{1}{2}\nabla^c\nabla_c\bar{\gamma}^{(L)}_{ab}+R_{a\phantom{cd}b}^{\phantom{a}cd}(g^{(0)})\bar{\gamma}^{(L)}_{cd}+R^d_{\phantom{d}(a}(g^{(0)})\bar{\gamma}^{(L)}_{b)d}+\frac{1}{2}g^{(0)}_{ab}R^{cd}(g^{(0)})\bar{\gamma}^{(L)}_{cd}-\frac{1}{2}R(g^{(0)})\bar{\gamma}^{(L)}_{ab}\nonumber\\
  &&+\Lambda\left(\bar{\gamma}^{(L)}_{ab}-\frac{1}{2}g^{(0)}_{ab}\bar{\gamma}^{(L)}{}^c_{\phantom{c}c}\right)\nonumber\\
  &=&8\pi T^{(1)}_{ab}+\alpha^{(1)}{}_{a\phantom{c}b\phantom{d}cd}^{\phantom{a}c\phantom{b}d}-2\pi\kappa^{(1)}{}_{ab\phantom{c}c}^{\phantom{ab}c}-8\pi\kappa^{(1)}{}_{(a\phantom{c}b)c}^{\phantom{(a}c}+2\pi g^{(0)}_{ab}\left\{\kappa^{(1)}{}^{c\phantom{c}d}_{\phantom{c}c\phantom{d}d}-\kappa^{(1)}{}^{cd}_{\phantom{cd}cd}\right\}\,.
\end{eqnarray}
where
$\bar{\gamma}^{(L)}_{ab}=\gamma^{(L)}_{ab}-\frac{1}{2}g_{ab}^{(0)}\gamma^{(L)}{}^c_{\phantom{c}c}$,
and where we have used the wave map gauge condition,
$\nabla^a\bar{\gamma}^{(L)}_{ab}=0$ (see \eqref{lorenz}).

Next, we evaluate $\kappa^{(1)}_{abcd}$. With our above assumptions of
slowly moving dust matter, we have
\begin{equation}
  \kappa^{(1)}_{abcd}=\wlim_{\lambda\to0}\frac{1}{\lambda}h^{(S)}_{ab}(\lambda)\rho(\lambda)u_a(\lambda)u_b(\lambda)=u^{(0)}_cu^{(0)}_d\wlim_{\lambda\to0}\frac{1}{\lambda}h^{(S)}_{ab}(\lambda)\rho(\lambda)\,.
\end{equation}
The second equality follows from our lemma of section II because
$\left[u_c(\lambda)-u^{(0)}_c\right]
h^{(S)}_{ab}(\lambda)/\lambda\to0$ uniformly as $\lambda\to0$ and
$\rho(\lambda)\ge0$.  Using the Newtonian assumption that all
components, except for the time-time component, of
$\hh^{(S)}_{ab}(\lambda)/\lambda$ converge uniformly to zero as
$\lambda\to0$, the lemma of section II tells us furthermore that the
only non-vanishing components of $\kappa^{(1)}_{abcd}$ are
\begin{eqnarray}
  \kappa^{(1)}_{0000}&=&-2\wlim_{\lambda\to0}\frac{1}{\lambda}\rho(\lambda)\phi(\lambda)\,,\\
  \kappa^{(1)}_{ij00}&=&-2g^{(0)}_{ij}\wlim_{\lambda\to0}\frac{1}{\lambda}\rho(\lambda)\phi(\lambda)\,.
\end{eqnarray}

To evaluate $\mu^{(1)}_{abcdef}$, it is more convenient 
to work with the quantity
\begin{eqnarray}
  \bar{\mu}^{(1)}{}_{ab}^{\phantom{ab}cdef}&\equiv&\wlim_{\lambda\to0}\frac{1}{\lambda}\left[\nabla_{(a}\hh^{(S)}{}^{cd}\nabla_{b)}\hh^{(S)}{}^{ef}-\wlim_{\lambda'\to0}\left[\nabla_{(a}\hh^{(S)}{}^{cd}\nabla_{b)}\hh^{(S)}{}^{ef}\right]\right]\nonumber\\
  &=&\wlim_{\lambda\to0}\frac{1}{\lambda}\left[\nabla_{a}\hh^{(S)}{}^{cd}\nabla_{b}\hh^{(S)}{}^{ef}\right] \, ,
\end{eqnarray}
where the second line follows by the same type of argument as used
above to show the vanishing of $\omega^{(1)}_{abcdefgh}$. Since
$h_{ab}^{(S)}=\hh^{(S)}_{ab}-\frac{1}{2}g^{(0)}_{ab}\hh^{(S)}{}^c_{\phantom{c}c}+O(\hh\hh)$,
it is not difficult to see that $\bar{\mu}^{(1)}_{abcdef}$ can be
expressed straightforwardly in terms of $\mu^{(1)}_{abcdef}$ and
vice-versa.  Our Newtonian assumptions \eqref{newt} directly imply
that the only potentially non-zero components of
$\bar{\mu}^{(1)}_{abcdef}$ are
\begin{equation}
  \begin{matrix}
    \bar{\mu}^{(1)}_{000000}\,, & \bar{\mu}^{(1)}_{0i0000}\,, & \bar{\mu}^{(1)}_{0i000j}\,, & \bar{\mu}^{(1)}_{ij0000}\,, \\
    \bar{\mu}^{(1)}_{ij000k}\,, & \bar{\mu}^{(1)}_{ij00kl}\,, & \bar{\mu}^{(1)}_{ij0k0l}\,,&
  \end{matrix}
\end{equation}
together with the components related to these by symmetries. 
However, it also follows that
$\bar{\mu}^{(1)}_{ij00kl}=0$, since
\begin{eqnarray}
  \left|\int f\bar{\mu}^{(1)}_{ij00kl}\ud^4x\right|&=&\left|\lim_{\lambda\to0}\int f\frac{1}{\lambda}\nabla_i\hh^{(S)}_{00}\nabla_j\hh^{(S)}_{kl}\ud^4x\right|\nonumber\\
  &\le& C \lim_{\lambda\to0}\int\left| f \nabla_i\hh^{(S)}_{00} \right| \ud^4x \nonumber\\
  &\leq&C'\lim_{\lambda\to0}\left\| \nabla_i (f \phi) \right\|_{L^2}\nonumber\\
  &=&0\,.
\end{eqnarray}
where in the second line we used that fact that
$\nabla_j\hh^{(S)}_{kl}/\lambda$ is uniformly bounded as $\lambda \to
0$.

To further simplify $\bar{\mu}^{(1)}_{abcdef}$, we now appeal to
\eqref{eqn:alpha1}. With the simplifications arising from the
vanishing of $\mu_{abcdef}$, $\omega^{(1,\alpha)}_{abcdefgh}$, and
$\nu^{(1)}_{abcde}$ together with the wave map gauge condition
$\nabla_a \hh^{(S)}{}^{ab} = 0$ (see \eqref{hhwmg}), we obtain
\begin{eqnarray}\label{eqn:mubar1-dust}
  \frac{1}{4}\bar{\mu}^{(1)}{}^e_{\phantom{e}eabcd}=4\pi\kappa^{(1)}_{cdab}-2\pi g_{cd}^{(0)}\kappa^{(1)}{}^e_{\phantom{e}eab}\,.
\end{eqnarray}
From the form we have derived above for $\kappa^{(1)}_{abcd}$, it
follows immediately that
\begin{equation}
  \bar{\mu}^{(1)}{}^i_{\phantom{i}i0j0k}=\bar{\mu}^{(1)}{}^a_{\phantom{a}a0j0k}=0\,.
\end{equation}
Again by the positive definiteness of the spatial metric, it follows
that $\|\nabla_i(f\hh^{(S)}_{0j})/\lambda^{1/2}\|_{L^2}\to 0$ as
$\lambda\to0$, which implies $\bar{\mu}^{(1)}_{ij0k0l}=0$ as well.  By
using similar Schwartz inequality-type arguments, it also follows that
$\bar{\mu}^{(1)}_{ij000k}=0$.  The wave map gauge condition $\nabla_a
\hh^{(S)}{}^{ab} = 0$ then yields
\begin{eqnarray}
  \bar{\mu}^{(1)}_{0i000j}=\bar{\mu}^{(1)}{}^k_{\phantom{k}i0k0j}&=&0\,,\nonumber\\
  \bar{\mu}^{(1)}_{0i0000}=\bar{\mu}^{(1)}{}^j_{\phantom{j}i0j00}&=&0\,,\nonumber\\
  \bar{\mu}^{(1)}_{000000}=\bar{\mu}^{(1)}{}^i_{\phantom{i}00i00}&=&0\,.
\end{eqnarray}
The only nonvanishing  components which remain are
\begin{equation}
\bar{\mu}^{(1)}_{ij0000}=16\Xi_{ij} \, , 
\end{equation}
where 
\begin{equation}
  \Xi_{ij}\equiv\wlim_{\lambda\to0}\frac{1}{\lambda}\nabla_i\phi\nabla_j\phi\,.
\end{equation}
Furthermore, by \eqref{eqn:mubar1-dust} and our previous expression
for $\kappa^{(1)}_{abcd}$, we obtain
\begin{equation}
  \Xi^i_{\phantom{i}i}=-4\pi\wlim_{\lambda\to0}\frac{1}{\lambda}\rho(\lambda)\phi(\lambda)\,.
\end{equation}
Finally, expressing $\mu^{(1)}_{abcdef}$ in terms of
$\bar{\mu}^{(1)}_{abcdef}$, we find that the nonvanishing components
of $\mu^{(1)}_{abcdef}$ are
\begin{eqnarray}
  \mu^{(1)}_{ij0000}&=&4\Xi_{ij}\,,\\
  \mu^{(1)}_{ijkl00}&=&4g^{(0)}_{kl}\Xi_{ij}\,,\\
  \mu^{(1)}_{ijklmn}&=&4g^{(0)}_{kl}g^{(0)}_{mn}\Xi_{ij}\,.
\end{eqnarray}

Next, we consider $T^{(1)}_{ab}$. We have
\begin{eqnarray}\label{eqn:T1-dust}
  T^{(1)}_{ab}&=&\wlim_{\lambda\to0}\frac{1}{\lambda}\left(T_{ab}(\lambda)-T^{(0)}_{ab}\right)\nonumber\\
  &=&\wlim_{\lambda\to0}\frac{1}{\lambda}\left(\rho(\lambda)u_a(\lambda)u_b(\lambda)-\rho^{(0)}u^{(0)}_au^{(0)}_b\right)\nonumber\\
  &=&\wlim_{\lambda\to0}\frac{1}{\lambda}\left(\left[\rho(\lambda)-\rho^{(0)}\right]u_a^{(0)}u_b^{(0)}+\rho(\lambda)v_a(\lambda)v_b(\lambda)+2\rho(\lambda)u_{(a}^{(0)}v_{b)}(\lambda)\right.\nonumber\\
  &&\left.+\rho(\lambda)u_a^{(0)}u_b^{(0)}\left[h_{cd}(\lambda)u^{(0)}{}^cu^{(0)}{}^d+v_c(\lambda)v^c(\lambda)\right]+2\rho(\lambda)u^{(0)}{}^cu^{(0)}_{(a}h_{b)c}(\lambda)\right)\nonumber\\
  &=&\rho^{(1)}u^{(0)}_au^{(0)}_b+p^{(1)}_{ab}+2u^{(0)}_{(a}P^{(1)}_{b)}+u^{(0)}_au^{(0)}_b\kappa^{(1)}{}^{cd}_{\phantom{cd}cd}+\rho^{(0)}u^{(0)}_au^{(0)}_bu^{(0)}{}^cu^{(0)}{}^d\gamma^{(L)}_{cd}\nonumber\\
  &&+u^{(0)}_au^{(0)}_bp^{(1)}{}^c_{\phantom{c}c}+2\kappa^{(1)}{}^c_{\phantom{c}(ab)c}+2\rho^{(0)}u^{(0)}{}^cu^{(0)}_{(a}\gamma^{(L)}_{b)c}\nonumber\\
  &=&\rho^{(1)}u^{(0)}_au^{(0)}_b+2u^{(0)}_{(a}P^{(1)}_{b)}+p^{(1)}_{ab}+u^{(0)}_au^{(0)}_bp^{(1)}{}^c_{\phantom{c}c}+u^{(0)}_au^{(0)}_b\kappa^{(1)}{}^{cd}_{\phantom{cd}cd}+2\kappa^{(1)}{}^c_{\phantom{c}(ab)c}\nonumber\\
  &&+u^{(0)}_au^{(0)}_bT^{(0)}{}^{cd}\bar{\gamma}^{(L)}_{cd}+2T^{(0)}{}^c_{\phantom{c}(a}\bar{\gamma}^{(L)}_{b)c}-\frac{1}{2}T^{(0)}_{ab}\bar{\gamma}^{(L)}{}^c_{\phantom{c}c}\,.
\end{eqnarray}
Here, we have introduced the quantities
\begin{equation}
\rho^{(1)}\equiv\wlim_{\lambda\to0}\left[\rho(\lambda)-\rho(0)\right]/\lambda \, ,
\end{equation}
\begin{equation}
P^{(1)}_a\equiv\wlim_{\lambda\to0}\rho(\lambda)v_a(\lambda)/\lambda \, ,
\end{equation}
and
\begin{equation}
p^{(1)}_{ab}\equiv\wlim_{\lambda\to0}\rho(\lambda)v_a(\lambda)v_b(\lambda)/\lambda \, .
\end{equation}

Now, for a pressureless fluid, $\rho u^a$ is a conserved current,
whose integrated flux through a spacelike hypersurface can be
interpreted as (proportional to) the ``number of particles'' (often
referred to as the ``number of baryons'') in the fluid. The
perturbation, $\rho^{(1)}_M$, of the density (relative to the
background metric $g^{(0)}_{ab}$) of particles on a $\tau = {\rm
  const.}$ hypersurface is given by
\begin{equation}
  \rho^{(1)}_M\boldsymbol{\epsilon}^{(0)}_{abc}\equiv\wlim_{\lambda\to0}\frac{1}{\lambda}\left[T_{de}(\lambda)u^d(\lambda)n^e(\lambda)\boldsymbol{\epsilon}_{abc}(\lambda)
-T^{(0)}_{de}u^{(0)}{}^dn^{(0)}{}^e \boldsymbol{\epsilon}^{(0)}_{abc} \right] \,,
\end{equation}
where $n^a(\lambda)$ and $\boldsymbol{\epsilon}_{abc}(\lambda)$ are
the unit normal and 3-volume of the hypersurface in the metric
$g_{ab}(\lambda)$ and $n^{(0)}{}^a$ and
$\boldsymbol{\epsilon}^{(0)}_{abc}$ are the corresponding quantities
for $g^{(0)}_{ab}$.  We obtain
\begin{equation}
  \rho^{(1)}_M=\rho^{(1)}+\frac{1}{2}p^{(1)}{}^i_{\phantom{i}i}+\frac{3}{4\pi}\Xi^i_{\phantom{i}i}+\frac{1}{2}\rho^{(0)}\left(\bar{\gamma}^{(L)}_{00}-\frac{1}{2}\bar{\gamma}^{(L)}{}^c_{\phantom{c}c}\right)\,.
\label{rhom}
\end{equation}

We now return to our perturbation equation \eqref{eqn:gammaLbar-dust}
for $\bar{\gamma}^{(L)}_{ab}$.  We use the above explicit expressions
for $\kappa^{(1)}_{abcd}$ and $\mu^{(1)}_{abcdef}$, we use
\eqref{eqn:T1-dust} to substitute for $T^{(1)}_{ab}$, and we use
\eqref{rhom} to eliminate $\rho^{(1)}$ in favor of $\rho^{(1)}_M$.
Bringing all terms explicitly involving $\bar{\gamma}^{(L)}_{ab}$ to
the left side and using the background Einstein equation, we obtain
\begin{eqnarray}
  &&-\frac{1}{2}\nabla^c\nabla_c\bar{\gamma}^{(L)}_{ab}+R_{a\phantom{cd}b}^{\phantom{a}cd}(g^{(0)})\bar{\gamma}^{(L)}_{cd}+\frac{1}{2}g^{(0)}_{ab}R^{cd}(g^{(0)})\bar{\gamma}^{(L)}_{cd}-\frac{1}{2}\Lambda g^{(0)}_{ab}\bar{\gamma}^{(L)}{}^c_{\phantom{c}c}\nonumber\\
  &&-4\pi T^{(0)}_{ab}\bar{\gamma}^{(L)}_{cd}u^{(0)}{}^cu^{(0)}{}^d-8\pi T^{(0)}{}^c_{\phantom{c}(a}\bar{\gamma}^{(L)}_{b)c}+2\pi T^{(0)}_{ab}\bar{\gamma}^{(L)}{}^c_{\phantom{c}c}\nonumber\\
  &\equiv&8\pi\Theta_{ab}\,,
\label{gammaLfin}
\end{eqnarray}
where the components of $\Theta_{ab}$ are explicitly given by
\begin{eqnarray}
  \Theta_{00}&=&\rho^{(1)}_M+\frac{1}{2}p^{(1)}{}^i_{\phantom{i}i}-\frac{1}{8\pi}\Xi^i_{\phantom{i}i} \,,\\
  \Theta_{0i}&=&P^{(1)}_i\,,\\
  \Theta_{ij}&=&p^{(1)}_{ij}+\frac{1}{4\pi}\Xi_{ij}-\frac{1}{8\pi}g^{(0)}_{ij}\Xi^k_{\phantom{k}k}\,.
\end{eqnarray}
This expression for the ``effective perturbed stress-energy'' of
matter agrees with the expression obtained by \cite{Baumann:2010tm}
(see also \cite{Futamase:1996fk}).

All of the terms on the left side of \eqref{gammaLfin} together with
the terms $\rho^{(1)}_M$ and $P^{(1)}_i$ appearing in $\Theta_{ab}$
would be present in ordinary linearized perturbation theory. We
previously showed in section II that, in the absence of gravitational
radiation, small scale inhomogenities cannot affect the dynamics of
the background metric. We now see that, under the assumptions stated
at the beginning of this subsection, the only effect that small scale
inhomogeneities have on long wavelength perturbations is to add the
terms involving $p^{(1)}_{ij}$ and $\Xi_{ij}$ to the effective
perturbed stress-energy $\Theta_{ab}$. These additional terms provide
precisely the contributions to the energy density and stresses that
one would expect from kinetic motions and Newtonian gravitational
potential energy and stresses.  In particular, for the energy density,
these terms have the effect of shifting the proper mass density
appearing in the FLRW background to an ``ADM mass density'' (to first
order in $\lambda$).  Although the presence of these terms is very
important as a matter of principle, they should be extremely small
compared with $\rho^{(1)}_M$ (and even with $P^{(1)}_i$).

Finally, we note that although we have written \eqref{gammaLfin} in
wave map gauge \eqref{lorenz}, we may use the gauge freedom with
respect to smooth, one-parameter groups of diffeomorphisms discussed
in the paragraph below \eqref{eqn:gammaL} to impose any other desired
gauge condition on $\gamma^{(L)}_{ab}$. Since the left side of
\eqref{gammaLfin} together with the terms $\rho^{(1)}_M$ and
$P^{(1)}_i$ in $\Theta_{ab}$ have the same gauge transformation
properties as in ordinary linearized perturbation theory, these terms
will take the same form as in ordinary linearized perturbation theory
when transformed to the new gauge. On the other hand, the terms in
$\Theta_{ab}$ involving $\Xi_{ij}$ and $p^{(1)}_{ij}$ are gauge
invariant under the transformations induced by a smooth, one-parameter
group of diffeomorphisms. This enables one to write down the form of
\eqref{gammaLfin} in other gauges.

\section{Applicability of Our Formalism to the Real Universe}

In this paper, we have developed a formalism that enables one to take
full account of small scale non-linearities in Einstein's
equation. This formalism is, in essence, an adaptation of Burnett's
formulation of the ``shortwave approximation'' for analyzing the
self-gravitating effects of short-wavelength gravitational
radiation. The key idea of our formalism is to consider an idealized,
one-parameter family of metrics, $g_{ab}(\lambda)$, with the property
that, as $\lambda \to 0$, the deviation of $g_{ab}(\lambda)$ from a
``background metric'' $g^{(0)}_{ab}$ goes to zero in proportion to
$\lambda$, but the scales over which the metric and stress-energy vary
also go to zero in proportion to $\lambda$. In principle, this allows
the small-scale nonlinear terms in Einstein's equation to have a
significant effect on the large scale dynamics of
$g^{(0)}_{ab}$. However, we proved in section II that the only such
effects that can actually arise correspond to the presence of an
effective stress-energy associated with gravitational radiation.  No
new effects can arise from the presence of matter (provided that the
matter satisfies the weak energy condition), and no effects can arise
that in any way mimic ``dark energy''.

In sections III and IV, we applied our formalism to analyze the
leading order deviations of $g_{ab}(\lambda)$ from
$g^{(0)}_{ab}$. Within our formalism, these deviations can naturally
be split into ``long wavelength'' and ``short wavelength'' parts. The
long wavelength part satisfies a modified version of the usual
linearized perturbation equation. We argued that, in the case of
cosmological perturbations with nearly Newtonian sources, the short
wavelength part of the leading order deviation of $g_{ab}(\lambda)$
from $g^{(0)}_{ab}$ is described by Newtonian gravity, taking into
account only the local matter distribution.

However, our actual universe is not an idealized limit of spacetimes
with inhomogeneities on arbitrarily small scales, but is a particular
spacetime with finite amplitude, finite wavelength deviations from a
FLRW model. Quantities in our formalism like the ``long wavelength
perturbation,'' $\gamma_{ab}^{(L)}$, are defined by taking weak limits
as $\lambda \to 0$.  How do quantities like $\gamma_{ab}^{(L)}$ that
arise in our formalism correspond to quantities observed in the actual
universe? Furthermore, to what extent are we justified in applying our
results to the actual universe, i.e., how do we know whether the
actual universe is sufficiently ``close'' to the idealized limit
considered in our formalism that results derived using our formalism
should hold to a good approximation?

The above questions can also be asked in the context of ordinary
linearized perturbation theory.  In this context, one also deals with
a limit of a one-parameter family of metrics $g_{ab}(\lambda,x)$ (now
assumed to be jointly smooth in $(\lambda,x)$) and the formalism
obtains results for idealized quantities like $\partial
g_{ab}(\lambda,x)/\partial \lambda|_{\lambda = 0}$.  Nevertheless, in
the case of ordinary perturbation theory, the above questions can be
answered in a relatively straightforward manner. In an actual
(``finite $\lambda$'') spacetime with metric $g_{ab}$, one may
introduce a background metric $g^{(0)}_{ab}$ and identify the
difference, $h_{ab} = g_{ab} - g^{(0)}_{ab}$, with $\lambda \partial
g_{ab}(\lambda,x)/\partial \lambda|_{\lambda = 0}$. A fundamental
criterion for the validity of ordinary perturbation theory is that
$|h_{\mu \nu}| \ll 1$ in some orthonormal basis of
$g^{(0)}_{ab}$. However, we also need to satisfy conditions that state
that first and second spacetime derivatives of $h_{\mu \nu}$ are
sufficiently ``small,'' since these quantities appear in the nonlinear
terms that are being neglected in Einstein's equation. In particular,
the criteria for the applicability of ordinary linear perturbation
theory cannot be satisfied in a cosmological spacetime with $|\delta
\rho| > \rho_0$, as occurs in the real universe. Indeed, it is for
this reason that we have developed the formalism of this paper. It
should be noted that even in the case of spacetimes that do satisfy
the basic criteria for the validity of ordinary linear perturbation
theory, it would be very difficult to obtain precise error estimates
for this approximation.

If one wishes to apply the formalism of this paper to a particular
spacetime, such as our universe, it also is necessary to introduce a
background metric $g^{(0)}_{ab}$. As in ordinary perturbation theory,
it is essential that the difference, $h_{ab} = g_{ab} - g^{(0)}_{ab}$,
satisfy $|h_{\mu \nu}| \ll 1$ in some orthonormal basis of
$g^{(0)}_{ab}$. The main advantage of our formalism over ordinary
perturbation theory is that it imposes much weaker restrictions on
spacetime derivatives of $h_{\mu \nu}$. In particular, for
cosmological spacetimes, it allows $|\delta \rho| \gg \rho_0$ on small
scales.  However, it is clear that some further restrictions in
addition to $|h_{\mu \nu}| \ll 1$ must hold for the results derived
using our formalism to be a good approximation. Although it would be
extremely difficult to formulate mathematically precise criteria for
the validity of applying our formalism to a given spacetime---much
more difficult than for ordinary perturbation theory---we now shall
propose a rough criterion for its validity.

In essence, taking the weak limit as $\lambda \to 0$ of quantities
that arise in our formalism corresponds to taking spacetime averages
of these quantities over arbitrarily small regions of spacetime before
letting $\lambda \to 0$. For a given, fixed spacetime, such as our
universe, the weak limits in our formalism should thus be identified
with spacetime averages over small regions.  Thus, for example, the
quantity $\lambda \gamma_{ab}^{(L)}(x)$ [where $\gamma_{ab}^{(L)}
\equiv \wlim_{\lambda \to 0} h_{ab} (\lambda)/\lambda$] should be
identified with the spacetime average of $h_{ab}$ over a suitable
region, ${\mathcal R}_x$, centered on $x$.  In order that the
spacetime under consideration be ``sufficiently close'' to the
idealized limit $\lambda \to 0$ of our formalism that our results
should apply, it is necessary that ${\mathcal R}_x$ be small compared
with the curvature scale, $R_C$, of $g^{(0)}_{ab}$; indeed, the notion
of ``averaging'' would be highly ambiguous if this were not the case.
However, it also is necessary that ${\mathcal R}_x$ be large enough
that our perturbative approximations should apply.  In particular,
$\lambda \gamma_{ab}^{(L)}$ must satisfy the same ``smallness''
conditions as required in ordinary linearized perturbation theory,
i.e., its spacetime derivatives must be appropriately small.
Similarly, the perturbative quantity $\lambda T^{(1)}_{00}$ (see
\eqref{T1}) should correspond to the spacetime average of
$T_{00}-T_{00}^{(0)}$. In order that our approximations should apply,
we must choose ${\mathcal R}_x$ large enough that $|\lambda
T^{(1)}_{00}| \ll T^{(0)}_{00}$.

The condition that the spacetime average of $T_{00}-T_{00}^{(0)}$ be
less than $T^{(0)}_{00}$ was considered in section IVB above and used
to define the homogeneity scale, $l$ (see \eqref{homscale}). Similar
``homogeneity scales'' can be defined for all other ``long
wavelength'' quantities arising in the analysis of section III. Since
our averaging region must be large compared with these homogeneity
scales but small compared with the scale of the curvature of the
background metric, it is clear that a necessary condition for the
applicability of our formalism is that these homogeneity scales be
small compared with the curvature of the background metric. {\em We
  believe that this condition, together with $|h_{\mu \nu}| \ll 1$,
  should also be sufficient for the applicability of our formalism},
since it should guarantee the appropriate ``smallness'' of all long
wavelength perturbative quantities.

In the case of our universe, it seems reasonable to assume that the
only relevant homogeneity scale is the density homogeneity scale, $l$,
defined by \eqref{homscale}. In the present universe, we have $l
\lesssim 100$~Mpc whereas the scale of the curvature of the background
metric (i.e., the present Hubble radius) is of order $R_C \sim
3000$~Mpc.  Thus $l\ll R_C$.  In addition, common sense estimates
indicate that, except in the immediate vicinity of black holes or
neutron stars, we have $|h_{\mu \nu}| \leq 10^{-5}$.  Thus, if we
define ``long wavelength perturbations'' by averaging on scales of
order $L\sim100$~Mpc, the criteria we have proposed for the validity
of our approximations should be satisfied.  {\em Consequently, we
  believe that the results of this paper should be applicable to our
  universe to an excellent approximation.}

If this is the case, then the following conclusions can be drawn: (1)
As shown in section II, the only effect that small scale
inhomogeneities can have on the leading (i.e., zeroth) order large
scale dynamics of our universe is that of a $P = \frac{1}{3}\rho$ fluid,
corresponding to the presence of gravitational radiation. In
particular, small scale inhomogeneities cannot mimic the effects of
``dark energy.'' (2) The deviation, $h_{ab}$, of the metric from a
FLRW model can be broken up into ``long wavelength'' and ``short
wavelength'' parts. The long wavelength part corresponds in our
universe to averaging $h_{ab}$ over a spatial scale of order $L \sim
100$~Mpc. As analyzed fully in section III and subsection IVC, it
satisfies a linear equation with an additional source term due to the
short wavelength part. As argued in subsection IVB, the short
wavelength part should be described by Newtonian gravity, taking into
account only the matter within a proper distance of order $L \sim 100$
Mpc of the point under consideration.

Thus, the analysis of this paper goes a long way towards justifying
many of the key assumptions made in cosmology. In particular, if the
matter in the universe is non-relativistic, it suggests that to
calculate structure formation in the universe on all scales---from the
Hubble radius or larger down to arbitrarily small scales---it should
be a good approximation to evolve the long wavelength part of the
deviation of the metric from a FLRW model by \eqref{gammaLfin}, to
calculate the short wavelength part by \eqref{eqn:hconstanttau2}, and
to calculate the motion of the matter by \eqref{eom} together with
\eqref{Ctot}--\eqref{C1}.

In addition, our analysis suggests how further improvements can be
made to get more accurate approximations. To improve upon the description
of the ``long wavelength'' part of the deviation of the metric from a 
FLRW model one could go to higher order perturbation theory. In particular,
the second order correction to $\gamma^{(L)}_{ab}$ would be defined in our
framework by
\begin{equation}
\gamma^{(2)}_{ab} = \wlim_{\lambda \to 0}  
\frac{g_{ab}(\lambda) - g^{(0)}_{ab} - \lambda \gamma^{(L)}_{ab}}{\lambda^2} \, ,
\end{equation}
assuming, of course, that this weak limit exists. The equations
satisfied by $\gamma^{(2)}_{ab}$---taking full account of the small
scale inhomogeneities---could be derived by the methods used in
section III.  The description of the short wavelength part of the
deviation of the metric from a FLRW model could be improved by a more
accurate treatment of the integral relation
\eqref{eqn:hS-wavemap-integral2}.  However, the investigation of such
improvements is beyond the scope of this paper.

\bigskip

This research was supported in part by NSF Grant No.~PHY08-54807 to the
University of Chicago and by an NSERC PGS D scholarship to S.G.

\bibliography{mybib.bib}

%merlin.mbs apsrev4-1.bst 2010-07-25 4.21a (PWD, AO, DPC) hacked
%Control: key (0)
%Control: author (8) initials jnrlst
%Control: editor formatted (1) identically to author
%Control: production of article title (-1) disabled
%Control: page (0) single
%Control: year (1) truncated
%Control: production of eprint (0) enabled
\begin{thebibliography}{37}%
\makeatletter
\providecommand \@ifxundefined [1]{%
 \@ifx{#1\undefined}
}%
\providecommand \@ifnum [1]{%
 \ifnum #1\expandafter \@firstoftwo
 \else \expandafter \@secondoftwo
 \fi
}%
\providecommand \@ifx [1]{%
 \ifx #1\expandafter \@firstoftwo
 \else \expandafter \@secondoftwo
 \fi
}%
\providecommand \natexlab [1]{#1}%
\providecommand \enquote  [1]{``#1''}%
\providecommand \bibnamefont  [1]{#1}%
\providecommand \bibfnamefont [1]{#1}%
\providecommand \citenamefont [1]{#1}%
\providecommand \href@noop [0]{\@secondoftwo}%
\providecommand \href [0]{\begingroup \@sanitize@url \@href}%
\providecommand \@href[1]{\@@startlink{#1}\@@href}%
\providecommand \@@href[1]{\endgroup#1\@@endlink}%
\providecommand \@sanitize@url [0]{\catcode `\\12\catcode `\$12\catcode
  `\&12\catcode `\#12\catcode `\^12\catcode `\_12\catcode `\%12\relax}%
\providecommand \@@startlink[1]{}%
\providecommand \@@endlink[0]{}%
\providecommand \url  [0]{\begingroup\@sanitize@url \@url }%
\providecommand \@url [1]{\endgroup\@href {#1}{\urlprefix }}%
\providecommand \urlprefix  [0]{URL }%
\providecommand \Eprint [0]{\href }%
\providecommand \doibase [0]{http://dx.doi.org/}%
\providecommand \selectlanguage [0]{\@gobble}%
\providecommand \bibinfo  [0]{\@secondoftwo}%
\providecommand \bibfield  [0]{\@secondoftwo}%
\providecommand \translation [1]{[#1]}%
\providecommand \BibitemOpen [0]{}%
\providecommand \bibitemStop [0]{}%
\providecommand \bibitemNoStop [0]{.\EOS\space}%
\providecommand \EOS [0]{\spacefactor3000\relax}%
\providecommand \BibitemShut  [1]{\csname bibitem#1\endcsname}%
\let\auto@bib@innerbib\@empty
%</preamble>
\bibitem [{\citenamefont {Holz}\ and\ \citenamefont
  {Wald}(1998)}]{Holz:1997ic}%
  \BibitemOpen
  \bibfield  {author} {\bibinfo {author} {\bibfnamefont {D.~E.}\ \bibnamefont
  {Holz}}\ and\ \bibinfo {author} {\bibfnamefont {R.~M.}\ \bibnamefont
  {Wald}},\ }\href {\doibase 10.1103/PhysRevD.58.063501} {\bibfield  {journal}
  {\bibinfo  {journal} {Phys. Rev.}\ }\textbf {\bibinfo {volume} {D58}},\
  \bibinfo {pages} {063501} (\bibinfo {year} {1998})},\ \Eprint
  {http://arxiv.org/abs/astro-ph/9708036} {arXiv:astro-ph/9708036} \BibitemShut
  {NoStop}%
%%CITATION = ASTRO-PH/9708036;%%
\bibitem [{\citenamefont {Ishibashi}\ and\ \citenamefont
  {Wald}(2006)}]{Ishibashi:2005sj}%
  \BibitemOpen
  \bibfield  {author} {\bibinfo {author} {\bibfnamefont {A.}~\bibnamefont
  {Ishibashi}}\ and\ \bibinfo {author} {\bibfnamefont {R.~M.}\ \bibnamefont
  {Wald}},\ }\href {\doibase 10.1088/0264-9381/23/1/012} {\bibfield  {journal}
  {\bibinfo  {journal} {Class. Quant. Grav.}\ }\textbf {\bibinfo {volume}
  {23}},\ \bibinfo {pages} {235} (\bibinfo {year} {2006})},\ \Eprint
  {http://arxiv.org/abs/gr-qc/0509108} {arXiv:gr-qc/0509108} \BibitemShut
  {NoStop}%
%%CITATION = GR-QC/0509108;%%
\bibitem [{\citenamefont {van Elst}\ and\ \citenamefont
  {Ellis}(1998)}]{vanElst:1998kb}%
  \BibitemOpen
  \bibfield  {author} {\bibinfo {author} {\bibfnamefont {H.}~\bibnamefont {van
  Elst}}\ and\ \bibinfo {author} {\bibfnamefont {G.~F.~R.}\ \bibnamefont
  {Ellis}},\ }\href {\doibase 10.1088/0264-9381/15/11/017} {\bibfield
  {journal} {\bibinfo  {journal} {Class. Quant. Grav.}\ }\textbf {\bibinfo
  {volume} {15}},\ \bibinfo {pages} {3545} (\bibinfo {year} {1998})},\ \Eprint
  {http://arxiv.org/abs/gr-qc/9805087} {arXiv:gr-qc/9805087} \BibitemShut
  {NoStop}%
%%CITATION = GR-QC/9805087;%%
\bibitem [{\citenamefont {Rasanen}(2010)}]{Rasanen:2010wz}%
  \BibitemOpen
  \bibfield  {author} {\bibinfo {author} {\bibfnamefont {S.}~\bibnamefont
  {Rasanen}},\ }\href {\doibase 10.1103/PhysRevD.81.103512} {\bibfield
  {journal} {\bibinfo  {journal} {Phys. Rev.}\ }\textbf {\bibinfo {volume}
  {D81}},\ \bibinfo {pages} {103512} (\bibinfo {year} {2010})},\ \Eprint
  {http://arxiv.org/abs/1002.4779} {arXiv:1002.4779 [astro-ph.CO]} \BibitemShut
  {NoStop}%
%%CITATION = 1002.4779;%%
\bibitem [{\citenamefont {Buchert}\ and\ \citenamefont
  {Carfora}(2003)}]{Buchert:2002ij}%
  \BibitemOpen
  \bibfield  {author} {\bibinfo {author} {\bibfnamefont {T.}~\bibnamefont
  {Buchert}}\ and\ \bibinfo {author} {\bibfnamefont {M.}~\bibnamefont
  {Carfora}},\ }\href {\doibase 10.1103/PhysRevLett.90.031101} {\bibfield
  {journal} {\bibinfo  {journal} {Phys. Rev. Lett.}\ }\textbf {\bibinfo
  {volume} {90}},\ \bibinfo {pages} {031101} (\bibinfo {year} {2003})},\
  \Eprint {http://arxiv.org/abs/gr-qc/0210045} {arXiv:gr-qc/0210045}
  \BibitemShut {NoStop}%
%%CITATION = GR-QC/0210045;%%
\bibitem [{\citenamefont {Coley}\ \emph {et~al.}(2005)\citenamefont {Coley},
  \citenamefont {Pelavas},\ and\ \citenamefont {Zalaletdinov}}]{Coley:2005ei}%
  \BibitemOpen
  \bibfield  {author} {\bibinfo {author} {\bibfnamefont {A.~A.}\ \bibnamefont
  {Coley}}, \bibinfo {author} {\bibfnamefont {N.}~\bibnamefont {Pelavas}}, \
  and\ \bibinfo {author} {\bibfnamefont {R.~M.}\ \bibnamefont {Zalaletdinov}},\
  }\href {\doibase 10.1103/PhysRevLett.95.151102} {\bibfield  {journal}
  {\bibinfo  {journal} {Phys. Rev. Lett.}\ }\textbf {\bibinfo {volume} {95}},\
  \bibinfo {pages} {151102} (\bibinfo {year} {2005})},\ \Eprint
  {http://arxiv.org/abs/gr-qc/0504115} {arXiv:gr-qc/0504115} \BibitemShut
  {NoStop}%
%%CITATION = GR-QC/0504115;%%
\bibitem [{\citenamefont {Buchert}(2005)}]{Buchert:2005xf}%
  \BibitemOpen
  \bibfield  {author} {\bibinfo {author} {\bibfnamefont {T.}~\bibnamefont
  {Buchert}},\ }\href {\doibase 10.1088/0264-9381/22/19/L01} {\bibfield
  {journal} {\bibinfo  {journal} {Class. Quant. Grav.}\ }\textbf {\bibinfo
  {volume} {22}},\ \bibinfo {pages} {L113} (\bibinfo {year} {2005})},\ \Eprint
  {http://arxiv.org/abs/gr-qc/0507028} {arXiv:gr-qc/0507028} \BibitemShut
  {NoStop}%
%%CITATION = GR-QC/0507028;%%
\bibitem [{\citenamefont {Ellis}\ and\ \citenamefont
  {Buchert}(2005)}]{Ellis:2005uz}%
  \BibitemOpen
  \bibfield  {author} {\bibinfo {author} {\bibfnamefont {G.~F.~R.}\
  \bibnamefont {Ellis}}\ and\ \bibinfo {author} {\bibfnamefont
  {T.}~\bibnamefont {Buchert}},\ }\href {\doibase
  10.1016/j.physleta.2005.06.087} {\bibfield  {journal} {\bibinfo  {journal}
  {Phys. Lett.}\ }\textbf {\bibinfo {volume} {A347}},\ \bibinfo {pages} {38}
  (\bibinfo {year} {2005})},\ \Eprint {http://arxiv.org/abs/gr-qc/0506106}
  {arXiv:gr-qc/0506106} \BibitemShut {NoStop}%
%%CITATION = GR-QC/0506106;%%
\bibitem [{\citenamefont {Nambu}\ and\ \citenamefont
  {Tanimoto}(2005)}]{Nambu:2005zn}%
  \BibitemOpen
  \bibfield  {author} {\bibinfo {author} {\bibfnamefont {Y.}~\bibnamefont
  {Nambu}}\ and\ \bibinfo {author} {\bibfnamefont {M.}~\bibnamefont
  {Tanimoto}},\ }\href@noop {} {\  (\bibinfo {year} {2005})},\ \Eprint
  {http://arxiv.org/abs/gr-qc/0507057} {arXiv:gr-qc/0507057} \BibitemShut
  {NoStop}%
%%CITATION = GR-QC/0507057;%%
\bibitem [{\citenamefont {Kolb}\ \emph {et~al.}(2006)\citenamefont {Kolb},
  \citenamefont {Matarrese},\ and\ \citenamefont {Riotto}}]{Kolb:2005da}%
  \BibitemOpen
  \bibfield  {author} {\bibinfo {author} {\bibfnamefont {E.~W.}\ \bibnamefont
  {Kolb}}, \bibinfo {author} {\bibfnamefont {S.}~\bibnamefont {Matarrese}}, \
  and\ \bibinfo {author} {\bibfnamefont {A.}~\bibnamefont {Riotto}},\ }\href
  {\doibase 10.1088/1367-2630/8/12/322} {\bibfield  {journal} {\bibinfo
  {journal} {New J. Phys.}\ }\textbf {\bibinfo {volume} {8}},\ \bibinfo {pages}
  {322} (\bibinfo {year} {2006})},\ \Eprint
  {http://arxiv.org/abs/astro-ph/0506534} {arXiv:astro-ph/0506534} \BibitemShut
  {NoStop}%
%%CITATION = ASTRO-PH/0506534;%%
\bibitem [{\citenamefont {Rasanen}(2006)}]{Rasanen:2006kp}%
  \BibitemOpen
  \bibfield  {author} {\bibinfo {author} {\bibfnamefont {S.}~\bibnamefont
  {Rasanen}},\ }\href {\doibase 10.1088/1475-7516/2006/11/003} {\bibfield
  {journal} {\bibinfo  {journal} {JCAP}\ }\textbf {\bibinfo {volume} {0611}},\
  \bibinfo {pages} {003} (\bibinfo {year} {2006})},\ \Eprint
  {http://arxiv.org/abs/astro-ph/0607626} {arXiv:astro-ph/0607626} \BibitemShut
  {NoStop}%
%%CITATION = ASTRO-PH/0607626;%%
\bibitem [{\citenamefont {Buchert}(2008)}]{Buchert:2007ik}%
  \BibitemOpen
  \bibfield  {author} {\bibinfo {author} {\bibfnamefont {T.}~\bibnamefont
  {Buchert}},\ }\href {\doibase 10.1007/s10714-007-0554-8} {\bibfield
  {journal} {\bibinfo  {journal} {Gen. Rel. Grav.}\ }\textbf {\bibinfo {volume}
  {40}},\ \bibinfo {pages} {467} (\bibinfo {year} {2008})},\ \Eprint
  {http://arxiv.org/abs/0707.2153} {arXiv:0707.2153 [gr-qc]} \BibitemShut
  {NoStop}%
%%CITATION = 0707.2153;%%
\bibitem [{\citenamefont {Kolb}\ \emph {et~al.}(2008)\citenamefont {Kolb},
  \citenamefont {Marra},\ and\ \citenamefont {Matarrese}}]{Kolb:2008bn}%
  \BibitemOpen
  \bibfield  {author} {\bibinfo {author} {\bibfnamefont {E.~W.}\ \bibnamefont
  {Kolb}}, \bibinfo {author} {\bibfnamefont {V.}~\bibnamefont {Marra}}, \ and\
  \bibinfo {author} {\bibfnamefont {S.}~\bibnamefont {Matarrese}},\ }\href
  {\doibase 10.1103/PhysRevD.78.103002} {\bibfield  {journal} {\bibinfo
  {journal} {Phys. Rev.}\ }\textbf {\bibinfo {volume} {D78}},\ \bibinfo {pages}
  {103002} (\bibinfo {year} {2008})},\ \Eprint {http://arxiv.org/abs/0807.0401}
  {arXiv:0807.0401 [astro-ph]} \BibitemShut {NoStop}%
%%CITATION = 0807.0401;%%
\bibitem [{\citenamefont {Rasanen}(2008)}]{Rasanen:2008it}%
  \BibitemOpen
  \bibfield  {author} {\bibinfo {author} {\bibfnamefont {S.}~\bibnamefont
  {Rasanen}},\ }\href {\doibase 10.1088/1475-7516/2008/04/026} {\bibfield
  {journal} {\bibinfo  {journal} {JCAP}\ }\textbf {\bibinfo {volume} {0804}},\
  \bibinfo {pages} {026} (\bibinfo {year} {2008})},\ \Eprint
  {http://arxiv.org/abs/0801.2692} {arXiv:0801.2692 [astro-ph]} \BibitemShut
  {NoStop}%
%%CITATION = 0801.2692;%%
\bibitem [{\citenamefont {Zalaletdinov}(2008)}]{Zalaletdinov:2008ts}%
  \BibitemOpen
  \bibfield  {author} {\bibinfo {author} {\bibfnamefont {R.}~\bibnamefont
  {Zalaletdinov}},\ }\href {\doibase 10.1142/S0217751X08040032} {\bibfield
  {journal} {\bibinfo  {journal} {Int. J. Mod. Phys.}\ }\textbf {\bibinfo
  {volume} {A23}},\ \bibinfo {pages} {1173} (\bibinfo {year} {2008})},\ \Eprint
  {http://arxiv.org/abs/0801.3256} {arXiv:0801.3256 [gr-qc]} \BibitemShut
  {NoStop}%
%%CITATION = 0801.3256;%%
\bibitem [{\citenamefont {Kolb}\ \emph {et~al.}(2010)\citenamefont {Kolb},
  \citenamefont {Marra},\ and\ \citenamefont {Matarrese}}]{Kolb:2009rp}%
  \BibitemOpen
  \bibfield  {author} {\bibinfo {author} {\bibfnamefont {E.~W.}\ \bibnamefont
  {Kolb}}, \bibinfo {author} {\bibfnamefont {V.}~\bibnamefont {Marra}}, \ and\
  \bibinfo {author} {\bibfnamefont {S.}~\bibnamefont {Matarrese}},\ }\href
  {\doibase 10.1007/s10714-009-0913-8} {\bibfield  {journal} {\bibinfo
  {journal} {Gen. Rel. Grav.}\ }\textbf {\bibinfo {volume} {42}},\ \bibinfo
  {pages} {1399} (\bibinfo {year} {2010})},\ \Eprint
  {http://arxiv.org/abs/0901.4566} {arXiv:0901.4566 [astro-ph.CO]} \BibitemShut
  {NoStop}%
%%CITATION = 0901.4566;%%
\bibitem [{\citenamefont {Wiegand}\ and\ \citenamefont
  {Buchert}(2010)}]{Wiegand:2010uh}%
  \BibitemOpen
  \bibfield  {author} {\bibinfo {author} {\bibfnamefont {A.}~\bibnamefont
  {Wiegand}}\ and\ \bibinfo {author} {\bibfnamefont {T.}~\bibnamefont
  {Buchert}},\ }\href {\doibase 10.1103/PhysRevD.82.023523} {\bibfield
  {journal} {\bibinfo  {journal} {Phys. Rev.}\ }\textbf {\bibinfo {volume}
  {D82}},\ \bibinfo {pages} {023523} (\bibinfo {year} {2010})},\ \Eprint
  {http://arxiv.org/abs/1002.3912} {arXiv:1002.3912 [astro-ph.CO]} \BibitemShut
  {NoStop}%
%%CITATION = 1002.3912;%%
\bibitem [{\citenamefont {Buchert}(2000)}]{Buchert:1999er}%
  \BibitemOpen
  \bibfield  {author} {\bibinfo {author} {\bibfnamefont {T.}~\bibnamefont
  {Buchert}},\ }\href {\doibase 10.1023/A:1001800617177} {\bibfield  {journal}
  {\bibinfo  {journal} {Gen. Rel. Grav.}\ }\textbf {\bibinfo {volume} {32}},\
  \bibinfo {pages} {105} (\bibinfo {year} {2000})},\ \Eprint
  {http://arxiv.org/abs/gr-qc/9906015} {arXiv:gr-qc/9906015} \BibitemShut
  {NoStop}%
%%CITATION = GR-QC/9906015;%%
\bibitem [{\citenamefont {Buchert}(2001)}]{Buchert:2001sa}%
  \BibitemOpen
  \bibfield  {author} {\bibinfo {author} {\bibfnamefont {T.}~\bibnamefont
  {Buchert}},\ }\href {\doibase 10.1023/A:1012061725841} {\bibfield  {journal}
  {\bibinfo  {journal} {Gen. Rel. Grav.}\ }\textbf {\bibinfo {volume} {33}},\
  \bibinfo {pages} {1381} (\bibinfo {year} {2001})},\ \Eprint
  {http://arxiv.org/abs/gr-qc/0102049} {arXiv:gr-qc/0102049} \BibitemShut
  {NoStop}%
%%CITATION = GR-QC/0102049;%%
\bibitem [{\citenamefont {Buchert}\ and\ \citenamefont
  {Obadia}(2010)}]{Buchert:2010ug}%
  \BibitemOpen
  \bibfield  {author} {\bibinfo {author} {\bibfnamefont {T.}~\bibnamefont
  {Buchert}}\ and\ \bibinfo {author} {\bibfnamefont {N.}~\bibnamefont
  {Obadia}},\ }\href@noop {} {\  (\bibinfo {year} {2010})},\ \Eprint
  {http://arxiv.org/abs/1010.4512} {arXiv:1010.4512 [gr-qc]} \BibitemShut
  {NoStop}%
%%CITATION = 1010.4512;%%
\bibitem [{\citenamefont {Futamase}(1996)}]{Futamase:1996fk}%
  \BibitemOpen
  \bibfield  {author} {\bibinfo {author} {\bibfnamefont {T.}~\bibnamefont
  {Futamase}},\ }\href {\doibase 10.1103/PhysRevD.53.681} {\bibfield  {journal}
  {\bibinfo  {journal} {Phys. Rev.}\ }\textbf {\bibinfo {volume} {D53}},\
  \bibinfo {pages} {681} (\bibinfo {year} {1996})}\BibitemShut {NoStop}%
%%CITATION = PHRVA,D53,681;%%
\bibitem [{\citenamefont {Zalaletdinov}(1997)}]{Zalaletdinov:1996aj}%
  \BibitemOpen
  \bibfield  {author} {\bibinfo {author} {\bibfnamefont {R.~M.}\ \bibnamefont
  {Zalaletdinov}},\ }\href@noop {} {\bibfield  {journal} {\bibinfo  {journal}
  {Bull. Astron. Soc. India}\ }\textbf {\bibinfo {volume} {25}},\ \bibinfo
  {pages} {401} (\bibinfo {year} {1997})},\ \Eprint
  {http://arxiv.org/abs/gr-qc/9703016} {arXiv:gr-qc/9703016} \BibitemShut
  {NoStop}%
%%CITATION = GR-QC/9703016;%%
\bibitem [{\citenamefont {Isaacson}(1968{\natexlab{a}})}]{Isaacson:1967zz}%
  \BibitemOpen
  \bibfield  {author} {\bibinfo {author} {\bibfnamefont {R.~A.}\ \bibnamefont
  {Isaacson}},\ }\href {\doibase 10.1103/PhysRev.166.1263} {\bibfield
  {journal} {\bibinfo  {journal} {Phys. Rev.}\ }\textbf {\bibinfo {volume}
  {166}},\ \bibinfo {pages} {1263} (\bibinfo {year}
  {1968}{\natexlab{a}})}\BibitemShut {NoStop}%
%%CITATION = PHRVA,166,1263;%%
\bibitem [{\citenamefont {Isaacson}(1968{\natexlab{b}})}]{Isaacson:1968zza}%
  \BibitemOpen
  \bibfield  {author} {\bibinfo {author} {\bibfnamefont {R.~A.}\ \bibnamefont
  {Isaacson}},\ }\href {\doibase 10.1103/PhysRev1.66.1272} {\bibfield
  {journal} {\bibinfo  {journal} {Phys. Rev.}\ }\textbf {\bibinfo {volume}
  {166}},\ \bibinfo {pages} {1272} (\bibinfo {year}
  {1968}{\natexlab{b}})}\BibitemShut {NoStop}%
%%CITATION = PHRVA,166,1272;%%
\bibitem [{\citenamefont {Misner}\ \emph {et~al.}(1973)\citenamefont {Misner},
  \citenamefont {Thorne},\ and\ \citenamefont {Wheeler}}]{MTW}%
  \BibitemOpen
  \bibfield  {author} {\bibinfo {author} {\bibfnamefont {C.~W.}\ \bibnamefont
  {Misner}}, \bibinfo {author} {\bibfnamefont {K.~S.}\ \bibnamefont {Thorne}},
  \ and\ \bibinfo {author} {\bibfnamefont {J.~A.}\ \bibnamefont {Wheeler}},\
  }\href@noop {} {\emph {\bibinfo {title} {Gravitation}}}\ (\bibinfo
  {publisher} {W. H. Freeman and Company},\ \bibinfo {address} {San
  Francisco},\ \bibinfo {year} {1973})\BibitemShut {NoStop}%
\bibitem [{\citenamefont {Burnett}(1989)}]{Burnett:1989gp}%
  \BibitemOpen
  \bibfield  {author} {\bibinfo {author} {\bibfnamefont {G.~A.}\ \bibnamefont
  {Burnett}},\ }\href {\doibase 10.1063/1.528594} {\bibfield  {journal}
  {\bibinfo  {journal} {J. Math. Phys.}\ }\textbf {\bibinfo {volume} {30}},\
  \bibinfo {pages} {90} (\bibinfo {year} {1989})}\BibitemShut {NoStop}%
%%CITATION = JMAPA,30,90;%%
\bibitem [{\citenamefont {Noonan}(1984)}]{Noonan:1984}%
  \BibitemOpen
  \bibfield  {author} {\bibinfo {author} {\bibfnamefont {T.}~\bibnamefont
  {Noonan}},\ }\href {http://dx.doi.org/10.1007/BF00760648} {\bibfield
  {journal} {\bibinfo  {journal} {Gen. Rel. Grav.}\ }\textbf {\bibinfo {volume}
  {16}},\ \bibinfo {pages} {1103} (\bibinfo {year} {1984})}\BibitemShut
  {NoStop}%
\bibitem [{\citenamefont {Baumann}\ \emph {et~al.}(2010)\citenamefont
  {Baumann}, \citenamefont {Nicolis}, \citenamefont {Senatore},\ and\
  \citenamefont {Zaldarriaga}}]{Baumann:2010tm}%
  \BibitemOpen
  \bibfield  {author} {\bibinfo {author} {\bibfnamefont {D.}~\bibnamefont
  {Baumann}}, \bibinfo {author} {\bibfnamefont {A.}~\bibnamefont {Nicolis}},
  \bibinfo {author} {\bibfnamefont {L.}~\bibnamefont {Senatore}}, \ and\
  \bibinfo {author} {\bibfnamefont {M.}~\bibnamefont {Zaldarriaga}},\
  }\href@noop {} {\  (\bibinfo {year} {2010})},\ \Eprint
  {http://arxiv.org/abs/1004.2488} {arXiv:1004.2488 [astro-ph.CO]} \BibitemShut
  {NoStop}%
%%CITATION = 1004.2488;%%
\bibitem [{\citenamefont {Wald}(1984)}]{Wald:1984}%
  \BibitemOpen
  \bibfield  {author} {\bibinfo {author} {\bibfnamefont {R.~M.}\ \bibnamefont
  {Wald}},\ }\href@noop {} {\emph {\bibinfo {title} {General Relativity}}}\
  (\bibinfo  {publisher} {University of Chicago Press},\ \bibinfo {address}
  {Chicago, IL},\ \bibinfo {year} {1984})\BibitemShut {NoStop}%
\bibitem [{\citenamefont {Mart{\'\i}n-Garc{\'\i}a}(2008)}]{martin-garcia:2008}%
  \BibitemOpen
  \bibfield  {author} {\bibinfo {author} {\bibfnamefont {J.~M.}\ \bibnamefont
  {Mart{\'\i}n-Garc{\'\i}a}},\ }\href {\doibase DOI:10.1016/j.cpc.2008.05.009}
  {\bibfield  {journal} {\bibinfo  {journal} {Computer Physics Communications}\
  }\textbf {\bibinfo {volume} {179}},\ \bibinfo {pages} {597 } (\bibinfo {year}
  {2008})},\ \Eprint {http://arxiv.org/abs/0803.0862} {arXiv:0803.0862 [cs.SC]}
  \BibitemShut {NoStop}%
\bibitem [{\citenamefont {Mart{\'\i}n-Garc{\'\i}a}(2010)}]{martin-garcia:web}%
  \BibitemOpen
  \bibfield  {author} {\bibinfo {author} {\bibfnamefont {J.~M.}\ \bibnamefont
  {Mart{\'\i}n-Garc{\'\i}a}},\ }\href@noop {} {\enquote {\bibinfo {title}
  {x{A}ct: Efficient tensor computer algebra for {M}athematica},}\ }\bibinfo
  {howpublished} {\url{http://www.xact.es}} (\bibinfo {year}
  {2010})\BibitemShut {NoStop}%
\bibitem [{\citenamefont {Futamase}\ and\ \citenamefont
  {Itoh}(2007)}]{lrr-2007-2}%
  \BibitemOpen
  \bibfield  {author} {\bibinfo {author} {\bibfnamefont {T.}~\bibnamefont
  {Futamase}}\ and\ \bibinfo {author} {\bibfnamefont {Y.}~\bibnamefont
  {Itoh}},\ }\href@noop {} {\bibfield  {journal} {\bibinfo  {journal} {Living
  Reviews in Relativity}\ }\textbf {\bibinfo {volume} {10}},\ \bibinfo {pages}
  {2} (\bibinfo {year} {2007})}\BibitemShut {NoStop}%
\bibitem [{\citenamefont {Choquet-Bruhat}(2009)}]{Choquet-Bruhat:2009}%
  \BibitemOpen
  \bibfield  {author} {\bibinfo {author} {\bibfnamefont {Y.}~\bibnamefont
  {Choquet-Bruhat}},\ }\href@noop {} {\emph {\bibinfo {title} {General
  Relativity and the Einstein Equations}}}\ (\bibinfo  {publisher} {Oxford
  University Press},\ \bibinfo {address} {New York},\ \bibinfo {year}
  {2009})\BibitemShut {NoStop}%
\bibitem [{\citenamefont {Futamase}\ and\ \citenamefont
  {Schutz}(1983)}]{Futamase:1983}%
  \BibitemOpen
  \bibfield  {author} {\bibinfo {author} {\bibfnamefont {T.}~\bibnamefont
  {Futamase}}\ and\ \bibinfo {author} {\bibfnamefont {B.~F.}\ \bibnamefont
  {Schutz}},\ }\href {\doibase 10.1103/PhysRevD.28.2363} {\bibfield  {journal}
  {\bibinfo  {journal} {Phys. Rev.}\ }\textbf {\bibinfo {volume} {D28}},\
  \bibinfo {pages} {2363} (\bibinfo {year} {1983})}\BibitemShut {NoStop}%
\bibitem [{\citenamefont {Oliynyk}(2010{\natexlab{a}})}]{Oliynyk:2009qa}%
  \BibitemOpen
  \bibfield  {author} {\bibinfo {author} {\bibfnamefont {T.~A.}\ \bibnamefont
  {Oliynyk}},\ }\href {\doibase 10.1007/s00220-009-0931-0} {\bibfield
  {journal} {\bibinfo  {journal} {Commun. Math. Phys.}\ }\textbf {\bibinfo
  {volume} {295}},\ \bibinfo {pages} {431} (\bibinfo {year}
  {2010}{\natexlab{a}})},\ \Eprint {http://arxiv.org/abs/0908.2836}
  {arXiv:0908.2836 [gr-qc]} \BibitemShut {NoStop}%
%%CITATION = 0908.2836;%%
\bibitem [{\citenamefont {Oliynyk}(2010{\natexlab{b}})}]{Oliynyk:2009bc}%
  \BibitemOpen
  \bibfield  {author} {\bibinfo {author} {\bibfnamefont {T.~A.}\ \bibnamefont
  {Oliynyk}},\ }\href@noop {} {\bibfield  {journal} {\bibinfo  {journal} {J.
  Hyperbol. Diff. Equat.}\ }\textbf {\bibinfo {volume} {7}},\ \bibinfo {pages}
  {405} (\bibinfo {year} {2010}{\natexlab{b}})},\ \Eprint
  {http://arxiv.org/abs/0908.4454} {arXiv:0908.4454 [gr-qc]} \BibitemShut
  {NoStop}%
%%CITATION = 0908.4454;%%
\bibitem [{\citenamefont {Poisson}(2004)}]{lrr-2004-6}%
  \BibitemOpen
  \bibfield  {author} {\bibinfo {author} {\bibfnamefont {E.}~\bibnamefont
  {Poisson}},\ }\href@noop {} {\bibfield  {journal} {\bibinfo  {journal}
  {Living Reviews in Relativity}\ }\textbf {\bibinfo {volume} {7}},\ \bibinfo
  {pages} {6} (\bibinfo {year} {2004})}\BibitemShut {NoStop}%
\end{thebibliography}%

\end{document}